\documentclass{article}%
\usepackage{amsfonts}
\usepackage{amsmath}
\usepackage{color}
\usepackage{isomath}
\usepackage{upgreek}
\usepackage{fancybox}
\usepackage{appendix}%
\usepackage{amssymb}%
\usepackage[dvips]{graphicx}
\usepackage{graphicx}
\usepackage{psfrag}
\usepackage{subcaption}
\usepackage{mathtools}
\usepackage{tikzpagenodes}
\usetikzlibrary{tikzmark}
\usepackage{tikz}
\usetikzlibrary{shapes,snakes}
\usepackage{amsmath,amssymb}
\usepackage{verbatim}
\usepackage{gensymb}

\providecommand{\U}[1]{\protect\rule{.1in}{.1in}}
\newtheorem{theorem}{Theorem}

\newtheorem{corollary}[theorem]{Corollary}

\newtheorem{lemma}[theorem]{Lemma}

\newtheorem{proposition}[theorem]{Proposition}
\newtheorem{remark}[theorem]{Remark}

\newenvironment{proof}[1][Proof]{\noindent\textbf{#1.} }{\ \rule{0.5em}{0.5em}}
\begin{document}

\title{\textbf{Differential curvature invariants and event horizon detection for accelerating Kerr-Newman black holes in (anti-)de Sitter spacetime}}
\author{G. V. Kraniotis \footnote{email: \textcolor{blue}{gvkraniotis@gmail.com}}\\
University of Ioannina, Physics Department \\ Section of
Theoretical Physics, GR- 451 10, Greece \\
}

 \maketitle
 \begin{abstract}
We compute analytically differential invariants for accelerating, rotating and charged black holes with a cosmological constant $\Lambda$. In particular, we compute in closed form novel explicit algebraic expressions for curvature invariants constructed from covariant derivatives of the Riemann and Weyl tensors, such as the Karlhede and the Abdelqader-Lake invariants,  for the Kerr-Newman-(anti-)de Sitter and accelerating Kerr-Newman-(anti-)de Sitter black hole spacetimes. We explicitly show that some of the computed curvature invariants are vanishing at the event and Cauchy horizons or the ergosurface of the accelerating, charged and  rotating black holes with a non-zero cosmological constant. Using a particular generalised null-tetrad and the Bianchi identities we compute in the Newman-Penrose formalism in closed-analytic form the Page-Shoom curvature invariant for the accelerating Kerr-Newman black hole in (anti-)de Sitter spacetime and prove that is vanishing at the black hole event and Cauchy horizon radii.  Therefore such invariants can serve as possible detectors of the event horizon and ergosurface for such black hole metrics which belong to the most general type D solution of the Einstein-Maxwell equations with a cosmological constant. Also the norms associated with the gradients of the first two Weyl invariants in the Zakhary-McIntosh classification, were studied in detail.
Although both locally single out the horizons, their global behaviour is also intriguing. Both reflect the background angular momentum and electric charge as the volume of space allowing a timelike gradient decreases with increasing angular momentum and charge.

\end{abstract}

\section{Introduction}
It is well known that in a semi-Riemannian manifold there are three causal types of submanifolds: spacelike (Riemannian), timelike (Lorentzian) and lightlike (degenerate), depending on the character of the induced metric on the tangent space \cite{Duggal},\cite{Katsuno}.
Coordinate singularities can sometimes be interpreted as various kinds of \textit{horizons}. Lightlike submanifolds (in particular, lightlike hypersurfaces) are interesting in general relativity since they produce models of different types of horizons .  These include the:  \textit{event horizon}, \textit{Killing horizons},\textit{Cauchy horizons},\textit{Cosmological} and \textit{acceleration horizons}.
The event horizon of a black hole is a codimension one null hypersurface, which constitutes the boundary of the black hole region from which causal geodesics cannot reach future null infinity. As a consequence,
the event horizon is highly nonlocal  and \textit{a priori} we need the full knowledge of spacetime to locate it \footnote{In this respect, the event horizon  constitutes essentially a global (\textit{teleological}) object \cite{KrishamAshtekar}. }. In relation to this, the images obtained recently by the Event Horizon Telescope (EHT) and analysed by computer simulations \cite{EHT}, contain the environment of the black hole as well as the codimension two cross section of the event horizon but not the event horizon which is a 2+1 dimensional hypersurface.
Thus the analytical study of black hole horizon, and the localisation of it becomes an issue of crucial importance.

The Killing horizons, are surfaces where the norm of some Killing vector vanishes. On the other hand, Cauchy horizons are future boundaries of regions that can be uniquely determined by initial data on some appropriate spacelike hypersurface.

The curvature scalar invariants of the Riemann tensor are important in General Relativity because they allow a manifestly coordinate invariant characterisation of certain geometrical properties of spacetimes such as, among others, curvature singularities, gravitomagnetism, anomalies \cite{Zahkary}-\cite{frolov}. Recently, we calculated  explicit analytic expressions for the set of Zakhary-McIntosh (ZM) curvature invariants for accelerating Kerr-Newman black holes in (anti-)de Sitter spacetime as well as for the Kerr-Newman-(anti-)de Sitter black hole \cite{KraniotisCurvature}. We also calculated in \cite{KraniotisCurvature}, explicit algebraic expressions  for the Euler-Poincare density invariant and the Kretschmann scalar  for both types of black hole spacetimes. We also highlighted that for accelerating rotating and charged black holes with $\Lambda\not =0$, the integrated Chern-Pontryagin-Hirzebruch invariant gives a non-zero result for the quantum photon chiral anomaly \cite{KraniotisCurvature}.

We also note that curvature invariants have been proposed as measures of gravitational entropy in an attempt to explain the arrow of time by the low entropy initial conditions. For instance,  the first Weyl curvature invariant in the ZM classification has been proposed as the entropy density for a 5-dimensional Schwarzschild black hole in \cite{nanli}. This is related to the Weyl curvature hypothesis by Penrose who argued that some scalar invariant of the Weyl tensor is a monotonically growing function of time and is thus somehow related to the gravitational entropy in the universe. Thus, the low entropy in the gravitational field is tied to constraints on the Weyl curvature \cite{RogerWeylCHyp}.

On the other hand differential curvature invariants,such as covariant derivatives of the curvature tensor \cite{KarlhedeA} and gradients of non-differential invariants \cite{Lakeein} are necessary for a complete description of the local geometry. They have been suggested as possible detectors of the event horizon and ergosurface of the Kerr black holes \cite{LakeZwei}.

A closely related matter is the equivalence problem \cite{karlhedeequiv}. A well-known theorem of E. Cartan (see \cite{Berger}, p.97) states that the Riemann curvature tensor and its covariant derivatives are a complete set of local invariants of the (analytic) Riemannian metrics. Therefore the knowledge of the Riemann tensor
$R^{\alpha}_{\;\beta\gamma\delta}$ and its covariant derivatives such as $D^1{R^{\alpha}_{\;\beta\gamma\delta}}\equiv R^{\alpha}_{\;\beta\gamma\delta;\epsilon_1},\cdots, D^s {R^{\alpha}_{\;\beta\gamma\delta}}\equiv R^{\alpha}_{\;\beta\gamma\delta;\epsilon_1\ldots, \epsilon_s}$, for all $s\geq0$, at a point $p$ determines  (up to local isometries) the germ of the metric at $p$.
Inspired by the work of E. Cartan, Karlhede developed an algorithm for addressing the problem of equivalence \cite{karlhedeequiv} \footnote{Let $R^s$ denotes the set $\{R^{\alpha}_{\;\beta\gamma\delta}, R^{\alpha}_{\;\beta\gamma\delta;\epsilon_1},\ldots, R^{\alpha}_{\;\beta\gamma\delta;\epsilon_1\ldots \epsilon_s}\}$. If $s$ is the lowest value for
which the elements of $R^{s+1}$ are functionally dependent as functions over the bundle of Lorentz frames $F(M)$ on
those in $R^s$, then according to Cartan the set $R^{s+1}$ gives a complete coordinate invariant description of
the local geometry. Two manifolds $(M,g)$ and $(\tilde{M},\tilde{g})$ are equivalent if and only if the sets $R^{s+1}$ and $\tilde{R}^{s+1}$ are equal.}. In addition, investigations on the isometry and isotropy groups of Riemannian spaces were reported in \cite{karlmacc}.

It is the purpose of this paper to apply the formalism  of Karlhede \textit{et al} \cite{KarlhedeA} and Abdelqader-Lake \cite{LakeZwei}  (see also \cite{PageD}), to the
case of accelerating and rotating charged black holes with non-zero cosmological constant $\Lambda$ and compute for the \textit{first time} analytic algebraic expressions for the corresponding curvature differential  invariants. Specifically, we compute novel closed-form algebraic expressions for these local curvature invariants for the
accelerating Kerr-Newman-(anti-)de Sitter black hole. Furthermore, we compute novel explicit algebraic expressions for the Karlhede  and Abdelqader-Lake  local curvature invariants also for the non-accelerating Kerr-Newman-(anti-)de Sitter (KN(a)dS) black hole.  These black hole metrics belong to the most general type D solution of the Einstein-Maxwell equations with a cosmological constant and constitute the physically most important case \cite{GrifPod},\cite{PlebanskiDemianski}.
Indeed, besides the intrinsic theoretical interest of accelerating or non-accelerating KN(a)dS black holes , a variety of observations support and single out their physical relevance in Nature.

An expansive range  of astronomical and cosmological observations in the last two decades, including high-redshift type Ia supernovae, cosmic microwave background radiation and large scale structure indicate convincingly an accelerating expansion of the Universe \cite{Supern},\cite{Jones},\cite{Aubourg},\cite{AbbottTMC}. Such observational data  can be explained by a positive cosmological constant $\Lambda$ ($\Lambda>0$) with a magnitude $\Lambda\sim 10^{-56}{\rm cm}^{-2}$ \cite{GVKSWB}.

 Recent observations of structures near the galactic centre region SgrA*  by the GRAVITY experiment, indicate possible presence of a small electric charge of central supermassive black hole \cite{Zajacek},\cite{Britzen}. Accretion disk physics around magnetised Kerr black holes under the influence of cosmic repulsion  is extensively discussed in the review \cite{universemdpi} \footnote{We also mention that supermassive black holes as possible sources of ultahigh-energy cosmic rays have been suggested in \cite{waldcharge}, where it has been shown that large values of the Lorentz $\gamma$ factor of an escaping ultrahigh-energy particle from the inner regions of the black hole accretion disk may occur only in the presence of the induced charge of the black hole.}.

Furthermore, observations of the galactic centre supermassive black hole indicate that it is rotating. With regard to the spin $a$ (Kerr rotation parameter) we note that observations of near-infrared periodic flares have revealed that the central black
hole $SgrA^{*}$ is rotating with a reported spin parameter: $a=0.52 (\pm 0.1,\pm 0.08,\pm 0,08)$ \cite{genzelr}. The error estimates here the uncertainties in the
period, black hole mass and distance to the galactic centre, respectively. Observation of X-ray flares confirmed that the spin of the supermassive black hole is
indeed substantial and values of the Kerr parameter as high as: $a=0.9939^{+0.0026}_{-0.0074}$ have been obtained \cite{aschenbach}. For our plots we choose values for the Kerr parameter consistent with these observations.
On the other hand, the Kerr parameter $a$ of SgrA* can be measured by precise observations of the predicted by theory Lense-Thirring and periastron precessions from the observed orbits of S-stars in the central arcsecond of Milky Way \cite{KraniotisStars1},\cite{KraniotisSstars2}.
Therefore, it is quite interesting to study the combined effect of the cosmological constant,rotation parameter and electromagnetic fields on the geometry of spacetime surrounding the black hole singularity through the explicit algebraic computation and plotting of the  Karlhede  and Abdelqader-Lake differential   curvature invariants,   taking into account also the acceleration parameter. In addition, determining the role of such local invariants as detectors of event and ergosurface horizons for the most general black hole solution of Einstein-Maxwell equations, is a very important step in black hole theory and phenomenology.

The material of this paper is organised as follows: In sections \ref{AbdelLimni} and
\ref{andersKarlhede1} we present the definitions of the  Abdelqader-Lake and Karlhede local scalar curvature  invariants that we shall use in computing explicit algebraic expressions for these differential invariants for the accelerating and non-accelerating Kerr-Newman black holes in (anti-)de Sitter spacetime.
In section \ref{karlhedeKNLambda} we derive a novel closed-form analytic expression for the norm of the covariant derivative of the Riemann tensor, i.e. the Karlhede invariant, for the Kerr-Newman-(anti-)de Sitter black hole, see Theorem \ref{unserekarlhede} and eqn.(\ref{KarlHKNdS}). In section \ref{KNlambdaAbdeLake} we derive explicit closed form expressions for the Abdelqader-Lake differential invariants for non-accelerating Kerr-Newman black holes with cosmological constant, see Theorems: Theorem \ref{synalParWeyl} (eqn.(\ref{normcovderweyl}))-Theorem \ref{mixedsynparweyl} (eqn.(\ref{covdermixedweyli4})), and Theorem \ref{difanalQ1}-Theorem \ref{kerrdeviation}.
Two of the Abdelqader-Lake invariants are norms associated with the gradients of the first two non-differential Weyl invariants in the ZM classification. We investigate in detail these two invariants for the Kerr-(anti-)de Sitter and Kerr-Newman-(anti-)de Sitter black holes. From the explicit \textit{novel} algebraic expressions we derive in this work, we find that both invariants locally single out the horizons. Even more so their global behaviour is very interesting. Both reflect the background angular momentum and electric charge as the volume of space allowing a timelike gradient decreases with increasing angular momentum and electric charge becoming zero for highly spinning and charged black holes. See Corollaries \ref{normdifinv1pi2}-\ref{norm2gradientflowI6eqax} and Figs.\ref{Regionsnegativenormkmu},\ref{RegionsnegativenormkmuKN},\ref{ContourPlotsgradI6curva},
\ref{RegionplotsofsignI6}.
In section \ref{epitaxydiaforikesanal}, we derive \textit{new} explicit analytic expressions for the Karlhede and Abdelqader-Lake differential invariants for accelerating Kerr-Newman black holes in the presence of the cosmological constant.
In Theorem \ref{KarlEpitaKerrdS}, Eqn.(\ref{marvelKarlhedeaccelKerrdS}), we derive for the first time an explicit algebraic expression for the Karlhede invariant for an accelerating Kerr black hole in the presence of the cosmological constant $\Lambda$.
Subsequently, we present novel explicit expressions for the Abdelqader-Lake curvature invariants for  accelerating Kerr black holes, see Theorems: Theorem \ref{weylcovI3epi} (eqn.(\ref{idreiepitkerr}))-Theorem \ref{sebenepitaxkerr}. In Theorems: Theorem \ref{Q2invAccelKerr}-Theorem \ref{q2detecKerrdSaccel} we derived closed-form analytic expressions for the Abdelqader-Lake  differential invariant $Q_2$ for accelerating Kerr black holes in (anti-)de Sitter spacetime.
Interestingly enough, we derived a very compact explicit formula for the differential invariant $Q_2$ for the accelerating Kerr-Newman black hole in (anti-)de Sitter spacetime: Theorem \ref{exactQ2accelKNdS} and eqn.(\ref{totalQ2accelKN(a)dS}).
From eqn.(\ref{totalQ2accelKN(a)dS}) we conclude that the local invariant $Q_2$ can serve as an event horizon detector, since it vanishes at the horizon radii and acceleration horizon radii.
For zero acceleration and zero electric charge the invariant $Q_2$ is given by expression (\ref{Q2KerrdeSitter1}) in Corollary \ref{porismaQ2kerrdeSitter}. In this case we provide a rigorous proof that the invariant can serve as a horizon detector: it vanishes at the stationary horizons and is nonzero everywhere else. For the proof we use Descarte's rule of signs Theorem \ref{rene},  and Bolzano's theorem \ref{bernardbolzano}.
In section \ref{RogerPenN} we apply the Newman-Penrose (NP) formalism to calculate the Abdelqader-Lake  differential invariants.
Using the Bianchi identities in NP formalism, eqns.(\ref{BianchiEin})-(\ref{BianchiVier}) and a specific null-tetrad we derive an explicit expression for the Page-Shoom invariant $W$, eqn.(\ref{pageshoominvt}), Theorem \ref{cohompageshoom}, for an accelerating Kerr-Newman black hole in (anti-)de Sitter spacetime. We then prove that it is vanishing at the stationary horizons, Corollary \ref{anixneftisorizontagegonotwn}.

In Appendix \ref{callosRT} we define the $\mathcal{R},\mathcal{T}$ regions relevant for the analysis of the norms of the gradient flows of the first two Weyl invariants in ZM classification. In Appendix \ref{Curbastro} we compute the explicit algebraic expression for the curvature invariant constructed from the covariant derivative of the Ricci tensor for accelerating, rotating, and charged black holes with $\Lambda\not =0$, see Theorem \ref{RiccisynalparagogosN} and Eqn.(\ref{CovRicciInvaccelKNdS}).

\section{Preliminaries on differential curvature invariants}
Taking into account the contribution from the cosmological
constant $\Lambda,$ the generalisation of the Kerr-Newman solution \cite{Newman},\cite{KerrR},
is described by the Kerr-Newman de Sitter $($KNdS$)$ metric
element which in Boyer-Lindquist (BL) coordinates is given by \cite{BCAR},\cite{Stuchlik1},\cite{GrifPod},\cite{ZdeStu} (in units where $G=1$ and $c=1$):
\begin{align}
\mathrm{d}s^{2}  & =-\frac{\Delta_{r}^{KN}}{\Xi^{2}\rho^{2}}(\mathrm{d}%
t-a\sin^{2}\theta\mathrm{d}\phi)^{2}+\frac{\rho^{2}}{\Delta_{r}^{KN}%
}\mathrm{d}r^{2}+\frac{\rho^{2}}{\Delta_{\theta}}\mathrm{d}\theta
^{2}\nonumber \\ &+\frac{\Delta_{\theta}\sin^{2}\theta}{\Xi^{2}\rho^{2}}(a\mathrm{d}%
t-(r^{2}+a^{2})\mathrm{d}\phi)^{2}%
\label{KNADSelement}
\end{align}%
\begin{equation}
\Delta_{\theta}:=1+\frac{a^{2}\Lambda}{3}\cos^{2}\theta,
\;\Xi:=1+\frac {a^{2}\Lambda}{3},
\end{equation}

\begin{equation}
\Delta_{r}^{KN}:=\left(  1-\frac{\Lambda}{3}r^{2}\right)  \left(  r^{2}
+a^{2}\right)  -2mr+q^{2},
\label{DiscrimiL}
\end{equation}

\begin{equation}
\rho^{2}=r^{2}+a^{2}\cos^{2}\theta,
\end{equation}
where $a,m,q,$ denote the Kerr parameter, mass and electric charge
of the black hole, respectively.
The KN(a)dS metric is the most general exact stationary  solution of the Einstein-Maxwell system of differential equations, that represents a non-accelerating, rotating, charged black hole with $\Lambda\not =0$.
This
is accompanied by a non-zero electromagnetic field
$F=\mathrm{d}A,$ where the vector potential is
\cite{GrifPod},\cite{ZST}:
\begin{equation}
A=-\frac{qr}{\Xi(r^{2}+a^{2}\cos^{2}\theta)}(\mathrm{d}t-a\sin^{2}\theta
\mathrm{d}\phi).
\end{equation}

The Christoffel symbols of the second kind are expressed in the coordinate basis in the form:
\begin{equation}
\Gamma^{\lambda}_{\;\mu\nu}=\frac{1}{2}g^{\lambda\alpha}(g_{\mu\alpha,\nu}+g_{\nu\alpha,\mu}-g_{\mu\nu,\alpha}),
\end{equation}
where the summation convention is adopted and a comma denotes a partial derivative.
The Riemann curvature tensor is given by:
\begin{align}
R^{\kappa}_{\;\;\lambda\mu\nu}=\Gamma^{\kappa}_{\;\lambda\nu,\mu}-\Gamma^{\kappa}_{\;\lambda\mu,\nu}+
\Gamma^{\alpha}_{\lambda\nu}\Gamma^{\kappa}_{\;\alpha\mu}-\Gamma^{\alpha}_{\lambda\mu}\Gamma^{\kappa}_{\alpha\nu}.
\end{align}
The symmetric Ricci tensor and the Ricci scalar are defined by:
\begin{equation}
R_{\mu\nu}=R^{\alpha}_{\;\mu\alpha\nu},\;\;\;R=g^{\alpha\beta}R_{\alpha\beta},
\end{equation}
while the Weyl tensor $C_{\kappa\lambda\mu\nu}$ (the trace-free part of the curvature tensor) is given explicitly in terms of the curvature tensor and the metric from the expression:
\begin{align}
C_{\kappa\lambda\mu\nu}=R_{\kappa\lambda\mu\nu}&+\frac{1}{2}(R_{\lambda\mu}g_{\kappa\nu}+R_{\kappa\nu}g_{\lambda\mu}-R_{\lambda\nu}g_{\kappa\mu}-R_{\kappa\mu}g_{\lambda\nu})\nonumber \\&+\frac{1}{6}R(g_{\kappa\mu}g_{\lambda\nu}-g_{\kappa\nu}g_{\lambda\mu}).
\label{WeylH}
\end{align}
The Weyl tensor has in general, ten independent components which at any point are completely independent of the Ricci components. It corresponds to the {\em free gravitational field} \footnote{Globally, however, the Weyl tensor and Ricci tensor are not independent, as they are connected by the differential Bianchi identities. These identities determine the interaction between the free gravitational field and the field sources.} \cite{szekeres}.

The dual of the Weyl tensor, $C_{\alpha\beta\gamma\delta}^{*}$,  is defined by:
 \begin{equation}
 C_{\alpha\beta\gamma\delta}^{*}=\frac{1}{2}E_{\alpha\beta\kappa\lambda}C^{\kappa\lambda}_{\;\;\gamma\delta},
 \end{equation}
 where $E_{\alpha\beta\kappa\lambda}$ is the Levi-Civita pseudotensor.

 \subsection{The Bianchi identities and the Karlhede invariant}\label{andersKarlhede1}

Using its covariant form, the classical Bianchi identities read \cite{lbianchi}:
\begin{equation}
R_{\lambda\mu\nu\kappa;\eta}+R_{\lambda\mu\eta\nu;\kappa}+R_{\lambda\mu\kappa\eta;\nu}=0
\label{Bianchi}
\end{equation}
In eqn(\ref{Bianchi}), the symbol \lq;\rq denotes covariant differentiation.
Karlhede and collaborators introduced the following coordinate-invariant and Lorentz invariant object the so called Karlhede invariant \cite{KarlhedeA}:
\begin{equation}
\mathfrak{K}=R^{\lambda\mu\nu\kappa;\eta}R_{\lambda\mu\nu\kappa;\eta}.
\label{andersKarlhede}
\end{equation}
Karlhede {\it et al} computed the invariant $\mathfrak{K}$ for the Schwarzschild black hole and remarkably showed that it is zero and changes sign  on the Schwarzschild event horizon.
Indeed, their result for the differential invariant in Eqn.(\ref{andersKarlhede})for the Schwarzschild black hole is \cite{KarlhedeA}:
\begin{equation}
\mathfrak{K}=-\frac{240 \left(6 m^{3} r^{3}-3 m^{2} r^{4}\right)}{r^{12}}.
\end{equation}
Unfortunately, this intriguing result does not generalises to the case of the Kerr black hole. The analytic computation of the Karlhede invariant for the Kerr solution reads:
\begin{align}
&\mathfrak{K}^{Kerr}=\frac{720 m^{2}}{\left(r^{2}+a^{2} \cos \! \left(\theta \right)^{2}\right)^{9}} \left(\cos \! \left(\theta \right)^{4} a^{4}-4 \cos \! \left(\theta \right)^{3} a^{3} r -6 \cos \! \left(\theta \right)^{2} a^{2} r^{2}+4 \cos \! \left(\theta \right) a \,r^{3}+r^{4}\right)\nonumber \\
& \left(\cos \! \left(\theta \right)^{4} a^{4}+4 \cos \! \left(\theta \right)^{3} a^{3} r -6 \cos \! \left(\theta \right)^{2} a^{2} r^{2}-4 \cos \! \left(\theta \right) a \,r^{3}+r^{4}\right) \left(a^{2} \cos \! \left(\theta \right)^{2}-2 m r +r^{2}\right).
\label{AKarlHKerr}
\end{align}
It is evident from Eqn.(\ref{AKarlHKerr}) that the Karlhede curvature invariant does not vanish on the Kerr event and Cauchy horizons.
In Boyer and Lindquist coordinates the event and Cauchy horizons are located on the surface defined by $\Delta(r)\equiv r^2+a^2-2mr=0$, and are given by the expressions:
\begin{equation}
r_{\pm}=m\pm \sqrt{m^2-a^{ 2}}.
\end{equation}
The outer horizon $r_+$ is referred to as the event horizon, while $r_-$ is known as the Cauchy horizon.

However, we note that it vanishes on the infinite-redshift surfaces, where $g_{tt}=0$. Equivalently at the roots of  the  quadratic equation:
\begin{equation}
r^2+a^2\cos \! \left(\theta \right)^{2}-2mr=0,
\end{equation}
or
\begin{equation}
r_E^{\pm}=m\pm\sqrt{m^2-a^2\cos \! \left(\theta \right)^{2}}.
\label{radiiergosurfaces}
\end{equation}
Indeed, shortly after the discovery of the Kerr black hole, it was realised that a region existed outside of the black hole's event horizon where no time-like observer could remain stationary. Six years after the discovery of the Kerr metric, Penrose showed that particles within this ergosphere region could possess negative energy, as measured by an observer at infinity \cite{PenroseRogerErgoRegion}.
Let us consider a coordinate-stationary observer with a four-velocity $u^{\mu}=(u^t,0,0,0)$. For the observer to have a time-like trajectory, we require $g_{tt}u^tu^t<0$, or alternatively:
\begin{align}
g_{tt}&=-\frac{a^{2}-2 m r +r^{2}}{r^{2}+a^{2} \cos \! \left(\theta \right)^{2}}+\frac{\sin \! \left(\theta \right)^{2} a^{2}}{r^{2}+a^{2} \cos \! \left(\theta \right)^{2}}\nonumber \\
&=-\left(1-\frac{2 m r}{r^{2}+a^{2} \cos \! \left(\theta \right)^{2}}\right)<0.
\end{align}
This inequality implies:
\begin{equation}
r^{2}-2 m r + a^{2} \cos \! \left(\theta \right)^{2}>0
\label{ergosphera}
\end{equation}
The two roots of the quadratic expression (with positive leading coefficient) in inequality (\ref{ergosphera}) are given by Eqn.(\ref{radiiergosurfaces}).
Thus the inequality for an observer with a physical trajectory is satisfied for
\begin{equation}
r>r_E^+
\end{equation}
or $r<r_E^-$.
Between the two roots $r_E^{\pm}$ the quadratic is negative (opposite sign of the sign of its leading coefficient).
These arguments imply that there is a region outside of $r_+$ where no stationary observer can exist. This space is called the \textit{ergosphere} and is bounded from above by the surface defined by $r_E^+$ \footnote{For the Schwarzschild solution the surface of infinite redshift $g_{tt}=0(r=2m)$ and the event horizon coincide. In \cite{Israel} the imbedding expression for the Kretschmann scalar was used  to prove the uniqueness theorems
for the Schwarzschild and Reissener-Nordstr\"{o}m black hole
solutions.}.

\subsection{Abdelqader-Lake differential  curvature invariants}\label{AbdelLimni}
In the work \cite{LakeZwei}, Abdelqader and Lake, introduced  the following curvature invariants and studied them for the Kerr metric:
\begin{align}
  I_1&\equiv C_{\alpha\beta\gamma\delta}C^{\alpha\beta\gamma\delta},\\
  I_2&\equiv C^{*}_{\alpha\beta\gamma\delta}C^{\alpha\beta\gamma\delta},\\
  I_3&\equiv \nabla_{\mu}C_{\alpha\beta\gamma\delta}\nabla^{\mu}C^{\alpha\beta\gamma\delta},\\
  I_4&\equiv \nabla_{\mu}C_{\alpha\beta\gamma\delta}\nabla^{\mu}C^{*\alpha\beta\gamma\delta},\\
  I_5&\equiv k_{\mu}k^{\mu},\\
  I_6&\equiv l_{\mu}l^{\mu},\;\;I_7\equiv k_{\mu}l^{\mu},
 \end{align}
  where $k_{\mu}=-\nabla_{\mu}I_1$ and $l_{\mu}=-\nabla_{\mu}I_2$.
  They also defined the following invariants:
\begin{align}
Q_1&\equiv\frac{1}{3\sqrt{3}}\frac{(I_1^2-I_2^2)(I_5-I_6)+4I_1I_2I_7}{(I_1^2+I_2^2)^{9/4}}\label{Qeins},\\
Q_2&\equiv\frac{1}{27}\frac{I_5I_6-I_7^2}{(I_1^2+I_2^2)^{5/2}}\label{qzwei},\\
Q_3&\equiv\frac{1}{6\sqrt{3}}\frac{I_5+I_6}{(I_1^2+I_2^2)^{5/4}\label{Qdrei}}.
\end{align}

Page and Shoom observed that the Abdelqader-Lake invariant $Q_2$ can be rewritten as follows \cite{PageD}:
\begin{equation}
Q_2=\frac{(I_6+I_5)^2-\left(\frac{12}{5}\right)^2(I_1^2+I_2^2)(I_3^2+I_4^2)}{108(I_1^2+I_2^2)^{5/2}}.
\end{equation}
Their crucial observation was that the curvature invariant $Q_2$ can be expressed as the norm of the wedge product of two differential forms, namely:
\begin{equation}
27 (I_1^2+I_2^2)^{5/2} Q_2=2\left\lVert dI_1\wedge dI_2 \right\rVert^2,
\label{donandrey}
\end{equation}
where \cite{PageD}:
\begin{equation}
\left\lVert dI_1\wedge dI_2 \right\rVert^2=\frac{1}{2}((k_{\mu}k^{\mu})(l_{\nu}l^{\nu})-(k_{\mu}l^{\mu})(k_{\nu}l^{\nu})).
\end{equation}
Under the light of this observation, the invariant $Q_2$ vanishes when the two gradient fields $k_{\mu}$ and $l_{\mu}$ are parallel \footnote{This led Page and Shoom \cite{PageD} to propose a generalisation of $Q_2$, and introduce an $n-$form differential invariant $W$ that vanishes on Killing horizons (i.e. $ \left\lVert W \right\rVert^2 \overset{\circ}{=}0$), in stationary spacetimes of local cohomogeneity $n$ . For $n=2$ we compute the invariant  $W$ for an accelerating Kerr-Newman black hole in (anti-)de Sitter spacetime in section \ref{winvps}. }.

To calculate the above curvature invariants for the metric (\ref{KNADSelement}) we used Maple\textsuperscript{TM}2021.

\section{Computation of the norm of the covariant derivative of Riemann tensor in the Kerr-Newman-(anti-)de Sitter spacetime}\label{karlhedeKNLambda}

We start our computations with the analytic calculation of the Karlhede curvature invariant for the Kerr-Newman-(anti-)de Sitter black hole.

\begin{theorem}\label{unserekarlhede}
We calculated in closed analytic form the  Karlhede invariant for the Kerr-Newman-(anti-)de Sitter black hole. Our result is:
\begin{align}
&\mathfrak{K}\equiv\nabla_{\mu}R_{\alpha\beta\gamma\delta}\nabla^{\mu}R^{\alpha\beta\gamma\delta}
\nonumber \\
&=\frac{1}{3 \left(r^{2}+a^{2} \cos \! \left(\theta \right)^{2}\right)^{9}}\Biggl[720 \Lambda  \cos \! \left(\theta \right)^{12} a^{12} m^{2}-720 a^{10} \Biggl(\Lambda  a^{2} m^{2}+28 \Lambda  m^{2} r^{2}-16 \Lambda  m \,q^{2} r \nonumber \\
&+\frac{76}{45} \Lambda  q^{4}-3 m^{2}\Biggr) \cos \! \left(\theta \right)^{10}+19440 \Biggl(\frac{23 \Lambda  m^{2} r^{4}}{9}-\frac{8 \Lambda  m \,q^{2} r^{3}}{3}+\left(\Lambda  a^{2} m^{2}+\frac{248}{405} \Lambda  q^{4}-3 m^{2}\right) r^{2}\nonumber \\
&-\frac{16 m \left(a^{2} q^{2} \Lambda +\frac{3}{8} m^{2}-3 q^{2}\right) r}{27}+\frac{76 q^{2} \left(a^{2} q^{2} \Lambda +\frac{135}{76} m^{2}-3 q^{2}\right)}{1215}\Biggr) a^{8} \cos \! \left(\theta \right)^{8}-30240 a^{6} \Biggl(-\frac{22 \Lambda  m \,q^{2} r^{5}}{35}\nonumber \\
&+\left(\Lambda  a^{2} m^{2}+\frac{278}{945} \Lambda  q^{4}-3 m^{2}\right) r^{4}\nonumber \\
&-\frac{10 m \left(a^{2} q^{2} \Lambda +\frac{14}{5} m^{2}-3 q^{2}\right) r^{3}}{7}+\frac{356 q^{2} \left(a^{2} q^{2} \Lambda +\frac{1755}{178} m^{2}-3 q^{2}\right) r^{2}}{945}\nonumber \\
&+\frac{2 q^{2} \left(a^{2}-\frac{151 q^{2}}{45}\right) m r}{7}-\frac{22 a^{2} q^{4}}{315}+\frac{16 q^{6}}{315}\Biggr) \cos \! \left(\theta \right)^{6}-30240 a^{4} r^{2} \Biggl(\frac{23 \Lambda  m^{2} r^{6}}{14}-\frac{166 \Lambda  m \,q^{2} r^{5}}{105}+\Biggl(\Lambda  a^{2} m^{2}\nonumber \\
&+\frac{278}{945} \Lambda  q^{4}-3 m^{2}\Biggr) r^{4}-\frac{74 \left(a^{2} q^{2} \Lambda -\frac{525}{37} m^{2}-3 q^{2}\right) m \,r^{3}}{105}-\frac{103 m^{2} q^{2} r^{2}}{7}-\frac{2 q^{2} \left(a^{2}-\frac{353 q^{2}}{15}\right) m r}{7}\nonumber\\
&+\frac{26 a^{2} q^{4}}{105}-\frac{14 q^{6}}{15}\Biggr) \cos \! \left(\theta \right)^{4}+19440 a^{2} r^{4} \Biggl(\frac{28 \Lambda  m^{2} r^{6}}{27}-\frac{76 \Lambda  m \,q^{2} r^{5}}{45}+\left(\Lambda  a^{2} m^{2}+\frac{248}{405} \Lambda  q^{4}-3 m^{2}\right) r^{4}\nonumber \\
&-\frac{44 \left(a^{2} q^{2} \Lambda -\frac{42}{11} m^{2}-3 q^{2}\right) m \,r^{3}}{27}+\frac{712 q^{2} \left(a^{2} q^{2} \Lambda -\frac{4023}{178} m^{2}-3 q^{2}\right) r^{2}}{1215}+\frac{4 q^{2} \left(a^{2}+\frac{295 q^{2}}{27}\right) m r}{5}-\frac{52 a^{2} q^{4}}{135}\nonumber \\
&-\frac{248 q^{6}}{135}\Biggr) \cos \! \left(\theta \right)^{2}-720 \Biggl(\Lambda  m^{2} r^{6}-\frac{12 \Lambda  m \,q^{2} r^{5}}{5}+\left(\Lambda  a^{2} m^{2}+\frac{76}{45} \Lambda  q^{4}-3 m^{2}\right) r^{4}\nonumber \\
&-\frac{12 \left(a^{2} q^{2} \Lambda -\frac{5}{2} m^{2}-3 q^{2}\right) m \,r^{3}}{5}+\frac{76 q^{2} \left(a^{2} q^{2} \Lambda -\frac{783}{76} m^{2}-3 q^{2}\right) r^{2}}{45}+\frac{12 q^{2} \left(a^{2}+\frac{65 q^{2}}{9}\right) m r}{5}\nonumber \\
&-\frac{44 a^{2} q^{4}}{15}-\frac{76 q^{6}}{15}\Biggr) r^{6}\Biggr].
\label{KarlHKNdS}
\end{align}
\end{theorem}

 \begin{theorem}
 We computed analytically the Karlhede invariant for the Kerr-(anti-)de Sitter black hole. Our result is:
 \begin{align}
 &\nabla_{\mu}R_{\alpha\beta\gamma\delta}\nabla^{\mu}R^{\alpha\beta\gamma\delta}=\frac{1}{3 \left(r^{2}+a^{2} \cos \! \left(\theta \right)^{2}\right)^{9}}\Biggl(720 \Lambda  \cos \! \left(\theta \right)^{12} a^{12} m^{2}\nonumber\\
 &-720 a^{10} \left(\Lambda  a^{2} m^{2}+28 \Lambda  m^{2} r^{2}-3 m^{2}\right) \cos \! \left(\theta \right)^{10}\nonumber \\
 &+19440 \left(\frac{23 \Lambda  m^{2} r^{4}}{9}+\left(\Lambda  a^{2} m^{2}-3 m^{2}\right) r^{2}-\frac{2 m^{3} r}{9}\right) a^{8} \cos \! \left(\theta \right)^{8}\nonumber \\
 &-30240 a^{6} \left(\left(\Lambda  a^{2} m^{2}-3 m^{2}\right) r^{4}-4 m^{3} r^{3}\right) \cos \! \left(\theta \right)^{6}\nonumber \\
 &-30240 a^{4} r^{2} \left(\frac{23 \Lambda  m^{2} r^{6}}{14}+\left(\Lambda  a^{2} m^{2}-3 m^{2}\right) r^{4}+10 m^{3} r^{3}\right) \cos \! \left(\theta \right)^{4}\nonumber \\
 &+19440 a^{2} r^{4} \left(\frac{28 \Lambda  m^{2} r^{6}}{27}+\left(\Lambda  a^{2} m^{2}-3 m^{2}\right) r^{4}+\frac{56 m^{3} r^{3}}{9}\right) \cos \! \left(\theta \right)^{2}\nonumber \\
 &-720 \left(\Lambda  m^{2} r^{6}+\left(\Lambda  a^{2} m^{2}-3 m^{2}\right) r^{4}+6 m^{3} r^{3}\right) r^{6}\Biggr)\nonumber\\
 &=\frac{240 m^{2}}{\left(r^{2}+a^{2} \cos \! \left(\theta \right)^{2}\right)^{9}}\left(\cos \! \left(\theta \right)^{4} a^{4}-4 \cos \! \left(\theta \right)^{3} a^{3} r -6 \cos \! \left(\theta \right)^{2} a^{2} r^{2}+4 \cos \! \left(\theta \right) a \,r^{3}+r^{4}\right) \nonumber \\
  &\times \Biggl(\cos \! \left(\theta \right)^{4} a^{4}+4 \cos \! \left(\theta \right)^{3} a^{3} r
   -6 \cos \! \left(\theta \right)^{2} a^{2} r^{2}-4 \cos \! \left(\theta \right) a \,r^{3}+r^{4}\Biggr)\nonumber \\
    &\times \left(\Lambda  \cos \! \left(\theta \right)^{4} a^{4}-\Lambda  \cos \! \left(\theta \right)^{2} a^{4}-\Lambda  a^{2} r^{2}-\Lambda  r^{4}+3 a^{2} \cos \! \left(\theta \right)^{2}-6 m r +3 r^{2}\right)\label{KarlhedeKerr(a)dS}.
 \end{align}
 \end{theorem}

 \begin{corollary}
 For zero rotation (i.e $a=0$), Eqn.(\ref{KarlHKNdS}) reduces to the analytic exact expression of the Karlhede invariant for the Reissner-Nordstr\"{o}m-(anti-)de Sitter black hole:
 \begin{align}
 &\nabla_{\mu}R_{\alpha\beta\gamma\delta}\nabla^{\mu}R^{\alpha\beta\gamma\delta}\nonumber \\
 &=-\frac{240 \left(\Lambda  m^{2} r^{6}-\frac{12 \Lambda  m \,q^{2} r^{5}}{5}+\left(\frac{76 \Lambda  q^{4}}{45}-3 m^{2}\right) r^{4}-\frac{12 \left(-\frac{5 m^{2}}{2}-3 q^{2}\right) m \,r^{3}}{5}+\frac{76 q^{2} \left(-\frac{783 m^{2}}{76}-3 q^{2}\right) r^{2}}{45}+\frac{52 q^{4} m r}{3}-\frac{76 q^{6}}{15}\right)}{r^{12}}\nonumber \\
 &=-\left(-\frac{720 m^{2}}{r^{8}}+\frac{1728 q^{2} m}{r^{9}}-\frac{1216 q^{4}}{r^{10}}\right) \left(1+\frac{q^{2}}{r^{2}}-\frac{2 m}{r}-\frac{\Lambda  r^{2}}{3}\right).\label{KarlRNadS}
 \end{align}
 \end{corollary}
\begin{remark}
We observe  from eqn.(\ref{KarlRNadS}), that the  Karlhede invariant for the Reissner-Nordstr\"{o}m-(anti-)de Sitter black hole vanishes on the black hole horizon.
\end{remark}

In Fig.\ref{ContourPlotsKarlhede} we plot level curves for the Karlhede invariant we computed in Eqn.( \ref{KarlhedeKerr(a)dS}) , in the $r-\theta$ space, for different sets of values of the physical black hole parameters $a,\Lambda,m$. In the context of setting $m=1$, a dimensionless form of $\Lambda$ corresponds to the dimensionless combination $m^2\Lambda$, unless otherwise stipulated. This means that for supermassive black holes such as at the centre of Galaxy M87 with mass $M^{M87}_{\rm BH}=6.7\times 10^9$ solar masses \cite{EHT} the value of dimensionless $\Lambda=3\times 10^{-4}$ corresponds to the value for the cosmological constant: $\Lambda=3.06\times 10^{-34}{\rm cm}^{-2}$. For the galactic centre supermassive SgrA* black hole with mass $M_{\rm BH}^{SgrA*}=4.06\times 10^6$ solar masses \cite{Eisenhauer} a value for dimensionless $\Lambda=3.6\times 10^{-33}$ (which is the value we use in our graphs) corresponds to the value for cosmological constant $10^{-56}{\rm cm}^{-2}$ consistent with observations.

\begin{figure}[ptbh]
\centering
  \begin{subfigure}[b]{.60\linewidth}
    \centering
    \includegraphics[width=.99\textwidth]{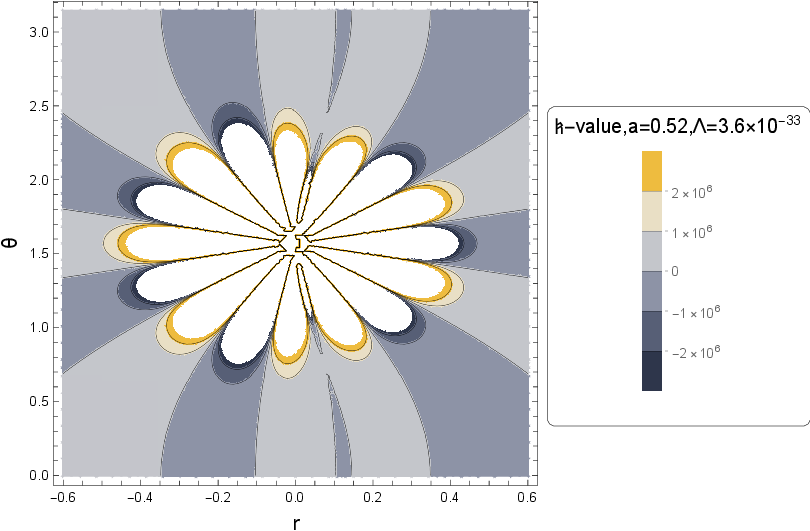}
    \caption{ Contour plot of Karlhede invariant $\mathfrak{K}$.}\label{ContourKarhedea052}
  \end{subfigure}%
  \begin{subfigure}[b]{.60\linewidth}
    \centering
    \includegraphics[width=.99\textwidth]{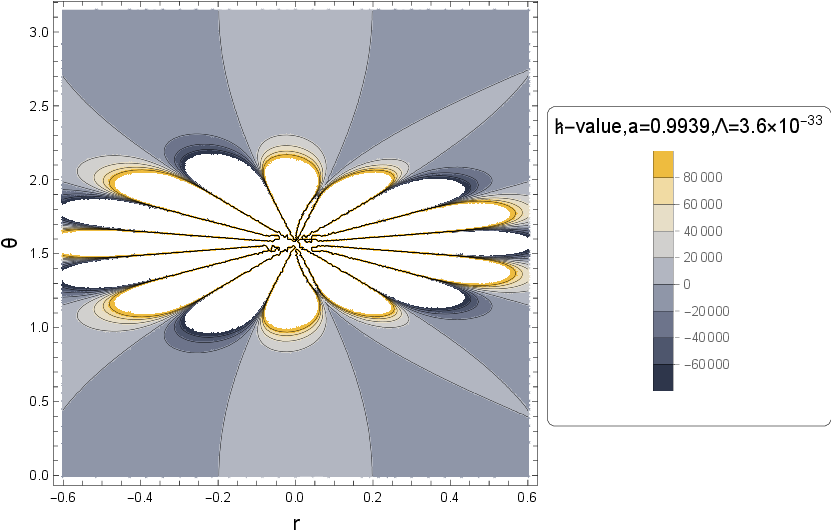}
    \caption{Contour plot of Karlhede invariant $\mathfrak{K}$.}\label{ContourKarhedea09939}
  \end{subfigure}\\
  \begin{subfigure}[b]{.60\linewidth}
    \centering
    \includegraphics[width=.99\textwidth]{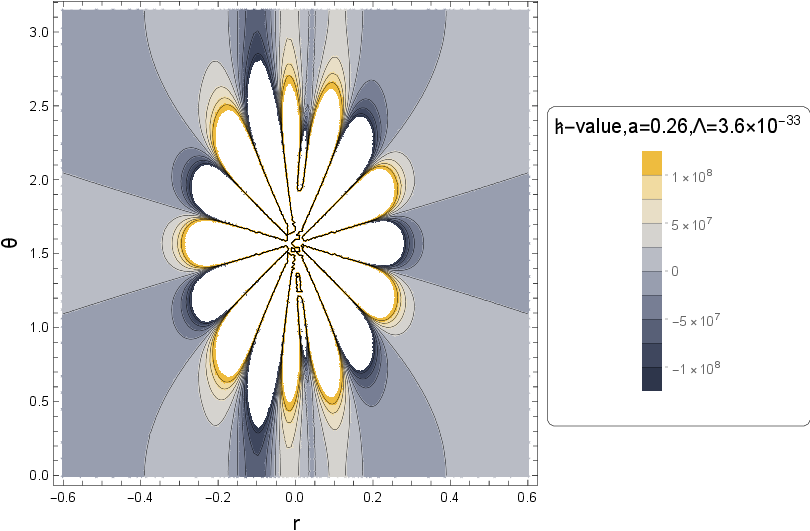}
    \caption{Contour Plot of Karlhede invariant $\mathfrak{K}$.}\label{ContourKarhedea026}
  \end{subfigure}%
  \caption{Contour plots of the Karlhede invariant $\mathfrak{K}$, eqn.(\ref{KarlhedeKerr(a)dS}), for the Kerr-(anti-)de Sitter black hole. (a) for spin parameter $a=0.52$, charge $q=0$,dimensionless cosmological parameter $\Lambda=3.6\times 10^{-33},m=1$. (b) For spin  $a=0.9939$, charge $q=0$,dimensionless cosmological parameter $\Lambda=3.6\times 10^{-33},m=1$. (c) For low spin  $a=0.26$, electric charge $q=0$,dimensionless cosmological parameter $\Lambda=3.6\times 10^{-33}$ and mass $m=1$.}\label{ContourPlotsKarlhede}
\end{figure}

\begin{theorem}
We computed an explicit algebraic expression for the differential curvature invariant $\nabla_{\mu}C^{*}_{\alpha\beta\gamma\delta}\nabla^{\mu}C^{*\alpha\beta\gamma\delta}$
for the Kerr-(anti-)de Sitter black hole spacetime:
\begin{align}
&\nabla_{\mu}C^{*}_{\alpha\beta\gamma\delta}\nabla^{\mu}C^{*\alpha\beta\gamma\delta}=-\frac{240 m^{2}}{\left(r^{2}+a^{2} \cos \! \left(\theta \right)^{2}\right)^{9}} \left(a^{4} \cos \! \left(\theta \right)^{4}-4 \cos \! \left(\theta \right)^{3} a^{3} r -6 a^{2} \cos \! \left(\theta \right)^{2} r^{2}+4 \cos \! \left(\theta \right) a \,r^{3}+r^{4}\right)\nonumber \\
&\times  \left(a^{4} \cos \! \left(\theta \right)^{4}+4 \cos \! \left(\theta \right)^{3} a^{3} r -6 a^{2} \cos \! \left(\theta \right)^{2} r^{2}-4 \cos \! \left(\theta \right) a \,r^{3}+r^{4}\right)\nonumber \\
& \times \left(\Lambda  \cos \! \left(\theta \right)^{4} a^{4}-\Lambda  \cos \! \left(\theta \right)^{2} a^{4}-\Lambda  a^{2} r^{2}-\Lambda  r^{4}+3 a^{2} \cos \! \left(\theta \right)^{2}-6 m r +3 r^{2}\right).
\label{KerrLambdanormnablaweyldual}
\end{align}
\end{theorem}

\begin{remark}
It is evident from Eqn.(\ref{KerrLambdanormnablaweyldual}) that the differential curvature invariant $\nabla_{\mu}C^{*}_{\alpha\beta\gamma\delta}\nabla^{\mu}C^{*\alpha\beta\gamma\delta}$ vanishes at the ergosurfaces, i.e the roots of the equation: $g_{tt}=0.$ It serves therefore as a detector of ergosurfaces for the Kerr-(anti-) de Sitter black hole.
\end{remark}

 \begin{corollary}
 The invariant $\nabla_{\mu}C^{*}_{\alpha\beta\gamma\delta}\nabla^{\mu}C^{*\alpha\beta\gamma\delta}$
for the Kerr spacetime is calculated in closed analytic form with the result:
\begin{align}
&\nabla_{\mu}C^{*}_{\alpha\beta\gamma\delta}\nabla^{\mu}C^{*\alpha\beta\gamma\delta}=
-\frac{720 m^{2}}{\left(r^{2}+a^{2} \cos \! \left(\theta \right)^{2}\right)^{9}} \left(\cos \! \left(\theta \right)^{4} a^{4}-4 \cos \! \left(\theta \right)^{3} a^{3} r -6 \cos \! \left(\theta \right)^{2} a^{2} r^{2}+4 \cos \! \left(\theta \right) a \,r^{3}+r^{4}\right)\nonumber \\
&\times \left(\cos \! \left(\theta \right)^{4} a^{4}+4 \cos \! \left(\theta \right)^{3} a^{3} r -6 \cos \! \left(\theta \right)^{2} a^{2} r^{2}-4 \cos \! \left(\theta \right) a \,r^{3}+r^{4}\right) \left(a^{2} \cos \! \left(\theta \right)^{2}-2 m r +r^{2}\right).
\end{align}
  \end{corollary}
  \begin{theorem}
  We calculated in closed analytic form  the invariant $\nabla_{\mu}C^{*}_{\alpha\beta\gamma\delta}\nabla^{\mu}C^{*\alpha\beta\gamma\delta}$
for the Reissner-Nordstr\"{o}m black hole. Indeed, our computation yields:
  \begin{align}
  \nabla_{\mu}C^{*}_{\alpha\beta\gamma\delta}\nabla^{\mu}C^{*\alpha\beta\gamma\delta}=\frac{48 \left(15 m^{2} r^{2}-36 m \,q^{2} r +22 q^{4}\right) \left(2 m r -q^{2}-r^{2}\right)}{r^{12}}
  \end{align}
  \end{theorem}

  \section{Analytic computation of the Abdelqade-Lake local invariants for the Kerr-Newman-(anti) de Sitter black hole and  their role in detecting black hole horizons and/or ergosurfaces}\label{KNlambdaAbdeLake}
  \subsection{Rotating and charged black holes with $\Lambda\not=0$}

The curvature invariants $I_1,I_2$ for the case of Kerr-Newman black holes in (anti-)de Sitter spacetime have been calculated-see Eqn(39) and Eqn.(33) in \cite{KraniotisCurvature}.
We now derive a novel explicit algebraic expression for the norm of the covariant derivative of the Weyl tensor:

   \begin{theorem}\label{synalParWeyl}
  Our analytic computation for the invariant $I_3$ for a Kerr-Newman-(anti-)de Sitter black hole yields the result:
  \begin{align}
  &I_3=\frac{1}{\left(r^{2}+a^{2} \cos \! \left(\theta \right)^{2}\right)^{9}}\Biggl(240 \Lambda  a^{12} \cos \! \left(\theta \right)^{12} m^{2}-240 \Biggl[\Lambda  a^{2} m^{2}+28 m^{2} r^{2} \Lambda -16 m \,q^{2} r \Lambda +\frac{22}{15} q^{4} \Lambda \nonumber \\
  &-3 m^{2}\Biggr] a^{10} \cos \! \left(\theta \right)^{10}+6480 \Biggl\{\frac{23 m^{2} r^{4} \Lambda}{9}-\frac{8 m \,q^{2} r^{3} \Lambda}{3}+\left(\Lambda  a^{2} m^{2}+\frac{254}{405} q^{4} \Lambda -3 m^{2}\right) r^{2}\nonumber \\
  &-\frac{16 \left(a^{2} q^{2} \Lambda +\frac{3}{8} m^{2}-3 q^{2}\right) m r}{27}+\frac{22 q^{2} \left(a^{2} q^{2} \Lambda +\frac{45}{22} m^{2}-3 q^{2}\right)}{405}\Biggr\} a^{8} \cos \! \left(\theta \right)^{8}-10080 \Biggl[-\frac{22 m \,q^{2} r^{5} \Lambda}{35}\nonumber \\
  &+\left(\Lambda  a^{2} m^{2}+\frac{92}{315} q^{4} \Lambda -3 m^{2}\right) r^{4}-\frac{10 \left(a^{2} q^{2} \Lambda +\frac{14}{5} m^{2}-3 q^{2}\right) m \,r^{3}}{7}+\frac{122 \left(a^{2} q^{2} \Lambda +\frac{585}{61} m^{2}-3 q^{2}\right) q^{2} r^{2}}{315}\nonumber \\
  &+\frac{2 q^{2} \left(a^{2}-\frac{10 q^{2}}{3}\right) m r}{7}-\frac{2 a^{2} q^{4}}{35}+\frac{q^{6}}{21}\Biggr] a^{6} \cos \! \left(\theta \right)^{6}-10080 \Biggl\{\frac{23 m^{2} r^{6} \Lambda}{14}-\frac{166 m \,q^{2} r^{5} \Lambda}{105}+\Biggl(\Lambda  a^{2} m^{2}+\frac{92}{315} q^{4} \Lambda\nonumber \\
   &-3 m^{2}\Biggr) r^{4}-\frac{74 m \left(a^{2} q^{2} \Lambda -\frac{525}{37} m^{2}-3 q^{2}\right) r^{3}}{105}-\frac{103 m^{2} q^{2} r^{2}}{7}-\frac{2 \left(a^{2}-\frac{352 q^{2}}{15}\right) q^{2} m r}{7}+\frac{2 a^{2} q^{4}}{7}\nonumber \\
   &-\frac{97 q^{6}}{105}\Biggr\} a^{4} r^{2} \cos \! \left(\theta \right)^{4}+6480 \Biggl[\frac{28 m^{2} r^{6} \Lambda}{27}-\frac{76 m \,q^{2} r^{5} \Lambda}{45}+\left(\Lambda  a^{2} m^{2}+\frac{254}{405} q^{4} \Lambda -3 m^{2}\right) r^{4}\nonumber \\
   &-\frac{44 m \left(a^{2} q^{2} \Lambda -\frac{42}{11} m^{2}-3 q^{2}\right) r^{3}}{27}+\frac{244 q^{2} \left(a^{2} q^{2} \Lambda -\frac{1341}{61} m^{2}-3 q^{2}\right) r^{2}}{405}+\frac{4 \left(a^{2}+\frac{298 q^{2}}{27}\right) q^{2} m r}{5}-\frac{4 a^{2} q^{4}}{9}\nonumber \\
   &-\frac{254 q^{6}}{135}\Biggr] a^{2} r^{4} \cos \! \left(\theta \right)^{2}-240 \Biggl[m^{2} r^{6} \Lambda -\frac{12 m \,q^{2} r^{5} \Lambda}{5}+\left(\Lambda  a^{2} m^{2}+\frac{22}{15} q^{4} \Lambda -3 m^{2}\right) r^{4}\nonumber \\
   &-\frac{12 \left(a^{2} q^{2} \Lambda -\frac{5}{2} m^{2}-3 q^{2}\right) m \,r^{3}}{5}+\frac{22 \left(a^{2} q^{2} \Lambda -\frac{261}{22} m^{2}-3 q^{2}\right) q^{2} r^{2}}{15}+\left(\frac{12}{5} a^{2} m \,q^{2}+16 q^{4} m \right) r -\frac{12 a^{2} q^{4}}{5}\nonumber \\
   &-\frac{22 q^{6}}{5}\Biggr] r^{6}\Biggr).
   \label{normcovderweyl}
  \end{align}
  \end{theorem}

\begin{theorem}
We calculated an exact algebraic expression for the invariant $I_3$ in the case of the Kerr-(anti-)de Sitter black hole. The result is:
\begin{align}
&I_3=\frac{240 m^{2}}{\left(r^{2}+a^{2} \cos \! \left(\theta \right)^{2}\right)^{9}} \left(a^{4} \cos \! \left(\theta \right)^{4}-4 \cos \! \left(\theta \right)^{3} a^{3} r -6 a^{2} \cos \! \left(\theta \right)^{2} r^{2}+4 \cos \! \left(\theta \right) a \,r^{3}+r^{4}\right)\nonumber \\
&\times \left(a^{4} \cos \! \left(\theta \right)^{4}+4 \cos \! \left(\theta \right)^{3} a^{3} r -6 a^{2} \cos \! \left(\theta \right)^{2} r^{2}-4 \cos \! \left(\theta \right) a \,r^{3}+r^{4}\right)\nonumber \\
&\times \left(\Lambda  a^{4} \cos \! \left(\theta \right)^{4}-\Lambda  \cos \! \left(\theta \right)^{2} a^{4}-\Lambda  a^{2} r^{2}-\Lambda  r^{4}+3 a^{2} \cos \! \left(\theta \right)^{2}-6 m r +3 r^{2}\right).
\label{LAI3kdS}
\end{align}
\end{theorem}

\begin{remark}
We observe from Eqn.(\ref{LAI3kdS}) that the curvature invariant $I_3$ for a Kerr-(anti-)de Sitter black hole vanishes on the boundary of the ergosphere region.
\end{remark}
Indeed the stationary limit surfaces of the rotational Kerr-(anti-)de Sitter black hole are defined by $g_{tt}=0$. The metric element $g_{tt}$ in this case is given by:
\begin{align}
&g_{tt}=-\frac{a^{2}+q^{2}-2 m r +r^{2}-\frac{\Lambda  \left(a^{2}+r^{2}\right) r^{2}}{3}}{\left(r^{2}+a^{2} \cos \! \left(\theta \right)^{2}\right) \left(1+\frac{a^{2} \Lambda}{3}\right)}+\frac{\left(1+\frac{a^{2} \Lambda  \cos \left(\theta \right)^{2}}{3}\right) \sin \! \left(\theta \right)^{2} a^{2}}{\left(r^{2}+a^{2} \cos \! \left(\theta \right)^{2}\right) \left(1+\frac{a^{2} \Lambda}{3}\right)}\nonumber \\
&=-\frac{1}{3\rho^2\Xi^2}\left(\Lambda  \cos \! \left(\theta \right)^{4} a^{4}-\Lambda  \cos \! \left(\theta \right)^{2} a^{4}-\Lambda  a^{2} r^{2}-\Lambda  r^{4}+3 a^{2} \cos \! \left(\theta \right)^{2}-6 m r +3 r^{2}\right)=0.
\end{align}
\begin{remark}
We observe from eqn.(\ref{KarlhedeKerr(a)dS}) and Eqn.(\ref{LAI3kdS}) that for the case of Kerr-(anti-)de Sitter spacetime the Karlhede invariant is equal to the invariant $I_3$. This is consistent with the fact that in vacuum spacetimes with a nonzero cosmological constant $\Lambda$, $R_{\alpha\beta}=\Lambda g_{\alpha\beta},R=4\Lambda\not=0$, so that $C_{\alpha\beta\gamma\delta;\epsilon}=R_{\alpha\beta\gamma\delta;\epsilon}$.
\end{remark}

The analytic computation of the invariant $I_4$, for the case of the Kerr-Newman-(anti-)de Sitter black hole yields the result:
\begin{theorem}\label{mixedsynparweyl}
\begin{align}
&I_4^{\rm {KN(a)dS}}=\frac{1920 a \cos \! \left(\theta \right) }{\left(r^{2}+a^{2} \cos \! \left(\theta \right)^{2}\right)^{9}} \Biggl(a^{10} m \Lambda  \left(m r -\frac{3 q^{2}}{10}\right) \cos \! \left(\theta \right)^{10}-a^{8} \Biggl(7 \Lambda  m^{2} r^{3}-\frac{57 \Lambda  m \,q^{2} r^{2}}{10}+\Biggl[a^{2} m^{2} \Lambda +\frac{29}{30} q^{4} \Lambda\nonumber \\
& -3 m^{2}\Biggr] r -\frac{3 m \,q^{2} \left(\Lambda  a^{2}-3\right)}{10}\Biggr) \cos \! \left(\theta \right)^{8}+6 a^{6} \Biggl(\Lambda  m^{2} r^{5}-\frac{83 m \,q^{2} r^{4} \Lambda}{60}+\left(a^{2} m^{2} \Lambda +\frac{37}{90} q^{4} \Lambda -3 m^{2}\right) r^{3}\nonumber \\
&-\frac{11 \left(a^{2} \Lambda  q^{2}+\frac{12}{11} m^{2}-3 q^{2}\right) m \,r^{2}}{12}+\frac{29 \left(a^{2} \Lambda  q^{2}+\frac{126}{29} m^{2}-3 q^{2}\right) q^{2} r}{180}+\frac{m \,q^{2} \left(a^{2}-2 q^{2}\right)}{20}\Biggr) \cos \! \left(\theta \right)^{6}\nonumber \\
&+\frac{37 a^{4} r }{10}\Biggl(\frac{60 \Lambda  m^{2} r^{6}}{37}-\frac{33 r^{5} m \,q^{2} \Lambda}{37}+m \left(a^{2} \Lambda  q^{2}+\frac{420}{37} m^{2}-3 q^{2}\right) r^{3}-\frac{19 \left(a^{2} \Lambda  q^{2}+\frac{498}{19} m^{2}-3 q^{2}\right) q^{2} r^{2}}{37}\nonumber \\
&-\frac{27 \left(a^{2}-\frac{178 q^{2}}{27}\right) m \,q^{2} r}{37}+\frac{12 q^{4} a^{2}}{37}-\frac{17 q^{6}}{37}\Biggr) \cos \left(\theta \right)^{4}-6 a^{2} r^{3} \Biggl(\frac{7 \Lambda  m^{2} r^{6}}{6}-\frac{3 r^{5} m \,q^{2} \Lambda}{2}+\Biggl[a^{2} m^{2} \Lambda +\frac{37}{90} q^{4} \Lambda \nonumber \\
&-3 m^{2}\Biggr] r^{4}-\frac{5 m \left(a^{2} \Lambda  q^{2}-\frac{28}{5} m^{2}-3 q^{2}\right) r^{3}}{4}+\frac{19 \left(a^{2} \Lambda  q^{2}-\frac{750}{19} m^{2}-3 q^{2}\right) q^{2} r^{2}}{60}+\frac{\left(a^{2}+\frac{418 q^{2}}{15}\right) m \,q^{2} r}{4}\nonumber \\
&-\frac{37 q^{6}}{30}\Biggr) \cos \! \left(\theta \right)^{2}+r^{5} \Biggl(\Lambda  m^{2} r^{6}-2 r^{5} m \,q^{2} \Lambda +\left(a^{2} m^{2} \Lambda +\frac{29}{30} q^{4} \Lambda -3 m^{2}\right) r^{4}-2 m \left(a^{2} \Lambda  q^{2}-3 m^{2}-3 q^{2}\right) r^{3}\nonumber \\
&+\left(\frac{29}{30} a^{2} q^{4} \Lambda -15 m^{2} q^{2}-\frac{29}{10} q^{4}\right) r^{2}+\frac{3 m \,q^{2} \left(a^{2}+\frac{118 q^{2}}{15}\right) r}{2}-\frac{6 q^{4} a^{2}}{5}-\frac{29 q^{6}}{10}\Biggr)\Biggr).
\label{covdermixedweyli4}
\end{align}
\end{theorem}

\begin{theorem}
The closed form analytic solution for the curvature invariant $I_4$ in the Kerr-(anti-)de Sitter spacetime is:
\begin{align}
&I_4=\frac{1920 \cos \! \left(\theta \right)r \,m^{2} a}{\left(r^{2}+a^{2} \cos \! \left(\theta \right)^{2}\right)^{9}} \left(\Lambda  \cos \! \left(\theta \right)^{4} a^{4}-\Lambda  \cos \! \left(\theta \right)^{2} a^{4}-\Lambda  a^{2} r^{2}-\Lambda  r^{4}+3 a^{2} \cos \! \left(\theta \right)^{2}-6 m r +3 r^{2}\right)\nonumber \\
&\times \left(a^{2} \cos \! \left(\theta \right)^{2}-2 \cos \! \left(\theta \right) a r -r^{2}\right) \left(a^{2} \cos \! \left(\theta \right)^{2}+2 \cos \! \left(\theta \right) a r -r^{2}\right)  \left(a^{2} \cos \! \left(\theta \right)^{2}-r^{2}\right).
\end{align}
\end{theorem}

In Fig.\ref{TRIAD14}, we display three-dimensional plots of the differential invariant $I_4$  as a function of the Boyer-Lindquist coordinates $r$ and $\theta$, for three sets of values for the spin,cosmological constant, electric charge and mass of the black hole.

\begin{figure}[ptbh]
\centering
  \begin{subfigure}[b]{.60\linewidth}
    \centering
    \includegraphics[width=.99\textwidth]{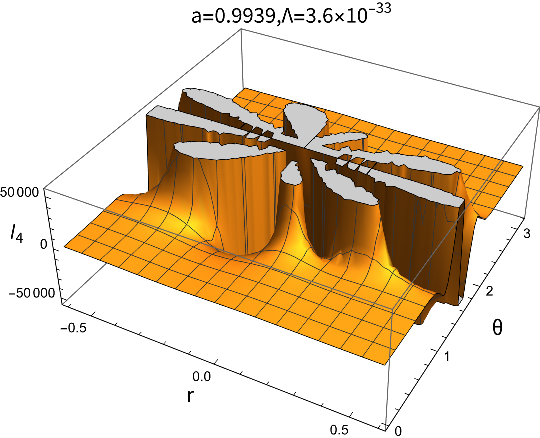}
    \caption{ 3D plot of  $I_4$.}\label{graphI4a09939q0Lobs}
  \end{subfigure}%
  \begin{subfigure}[b]{.60\linewidth}
    \centering
    \includegraphics[width=.99\textwidth]{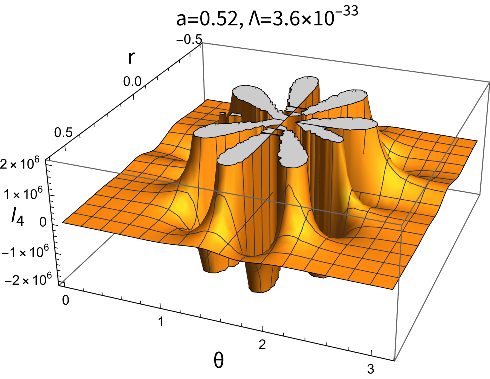}
    \caption{3D plot of $I_4$}\label{graphI4a052Lobs}
  \end{subfigure}\\
  \begin{subfigure}[b]{.60\linewidth}
    \centering
    \includegraphics[width=.99\textwidth]{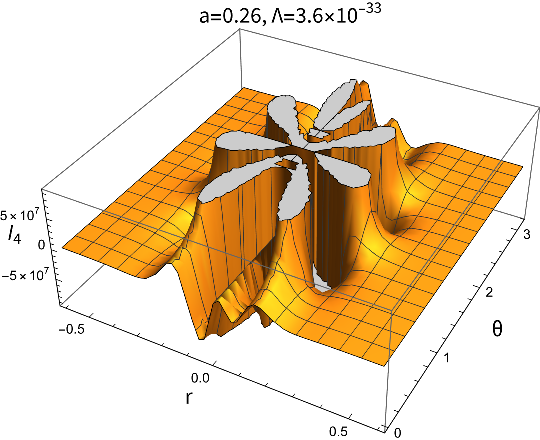}
    \caption{3D Plot of $I_4$.}\label{grafosI4a09939Lobs}
  \end{subfigure}%
  \caption{3D plots of the differential curvature invariant,  $I_4$, plotted as a function of the Boyer-Lindquist coordinates $r$ and $\theta$ . (a) for high spin parameter $a=0.9939$, charge $q=0$, $\Lambda=3.6\times 10^{-33}$,$m=1$. (b) For spin  $a=0.52$, charge $q=0$,$\Lambda=3.6\times 10^{-33}$, $m=1$. (c) For low spin  $a=0.26$, electric charge $q=0$,$\Lambda=3.6\times 10^{-33}$  and mass $m=1$.}\label{TRIAD14}
\end{figure}

  We computed the invariants $I_6,I_5$ for the Kerr and Kerr-(anti-)de Sitter metrics \footnote{We have calculated the corresponding explicit expressions for the KN(a-)dS black hole, however the resulting expressions are long and cumbersome and we refrain from presenting them.}:

\begin{theorem}
We computed the explicit algebraic expression for the invariant $I_6$ for the Kerr-(anti-)de Sitter black hole. The result is:
\begin{align}
&I_6=-\frac{1354752m^{4} a^{2}}{\left(r^{2}+a^{2} \cos \! \left(\theta \right)^{2}\right)^{15}} \Biggl(\Lambda  \cos \! \left(\theta \right)^{16} a^{14} r^{2}-\frac{48 \left(\frac{163 \Lambda  r^{4}}{16}+\left(\Lambda  a^{2}-3\right) r^{2}-\frac{m r}{8}+\frac{a^{2}}{16}\right) a^{12} \cos \left(\theta \right)^{14}}{49}\nonumber \\
&+\frac{64}{7} \left(\frac{211 \Lambda  r^{4}}{64}+\left(\Lambda  a^{2}-3\right) r^{2}-\frac{9 m r}{16}-\frac{3 a^{2}}{64}\right) r^{2} a^{10} \cos \left(\theta \right)^{12}-\frac{144 r^{4}}{7} \Biggl[\frac{139 \Lambda  r^{4}}{144}+\left(\Lambda  a^{2}-3\right) r^{2}-\frac{73 m r}{24}\nonumber \\ &+\frac{a^{2}}{16}\Biggr] a^{8} \cos \left(\theta \right)^{10}
-\frac{15 r^{6} \left(\frac{139}{15} \Lambda  r^{4}+\frac{424}{5} m r +a^{2}\right) a^{6} \cos \left(\theta \right)^{8}}{7}+\frac{144 r^{8} a^{4} }{7} \Biggl[\frac{211 \Lambda  r^{4}}{144}+\left(\Lambda  a^{2}-3\right) r^{2}+\frac{217 m r}{24}\nonumber \\
&-\frac{5 a^{2}}{48}\Biggr] \cos \left(\theta \right)^{6}
-\frac{64 r^{10} \left(\frac{489 \Lambda  r^{4}}{448}+\left(\Lambda  a^{2}-3\right) r^{2}+\frac{105 m r}{16}+\frac{9 a^{2}}{64}\right) a^{2} \cos \left(\theta \right)^{4}}{7}\nonumber \\
&+\left(\Lambda  r^{16}+\left(-\frac{144}{49}+\frac{48 \Lambda  a^{2}}{49}\right) r^{14}+6 m \,r^{13}-\frac{3 a^{2} r^{12}}{7}\right) \cos \! \left(\theta \right)^{2}-\frac{3 r^{14}}{49}\Biggr).
\label{gradientcurvaI6}
\end{align}
\end{theorem}
\begin{corollary}
 For the Kerr metric ($\Lambda=0$) we find:
\begin{align}
  &I_6=\frac{82944m^{4} a^{2}}{\left(r^{2}+a^{2} \cos \! \left(\theta \right)^{2}\right)^{15}} \Biggl(a^{12} \left(a^{2}-2 m r -48 r^{2}\right) \cos \! \left(\theta \right)^{14}+7 a^{10} r^{2} \left(a^{2}+12 m r +64 r^{2}\right) \cos \! \left(\theta \right)^{12}\nonumber \\
  &+21 r^{4} \left(a^{2}-\frac{146}{3} m r -48 r^{2}\right) a^{8} \cos \! \left(\theta \right)^{10}+\left(35 a^{8} r^{6}+2968 a^{6} m \,r^{7}\right) \cos \! \left(\theta \right)^{8}+\Biggl[35 r^{8} a^{6}-3038 a^{4} m \,r^{9}\nonumber \\
  &+1008 r^{10} a^{4}\Biggr] \cos \! \left(\theta \right)^{6}+\left(21 r^{10} a^{4}+980 a^{2} m \,r^{11}-448 a^{2} r^{12}\right) \cos \! \left(\theta \right)^{4}\nonumber \\
  &+\left(7 a^{2} r^{12}-98 m \,r^{13}+48 r^{14}\right) \cos \! \left(\theta \right)^{2}+r^{14}\Biggr).
  \end{align}
\end{corollary}

\begin{theorem}
We computed the following  explicit analytic expression for the  differential invariant $I_5$ for the spacetime of the Kerr-(anti-)de Sitter black hole:
\begin{align}
&I_5=-\frac{27648 m^{4}}{\left(r^{2}+a^{2} \cos \! \left(\theta \right)^{2}\right)^{15}} \Biggl(\Lambda  \cos \! \left(\theta \right)^{18} a^{18}-a^{16} \left(\Lambda  a^{2}+42 \Lambda  r^{2}-3\right) \cos \! \left(\theta \right)^{16}+42 a^{14} \Biggl(\left(\Lambda  r^{2}-\frac{1}{14}\right) a^{2}\nonumber \\
&+\frac{73 \Lambda  r^{4}}{6}-3 r^{2}\Biggr) \cos \! \left(\theta \right)^{14}-462 \left(\left(\Lambda  r^{2}+\frac{1}{22}\right) a^{2}-\frac{7 \left(-\frac{205}{42} \Lambda  r^{3}+m +\frac{33}{7} r \right) r}{11}\right) r^{2} a^{12} \cos \! \left(\theta \right)^{12}\nonumber \\
&+994 \left(\left(\Lambda  r^{2}-\frac{9}{142}\right) a^{2}-\frac{210 \left(-\frac{7}{20} \Lambda  r^{3}+m +\frac{71}{70} r \right) r}{71}\right) r^{4} a^{10} \cos \! \left(\theta \right)^{10}\nonumber \\
&+\left(-105 a^{10} r^{6}+\Biggl[1029 \Lambda  r^{10}\nonumber +9114 m \,r^{7}\Bigg] a^{8}\right) \cos \! \left(\theta \right)^{8}-994 \Biggl[\left(\Lambda  r^{2}+\frac{15}{142}\right) a^{2}+\frac{205 \Lambda  r^{4}}{142}+\frac{636 m r}{71}\nonumber \\
&-3 r^{2}\Biggr] r^{8} a^{6} \cos \! \left(\theta \right)^{6}+462 r^{10} a^{4} \left(\left(\Lambda  r^{2}-\frac{3}{22}\right) a^{2}+\frac{73 r \left(\frac{1}{6} \Lambda  r^{3}+m -\frac{33}{73} r \right)}{11}\right) \cos \! \left(\theta \right)^{4}\nonumber \\
&-42 r^{12} \left(\left(\Lambda  r^{2}+\frac{1}{2}\right) a^{2}+\Lambda  r^{4}+6 m r -3 r^{2}\right) a^{2} \cos \! \left(\theta \right)^{2}+r^{14} \left(\left(\Lambda  r^{2}-3\right) a^{2}+\Lambda  r^{4}+6 m r -3 r^{2}\right)\Biggr)
\label{Ifunfkds}
\end{align}
\end{theorem}
\begin{corollary}
For $\Lambda=0$ (i.e. the Kerr metric), the invariant $I_5$ becomes:
\begin{align}
&I_5=-\frac{82944 m^{4}}{\left(r^{2}+a^{2} \cos \! \left(\theta \right)^{2}\right)^{15}} \Biggl(a^{16} \cos \! \left(\theta \right)^{16}+\left(-a^{16}-42 a^{14} r^{2}\right) \cos \! \left(\theta \right)^{14}-7 \left(a^{2}-14 m r -66 r^{2}\right) r^{2} a^{12} \cos \! \left(\theta \right)^{12}\nonumber \\
&+\left(-21 r^{4} a^{12}+\left(-980 m \,r^{5}-994 r^{6}\right) a^{10}\right) \cos \! \left(\theta \right)^{10}+\left(-35 a^{10} r^{6}+3038 m \,r^{7} a^{8}\right) \cos \! \left(\theta \right)^{8}\nonumber \\
&+\left(-35 r^{8} a^{8}+\left(-2968 m \,r^{9}+994 r^{10}\right) a^{6}\right) \cos \! \left(\theta \right)^{6}-21 \left(a^{2}-\frac{146}{3} m r +22 r^{2}\right) r^{10} a^{4} \cos \! \left(\theta \right)^{4}\nonumber \\
&-7 r^{12} \left(a^{2}+12 m r -6 r^{2}\right) a^{2} \cos \! \left(\theta \right)^{2}-a^{2} r^{14}+2 m \,r^{15}-r^{16}\Biggr)
\end{align}
\end{corollary}
\begin{corollary}\label{normdifinv1pi2}
For $\theta=\pi/2$, i.e. for the equatorial plane, $I_5$ reduces to:
\begin{equation}
I_5=-\frac{27648 \left(\left(\Lambda  r^{2}-3\right) a^{2}+\Lambda  r^{4}+6 m r -3 r^{2}\right) m^{4}}{r^{16}}.
\end{equation}
Thus $I_5$ in this case, vanishes at the stationary horizons of the Kerr black hole with cosmological constant.
\end{corollary}
\begin{corollary}
For $\theta=0$, i.e. along the axis, $I_5$ reduces to:
\begin{align}
I_5=-\frac{1354752 m^{4} \left(a^{6}-5 a^{4} r^{2}+3 a^{2} r^{4}-\frac{1}{7} r^{6}\right)^{2} r^{2} \left(\Lambda  r^{4}+\left(\Lambda  a^{2}-3\right) r^{2}+6 m r -3 a^{2}\right)}{\left(a^{2}+r^{2}\right)^{15}}.
\end{align}
For all the values of the Kerr-parameter $a$ consistent with a Kerr-(anti-)de Sitter black hole, the norm of the gradient vector $k_{\mu}$ (i.e. $I_5$) vanishes at the stationary black hole horizons and at the real positive roots of the sextic radial polynomial:
\begin{align}
r\simeq 0.481574618807 a,\;1.2539603376 a,\;4.3812862675 a.
\label{discrete}
\end{align}
\end{corollary}
Whereas the vanishing of $I_5$ in the equatorial plane (and on the axis away from the discrete roots in Eqn.(\ref{discrete})) singles out the horizons the global behaviour of $I_5$ is also interesting. In general terms, the area of regions where $I_5<0$ (this means that in these regions the vector $k_{\mu}$ is timelike) on any $r-\theta$ hypersurface decreases as the Kerr parameter increases. This strong sign dependence of the curvature invariant $I_5$ with black hole's spin $a$ is amply demonstrated in Fig. \ref{Regionsnegativenormkmu}.

\begin{figure}[ptbh]
\centering
  \begin{subfigure}[b]{.60\linewidth}
    \centering
    \includegraphics[width=.99\textwidth]{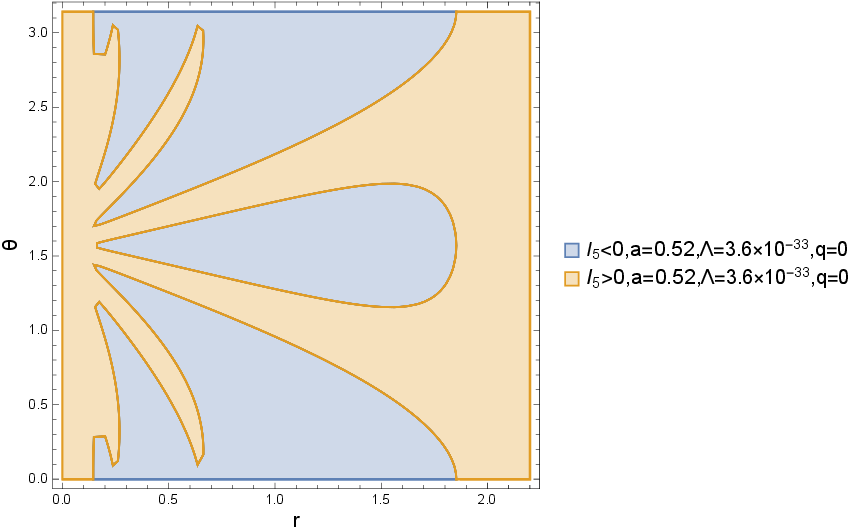}
    \caption{Region plot of gradient curvature invariant $I_5$.}\label{RegionnNormI5kdsa052}
  \end{subfigure}%
  \begin{subfigure}[b]{.60\linewidth}
    \centering
    \includegraphics[width=.99\textwidth]{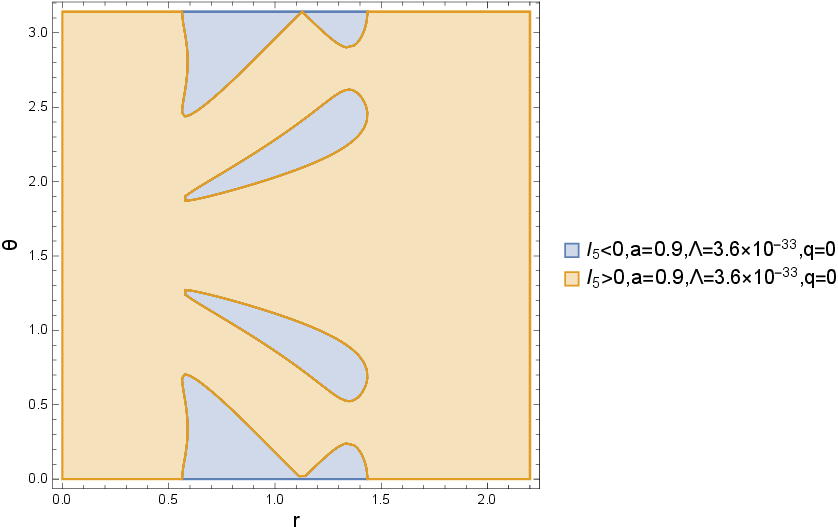}
    \caption{Region plot of gradient curvature invariant $I_5$.}\label{RegionI5kds09}
  \end{subfigure}\\
  \begin{subfigure}[b]{.60\linewidth}
    \centering
    \includegraphics[width=.99\textwidth]{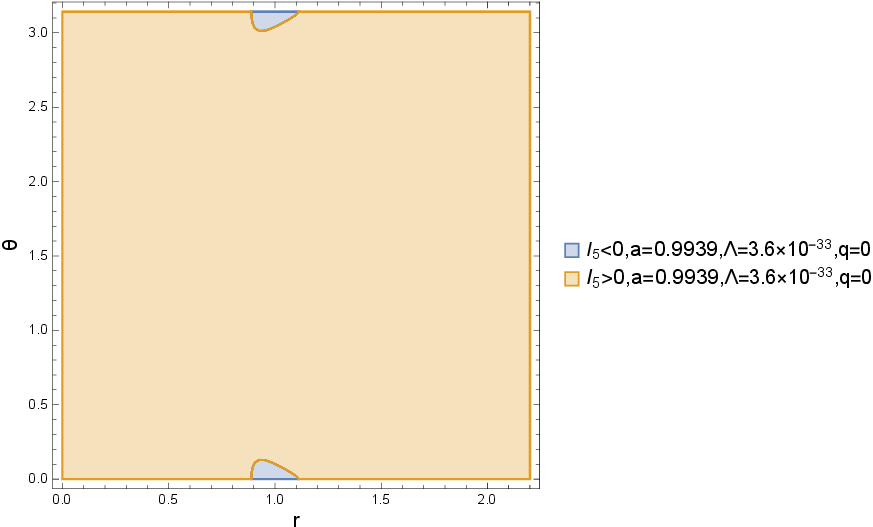}
    \caption{Region Plot of  gradient curvature invariant $I_5$.}\label{RegionI5kds09939}
  \end{subfigure}
  \caption{Region plots of the  curvature invariant $I_5$ , Eqn.(\ref{Ifunfkds}), for the Kerr-(anti-)de Sitter black hole. (a) for spin parameter $a=0.52$, charge $q=0$,dimensionless cosmological parameter $\Lambda=3.6\times 10^{-33},m=1$. (b) For spin  $a=0.9$, charge $q=0$,dimensionless cosmological parameter $\Lambda=3.6\times 10^{-33},m=1$. (c) For high spin  $a=0.9939$, electric charge $q=0$,dimensionless cosmological parameter $\Lambda=3.6\times 10^{-33}$ and mass $m=1$. We note that the area of regions of negative sign of the invariant $I_5$ (and hence the regions where the gradient vector $k_{\mu}=-\nabla_{\mu}(C_{\alpha\beta\gamma\delta}C^{\alpha\beta\gamma\delta})$ is timelike)  decreases as the  Kerr parameter $a$ increases.}\label{Regionsnegativenormkmu}
\end{figure}

We reiterated the analysis of the $I_5$ invariant for the case of the Kerr-Newman black hole (KN BH) in Figure \ref{RegionsnegativenormkmuKN}. The strong sign dependence on the black hole's spin is still present. However, there is now a new effect due to the electric charge. For a fixed value of the Kerr parameter $a$, increasing the electric charge results in a decrease of the area of regions where $I_5<0$ and thus the areas in the $r-\theta$ space where the gradient vector $k_{\mu}$
is timelike. As a matter of fact for the choice of values $a=0.9939,q=0.11$ the gradient vector $k_{\mu}$ is spacelike throughout the $r-\theta$ space, as shown in the Fig.\ref{RegionI5kna9939q011}.
As regards the values of $I_5$ in the equatorial plane and on the axis for a KN(a-)dS black hole we obtain:
\begin{corollary}
\begin{align}
I_5(\theta=\frac{\pi}{2})=-\frac{27648 \left(\Lambda  r^{4}+\left(\Lambda  a^{2}-3\right) r^{2}+6 m r -3 a^{2}-3 q^{2}\right) \left(m r -q^{2}\right)^{2} \left(m r -\frac{4 q^{2}}{3}\right)^{2}}{r^{20}},
\end{align}
\end{corollary}
and
\begin{corollary}
\begin{align}
I_5(\theta=0)&=-\frac{3072 \left(\Lambda  a^{2} r^{2}+\Lambda  r^{4}-3 a^{2}+6 m r -3 q^{2}-3 r^{2}\right)}{\left(a^{2}+r^{2}\right)^{15}} \Biggl(21 a^{6} m^{2} r -5 a^{6} m \,q^{2}-105 a^{4} m^{2} r^{3}\nonumber \\
&+85 a^{4} m \,q^{2} r^{2}-12 a^{4} q^{4} r +63 a^{2} m^{2} r^{5}-95 a^{2} m \,q^{2} r^{4}+32 a^{2} q^{4} r^{3}-3 r^{7} m^{2}+7 m \,q^{2} r^{6}-4 q^{4} r^{5}\Biggr)^{2}.
\end{align}
\end{corollary}

\begin{figure}[ptbh]
\centering
  \begin{subfigure}[b]{.60\linewidth}
    \centering
    \includegraphics[width=.99\textwidth]{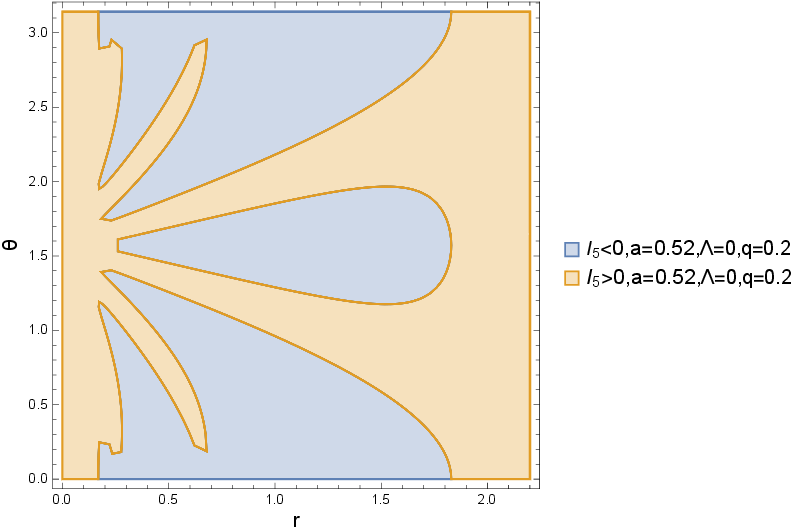}
    \caption{Region plot of curvature invariant $I_5$ for KN BH.}\label{RegionnNormI5kna052q02}
  \end{subfigure}%
  \begin{subfigure}[b]{.60\linewidth}
    \centering
    \includegraphics[width=.99\textwidth]{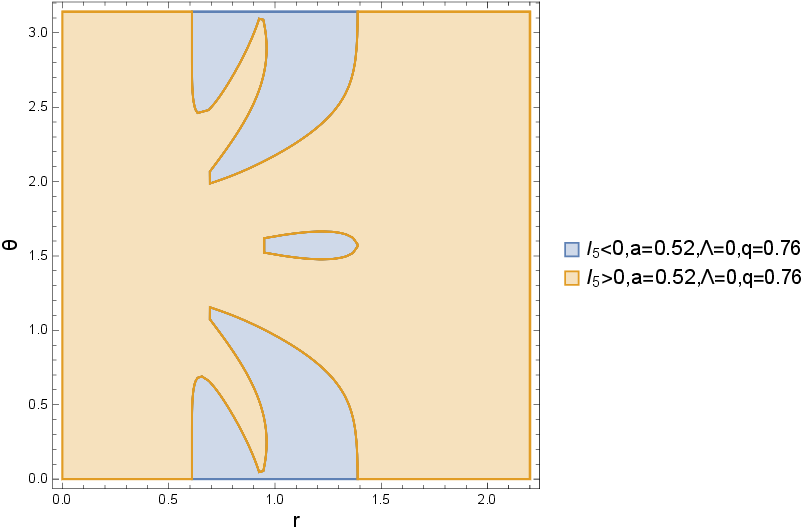}
    \caption{Region plot of curvature invariant $I_5$ for KN BH.}\label{RegionI5kna052q076}
  \end{subfigure}\\
  \begin{subfigure}[b]{.60\linewidth}
    \centering
    \includegraphics[width=.99\textwidth]{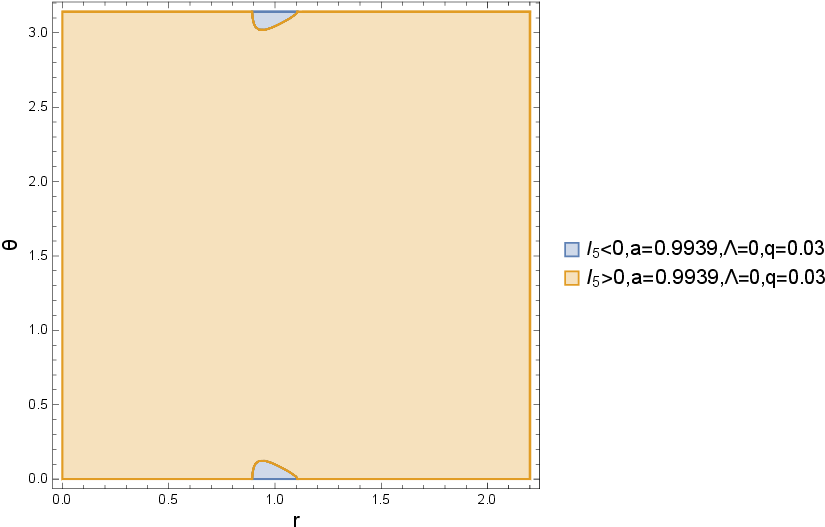}
    \caption{Region Plot of curvature invariant $I_5$ for KN BH.}\label{RegionI5kn09939q003}
  \end{subfigure}%
  \begin{subfigure}[b]{.60\linewidth}
    \centering
    \includegraphics[width=.99\textwidth]{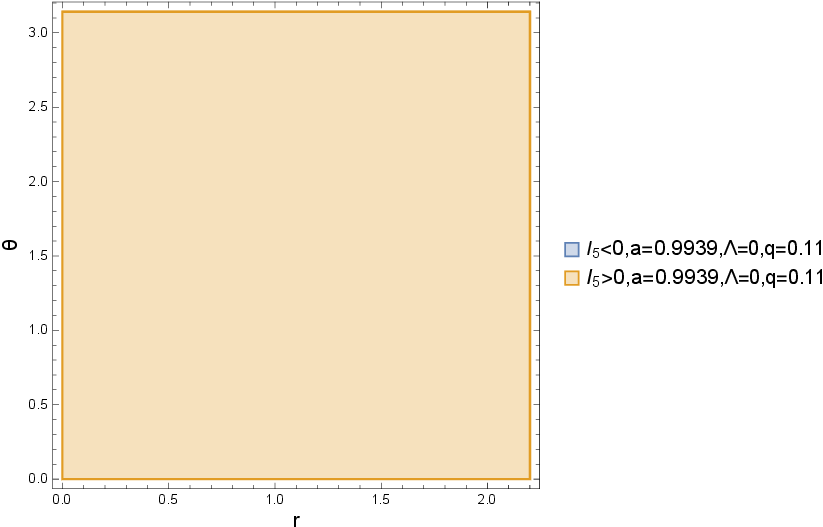}
    \caption{Region Plot of the invariant $I_5$ for KN BH.}\label{RegionI5kna9939q011}
  \end{subfigure}
  \caption{Region plots of the  curvature invariant $I_5$ , for the Kerr-Newman black hole. (a) for spin parameter $a=0.52$, charge $q=0.2,m=1$. (b) For spin  $a=0.52$, charge $q=0.76,m=1$. (c) For high spin  $a=0.9939$, electric charge $q=0.03$, and mass $m=1$. (d) For high spin  $a=0.9939$, electric charge $q=0.11$, and mass $m=1$.}\label{RegionsnegativenormkmuKN}
\end{figure}

Let us turn our attention to the curvature invariant $I_6$ for the Kerr-(anti-)de Sitter black hole. $I_6$ is the norm associated with the gradient of the Chern-Pontryagin invariant $I_2$. We note the following:
\begin{corollary}
In the equatorial plane, i.e.for $\theta=\frac{\pi}{2}$  eqn.(\ref{gradientcurvaI6}) reduces to:
\begin{equation}
I_6=\frac{82944 m^{4} a^{2}}{r^{16}}.
\end{equation}
\end{corollary}
\begin{corollary}
On the symmetry axis, i.e. for $\theta=0$   eqn.(\ref{gradientcurvaI6}) reduces to the expression:
\begin{equation}
I_6=-\frac{27648 m^{4} \left(a^{6}-21 a^{4} r^{2}+35 a^{2} r^{4}-7 r^{6}\right)^{2} \left(\Lambda  r^{4}+\left(\Lambda  a^{2}-3\right) r^{2}+6 m r -3 a^{2}\right) a^{2}}{\left(a^{2}+r^{2}\right)^{15}}.
\label{discreteI6}
\end{equation}
In this case the invariant $I_6$, i.e. the norm of the gradient vector $l_{\mu}$, vanishes at the horizon radii of the Kerr black hole with the cosmological constant present and also at the three positive real roots of the sextic polynomial in (\ref{discreteI6}):
\begin{equation}
r\simeq 0.22824347439 a,\;0.79747338888 a,\; 2.0765213965 a.
\label{diakritaI6}
\end{equation}
\end{corollary}
Again, $I_6$ shows a strong sign dependence on the Kerr parameter $a$. Whereas the vanishing of $I_6$ on the symmetry axis obviously singles out the horizons of the K(a-)dS black hole (away from the discrete points in eqn.(\ref{diakritaI6})), the global behaviour of $I_6$, like $I_5$, is also quite interesting. In general terms, the area of regions where $I_6<0$ (and so $l_{\mu}$ is a timelike vector) on any $r-\theta$ hypersurface decreases as $a$ increases.

In Fig. \ref{ContourPlotsgradI6curva} we display contour plots of the differential invariant $I_6$, eqn.(\ref{gradientcurvaI6}), for the Kerr-(anti-)de Sitter black hole. In particular, the strong sign dependence of $I_6$ with the Kerr parameter is evident in Fig.\ref{nullIsolinesI6}. Indeed, a plot of $I_6=0$ for different values of the Kerr parameter is shown in Fig.\ref{nullIsolinesI6}. The enclosed area decreases for increasing Kerr parameter $a$ and shrinks to zero for $a\rightarrow 1^-$.

Let us now investigate the curvature invariant $I_6$ for a Kerr-Newman black hole.
\begin{corollary}
For $\theta=\frac{\pi}{2}$in the equatorial plane,  $I_6$   reduces to:
\begin{align}
I_6=\frac{82944 \left(m r -\frac{2 q^{2}}{3}\right)^{2} \left(m r -q^{2}\right)^{2} a^{2}}{r^{20}}.
\label{klisiCPknequato}
\end{align}
\end{corollary}
\begin{corollary}
Along the symmetry axis ($\theta=0$,$\theta=\pi$), the norm of the gradient vector $l_{\mu}$ for a KN BH is given by:
\begin{align}
I_6&=\frac{9216 \left(a^{2}-2 m r +q^{2}+r^{2}\right)}{\left(a^{2}+r^{2}\right)^{15}} \Biggl(3 a^{6} m^{2}-63 a^{4} m^{2} r^{2}+32 a^{4} m \,q^{2} r -2 a^{4} q^{4}+105 a^{2} m^{2} r^{4}\nonumber \\&-120 a^{2} m \,q^{2} r^{3}
+28 a^{2} q^{4} r^{2}-21 m^{2} r^{6}+40 m \,q^{2} r^{5}-18 q^{4} r^{4}\Biggr)^{2} a^{2}.
\label{gradCPknaxis}
\end{align}
\end{corollary}
In Fig.\ref{RegionplotsofsignI6} the sign dependence of the norm of the gradient vector $l_{\mu}$ on the Kerr parameter $a$ and the electric charge $q$ of the Kerr-Newman black hole is nicely exhibited. The strong sign dependence on the black hole's spin is still present. On the other hand,  for a fixed value of the Kerr parameter $a$, increasing the electric charge results in a decrease of the area of regions where $I_6<0$ and thus the areas in the $r-\theta$ space where the gradient vector $l_{\mu}$ is timelike. Again, as $I_5$, for the values of physical black hole  parameters: $a=0.9939,q=0.11$ the gradient vector $l_{\mu}$ is spacelike throughout the $r-\theta$ space. In particular in this case, the positive norm of the gradient vector $l_{\mu}=-\nabla_{\mu}(C^{*}_{\alpha\beta\gamma\delta}C^{\alpha\beta\gamma\delta})$  defines an $\mathcal{R}$ region.
\begin{corollary}\label{norm2gradientflowI6eqax}
The generalisation of eqn.(\ref{gradCPknaxis}) for a Kerr-Newman-(anti-)de Sitter black hole is:
\begin{align}
I_6&=-\frac{3072 \left(\Lambda  a^{2} r^{2}+\Lambda  r^{4}-3 a^{2}+6 m r -3 q^{2}-3 r^{2}\right)}{\left(a^{2}+r^{2}\right)^{15}} \Biggl(3 m^{2} a^{6}-63 r^{2} a^{4} m^{2}+32 a^{4} m \,q^{2} r -2 a^{4} q^{4}\nonumber \\
&+105 a^{2} m^{2} r^{4}-120 a^{2} m \,q^{2} r^{3}+28 r^{2} a^{2} q^{4}-21 m^{2} r^{6}+40 m \,q^{2} r^{5}-18 q^{4} r^{4}\Biggr)^{2} a^{2},
\end{align}
while eqn.(\ref{klisiCPknequato}) remains the same in the presence of the cosmological constant.
\end{corollary}

\begin{figure}[ptbh]
\centering
  \begin{subfigure}[b]{.60\linewidth}
    \centering
    \includegraphics[width=.99\textwidth]{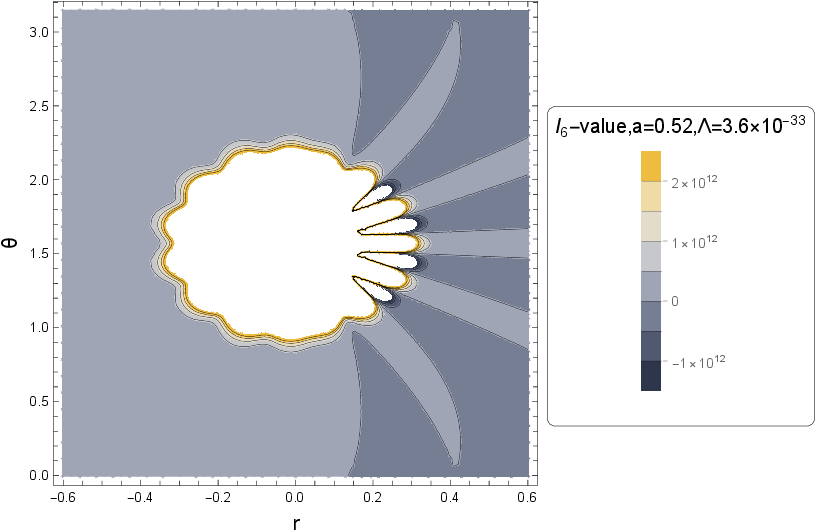}
    \caption{ Contour plot of gradient curvature invariant $I_6$.}\label{ContourI6kdsa052}
  \end{subfigure}%
  \begin{subfigure}[b]{.60\linewidth}
    \centering
    \includegraphics[width=.99\textwidth]{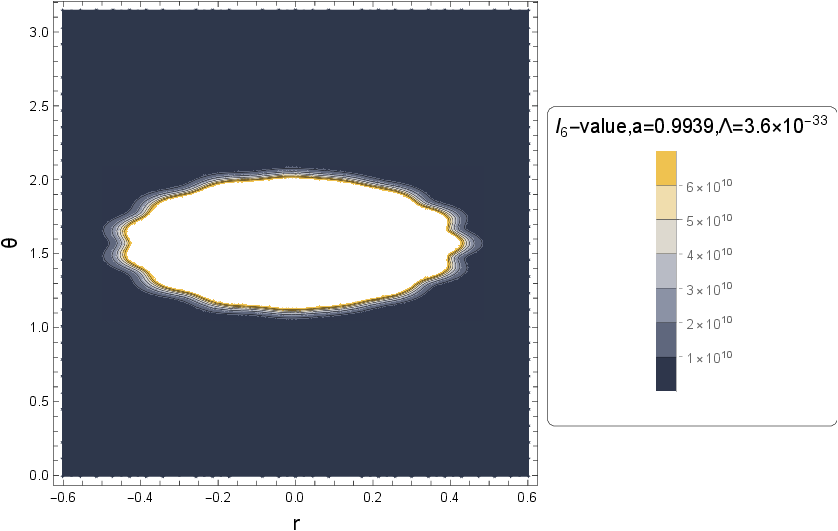}
    \caption{Contour plot of gradient curvature invariant $I_6$.}\label{ContourI6kds09939}
  \end{subfigure}\\
  \begin{subfigure}[b]{.60\linewidth}
    \centering
    \includegraphics[width=.99\textwidth]{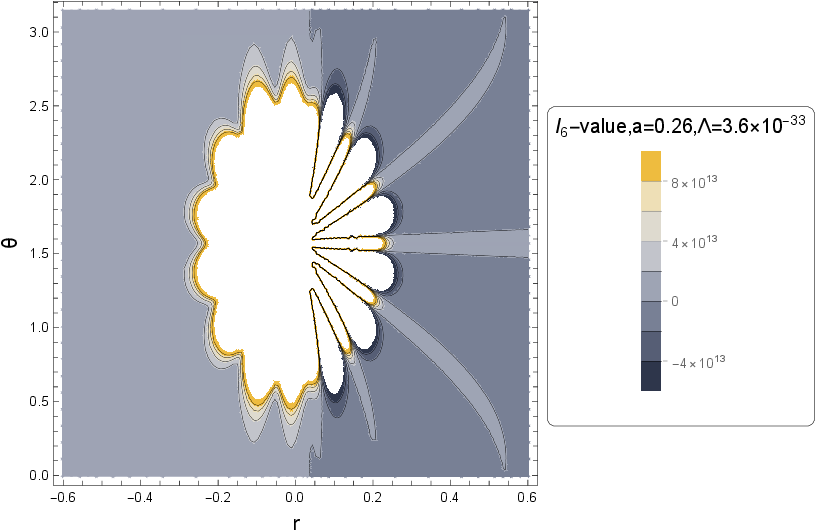}
    \caption{Contour Plot of  gradient curvature invariant $I_6$.}\label{ContourI6kdsea026}
  \end{subfigure}%
  \begin{subfigure}[b]{.60\linewidth}
    \centering
    \includegraphics[width=.99\textwidth]{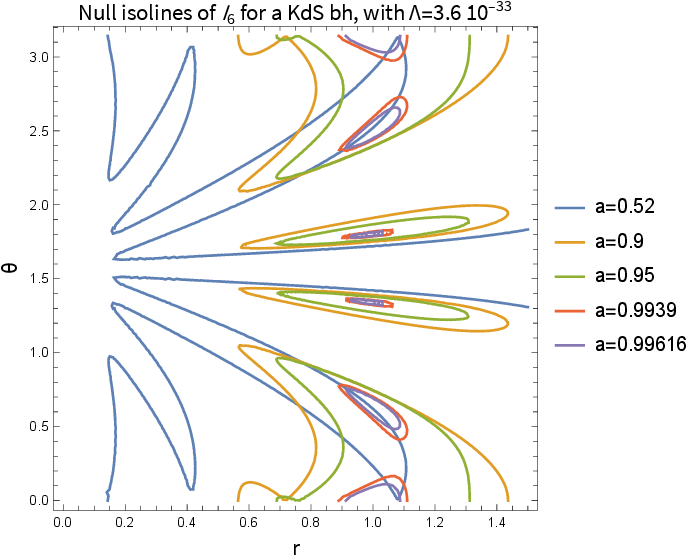}
    \caption{Plot of $I_6=0$ for different values of $a$.}\label{nullIsolinesI6}
  \end{subfigure}%
  \caption{Contour plots of the gradient curvature $I_6$ invariant, Eqn.(\ref{gradientcurvaI6}), for the Kerr-(anti-)de Sitter black hole. (a) for spin parameter $a=0.52$, charge $q=0$,dimensionless cosmological parameter $\Lambda=3.6\times 10^{-33},m=1$. (b) For spin  $a=0.9939$, charge $q=0$,dimensionless cosmological parameter $\Lambda=3.6\times 10^{-33},m=1$. (c) For low spin  $a=0.26$, electric charge $q=0$,dimensionless cosmological parameter $\Lambda=3.6\times 10^{-33}$ and mass $m=1$. (d) Plot of $I_6=0$ for $a=0.52,0.9,0.95,0.9939,0.99616$ and $\Lambda=3.6\times 10^{-33},m=1$. The enclosed area decreases for increasing Kerr parameter $a$. }\label{ContourPlotsgradI6curva}
\end{figure}

\begin{figure}[ptbh]
\centering
  \begin{subfigure}[b]{.60\linewidth}
    \centering
    \includegraphics[width=.99\textwidth]{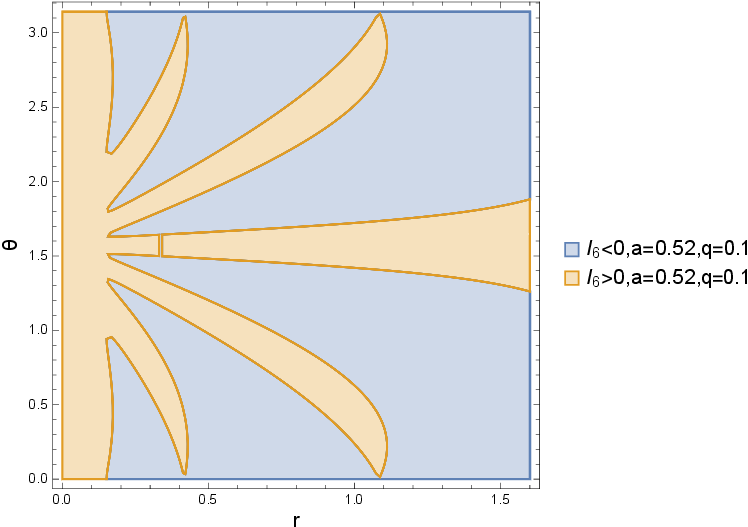}
    \caption{Region plot of $I_6$ for KN BH.}\label{RegI6kna052q01}
  \end{subfigure}%
  \begin{subfigure}[b]{.60\linewidth}
    \centering
    \includegraphics[width=.99\textwidth]{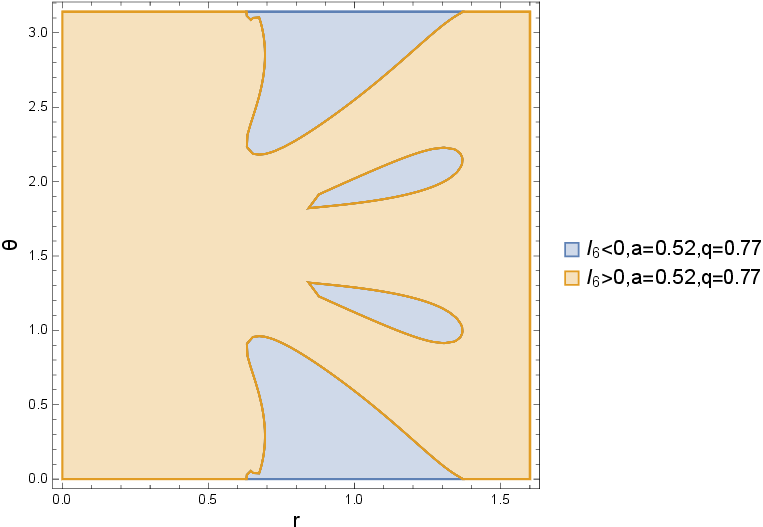}
    \caption{Region plot of $I_6$ for KN BH.}\label{RegiI6kna052q077}
  \end{subfigure}\\
  \begin{subfigure}[b]{.60\linewidth}
    \centering
    \includegraphics[width=.99\textwidth]{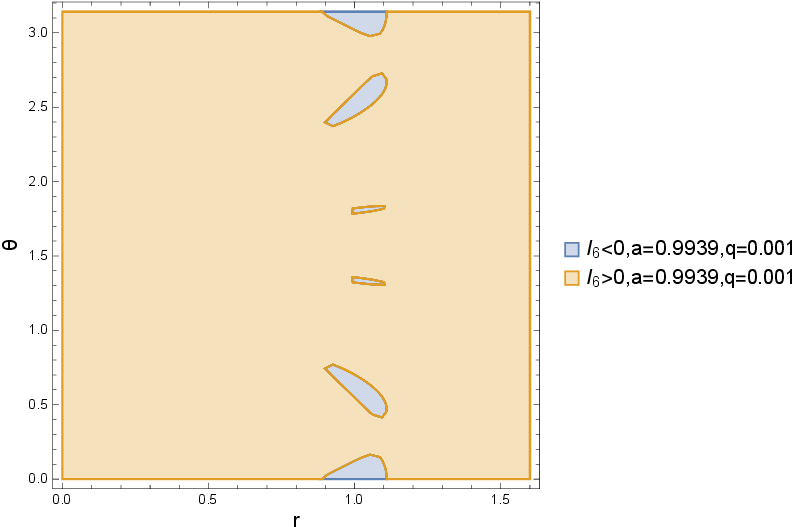}
    \caption{Region plot of $I_6$ for KN BH.}\label{RegiI6a09939q0001}
  \end{subfigure}%
  \begin{subfigure}[b]{.60\linewidth}
    \centering
    \includegraphics[width=.99\textwidth]{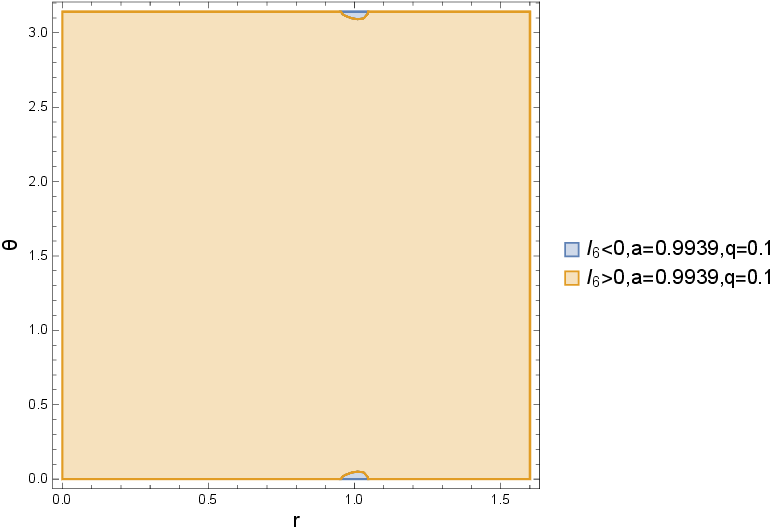}
    \caption{Region plot of $I_6$ for KN BH.}\label{regplotI6a09939q01}
  \end{subfigure}%
  \caption{Region plots of $I_6$ the norm of the gradient vector $l_{\mu}=-\nabla_{\mu}I_2$, for the Kerr-Newman black hole. (a) for spin parameter $a=0.52$, charge $q=0.1,m=1$. (b) For spin  $a=0.52$, charge $q=0.77,m=1$. (c) For high spin  $a=0.9939$, electric charge $q=0.001$,  and mass $m=1$. (d) for $a=0.9939,q=0.1,m=1$. }\label{RegionplotsofsignI6}
\end{figure}

In the following theorem we computed the invariant $I_7\equiv k_{\mu}l^{\mu}$:
\begin{theorem}
The closed form analytic expression for the invariant $I_7$ in the Kerr-(anti-)de Sitter spacetime is given by:
\begin{align}
&I_7=-\frac{27648 m^{4} r \cos \! \left(\theta \right) a}{\left(r^{2}+a^{2} \cos \! \left(\theta \right)^{2}\right)^{15}}\left(7 \cos \! \left(\theta \right)^{6} a^{6}-35 a^{4} \cos \! \left(\theta \right)^{4} r^{2}+21 \cos \! \left(\theta \right)^{2} a^{2} r^{4}-r^{6}\right) \nonumber \\
&\times \left(\cos \! \left(\theta \right)^{6} a^{6}-21 a^{4} \cos \! \left(\theta \right)^{4} r^{2}+35 \cos \! \left(\theta \right)^{2} a^{2} r^{4}-7 r^{6}\right) \nonumber \\
&\times\left(\Lambda  \cos \! \left(\theta \right)^{4} a^{4}-\Lambda  \cos \! \left(\theta \right)^{2} a^{4}-\Lambda  a^{2} r^{2}-\Lambda  r^{4}+3 a^{2} \cos \! \left(\theta \right)^{2}-6 m r +3 r^{2}\right).
\label{difinsieben}
\end{align}
\end{theorem}

In Fig.\ref{TRIADI17} we display three-dimensional plots of the curvature invariant $I_7$ as a function of the Boyer-Lindquist coordinates $r$ and $\theta$, for three sets of values for the spin,cosmological constant, electric charge and mass of the black hole.

\begin{figure}[ptbh]
\centering
  \begin{subfigure}[b]{.60\linewidth}
    \centering
    \includegraphics[width=.99\textwidth]{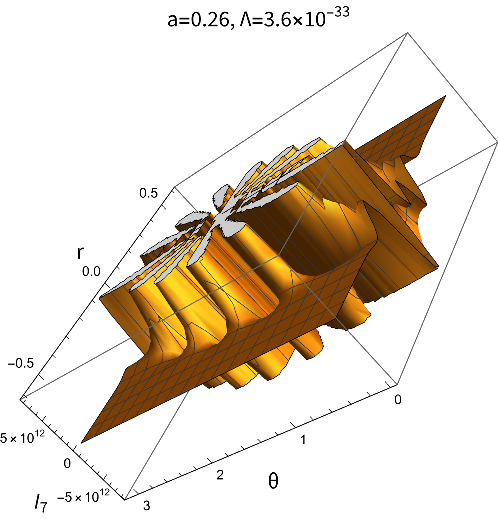}
    \caption{ 3D plot of  $I_7$.}\label{graphI7a026q0Lobs}
  \end{subfigure}%
  \begin{subfigure}[b]{.60\linewidth}
    \centering
    \includegraphics[width=.99\textwidth]{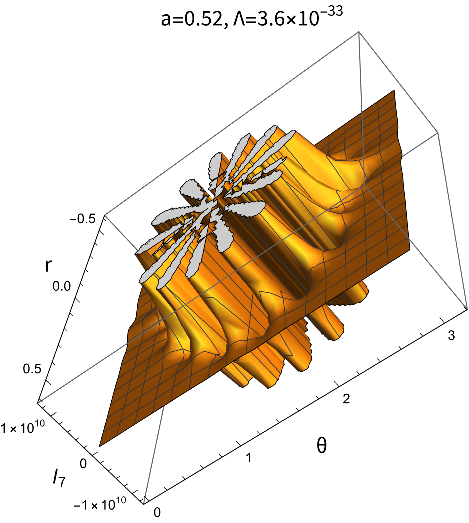}
    \caption{3D plot of $I_7$}\label{graphI7a052Lobs}
  \end{subfigure}\\
  \begin{subfigure}[b]{.60\linewidth}
    \centering
    \includegraphics[width=.99\textwidth]{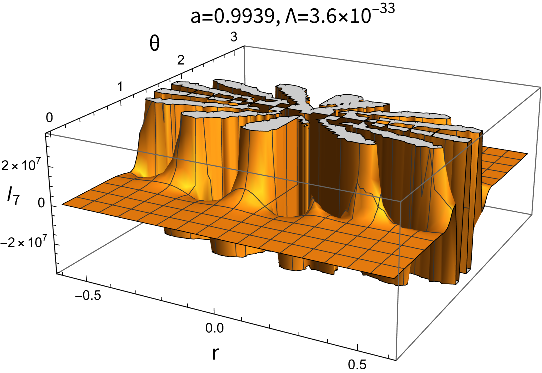}
    \caption{3D Plot of $I_7$.}\label{grafosI7a09939Lobs}
  \end{subfigure}%
  \caption{3D plots of the differential invariant,  $I_7$, plotted as a function of the Boyer-Lindquist coordinates $r$ and $\theta$ . (a) for spin parameter $a=0.26$, charge $q=0$, $\Lambda=3.6\times 10^{-33}$,$m=1$. (b) For spin  $a=0.52$, charge $q=0$,$\Lambda=3.6\times 10^{-33}$, $m=1$. (c) For high spin  $a=0.9939$, electric charge $q=0$,$\Lambda=3.6\times 10^{-33}$  and mass $m=1$.}\label{TRIADI17}
\end{figure}
We also note from (\ref{difinsieben}) that the invariant $I_7$ vanishes on the boundary of the \textit{ergosphere} region.

 In \cite{LakeZwei}, the following syzygy was discovered for the Kerr metric:
  \begin{equation}
  I_6-I_5+\frac{12}{5}(I_1I_3-I_2I_4)=0.
  \label{erstesyzygie}
  \end{equation}
  We have discovered. using our explicit algebraic expressions for the curvature invariants, that Eqn.(\ref{erstesyzygie}) \textit{also} holds for the Kerr-(anti-) de Sitter black hole.
Moreover, we find that the following syzygy is valid for the case of the Kerr-(anti-)de Sitter spacetime \footnote{This syzygy was discovered in the Kerr case in \cite{PageD}. Our result is that this syzygy is also satisfied for the Kerr-(anti-)de Sitter black hole. }:
\begin{align}
I_7=\frac{6}{5}(I_1I_4+I_2I_3)
\label{zweisyzygi}
\end{align}
The syzygies (\ref{erstesyzygie}) and (\ref{zweisyzygi}) may be expressed as the real and imaginary part of the complex syzygy \cite{PageD}:
\begin{equation}
\nabla_{\mu}(I_1+iI_2)\nabla^{\mu}(I_1+iI_2)=\frac{12}{5}(I_1+iI_2)(I_3+iI_4).
\end{equation}
On the other hand, we computed the following invariant for the Kerr-(anti-)de Sitter black hole:
\begin{align}
(I_1+i I_2)^4 (I_3-i I_4)^3&=\frac{73383542784000 m^{14}}{\left(r^{2}+a^{2} \cos \! \left(\theta \right)^{2}\right)^{27}} \Biggl(\Lambda  \cos \! \left(\theta \right)^{4} a^{4}-\Lambda  \cos \! \left(\theta \right)^{2} a^{4}-\Lambda  a^{2} r^{2}\nonumber \\
&-\Lambda  r^{4}+3 a^{2} \cos \! \left(\theta \right)^{2}-6 m r +3 r^{2}\Biggr)^{3},
\end{align}
which is a purely real expression.
Thus,
\begin{align}
\Im((I_1+i I_2)^4 (I_3-i I_4)^3)=0.
\end{align}
This is equivalent with the syzygy:
\begin{align}
4I_1I_2I_3(I_3^2-3I_4^2)(I_1^2-I_2^2)=I_4(3I_3^2-I_4^2)(-6I_1^2I_2^2+I_2^4+I_1^4).
\label{dreisyzygy}
\end{align}

\begin{theorem}\label{difanalQ1}
The exact analytic expression for the differential curvature invariant $Q_1$ for the Kerr-(anti-) de Sitter black hole is the following:
\begin{align}
Q_1=-\frac{ \left(\Lambda  \cos \! \left(\theta \right)^{4} a^{4}-\Lambda  \cos \! \left(\theta \right)^{2} a^{4}-\Lambda  a^{2} r^{2}-\Lambda  r^{4}+3 a^{2} \cos \! \left(\theta \right)^{2}-6 m r +3 r^{2}\right) \left(a^{2} \cos \! \left(\theta \right)^{2}-r^{2}\right)}{3 \left(r^{2}+a^{2} \cos \! \left(\theta \right)^{2}\right)^{\frac{3}{2}} m}.
\label{Q1KerrdeSitter}
\end{align}
\end{theorem}
\begin{corollary}\label{outerergosurfaceQ1det}
The invariant $Q_1$ vanishes on the boundary of the ergosphere region and at $r=\pm a\cos(\theta)$. $Q_1$ is strictly positive outside the outer ergosurface, vanishes at the ergosurface, and then becomes negative as soon as we cross it. The surfaces associated with the roots of $Q_1$ at the inner ergosurface and at $r=\pm a\cos(\theta)$, for the observed $\Lambda$ lie strictly within the outer ergosurface regardless of the values of $m$ and $a$. Thus, these additional roots of $Q_1$ do not affect its capability to detect the outer ergosurface. We conclude that $Q_1$ is a very convenient invariant to use for detecting the ergosurface in Kerr-de Sitter spacetime.
\end{corollary}
\begin{corollary}
For $\Lambda=0$, Eqn.(\ref{Q1KerrdeSitter}) reduces to:
\begin{equation}
Q_1=-\frac{ \left(3 a^{2} \cos \! \left(\theta \right)^{2}-6 m r +3 r^{2}\right) \left(a^{2} \cos \! \left(\theta \right)^{2}-r^{2}\right)}{3 \left(r^{2}+a^{2} \cos \! \left(\theta \right)^{2}\right)^{\frac{3}{2}} m}.
\end{equation}
This agrees with the result for the invariant $Q_1$ for the Kerr black hole, derived in \cite{LakeZwei}.
\end{corollary}

\begin{theorem}\label{difQ3inv}
The exact explicit algebraic expression for the differential curvature invariant $Q_3$ for the Kerr-(anti-) de Sitter black hole is the following:
\begin{align}
Q_3=\frac{-\Lambda  \cos \! \left(\theta \right)^{4} a^{4}-3 a^{2} \cos \! \left(\theta \right)^{2}+\Lambda  \cos \! \left(\theta \right)^{2} a^{4}-\Lambda  a^{2} r^{2}+6 a^{2}-\Lambda  r^{4}-6 m r +3 r^{2}}{6 m \sqrt{r^{2}+a^{2} \cos \! \left(\theta \right)^{2}}}.
\label{Q3KerradS}
\end{align}
\end{theorem}

\begin{corollary}
For zero cosmological constant, Eqn.(\ref{Q3KerradS}) reduces to:
\begin{align}
Q_3=\frac{-a^{2} \cos \! \left(\theta \right)^{2}+2 a^{2}-2 m r +r^{2}}{2 m \sqrt{r^{2}+a^{2} \cos \! \left(\theta \right)^{2}}}.
\end{align}
This agrees with the result for the local invariant $Q_3$ for the Kerr black hole, derived in \cite{LakeZwei}.
\end{corollary}

Returning to the syzygies (\ref{erstesyzygie}),(\ref{zweisyzygi}) and  (\ref{dreisyzygy}) we note the following:
\begin{proposition}\label{tritisyzygiaLkerr}
We have checked that equations (\ref{erstesyzygie}),(\ref{zweisyzygi}) and (\ref{dreisyzygy}) do not hold in the cases of the Kerr-Newman-(anti-)de Sitter and Kerr-Newman black holes.
\end{proposition}
We conclude that the syzygies for the cases of Proposition \ref{tritisyzygiaLkerr} remain to be found.

Abdelqader and Lake defined the following dimensionless invariant $\chi$, in order to construct an invariant measure of the \lq Kerrness \rq of a spacetime locally \cite{LakeZwei}:
\begin{equation}
\chi\equiv \frac{I_6-I_5+\frac{12}{5}(I_1I_3-I_2I_4)}{(I_1^2+I_2^2)^{5/4}}.
\end{equation}

Indeed, we computed this invariant for the Kerr-Newman-(anti-) de Sitter black hole.
We present our result for the Kerr-Newman black hole:

\begin{theorem}\label{kerrdeviation}
The invariant $\chi$ for the Kerr-Newman black hole takes the form:
\begin{align}
&\chi\overset{\Lambda=0}{=}\frac{44\sqrt{3}\, q^{2}}{5 \left(\cos \! \left(\theta \right)^{2} a^{2} m^{2}+\left(r m -q^{2}\right)^{2}\right)^{\frac{5}{2}} \left(r^{2}+a^{2} \cos \! \left(\theta \right)^{2}\right)^{4}} \Biggl(a^{12} m^{2} \left(r m -\frac{19 q^{2}}{66}\right) \cos \! \left(\theta \right)^{12}-2 \Biggl[-\frac{7 q^{6}}{132}\nonumber \\
&-\frac{17 m \left(m -\frac{176 r}{17}\right) q^{4}}{132}-\frac{3 m^{2} \left(a^{2}-\frac{25}{9} r m +\frac{427}{18} r^{2}\right) q^{2}}{11}+m^{3} r \left(a^{2}-r m +8 r^{2}\right)\Biggr] a^{10} \cos \! \left(\theta \right)^{10}+30 a^{8} \Biggl(-\frac{q^{8}}{396}\nonumber \\
&+\left(-\frac{1}{165} a^{2}+\frac{79}{990} r m -\frac{17}{180} r^{2}\right) q^{6}+\frac{9 \left(a^{2}-\frac{1079}{324} r m +\frac{143}{27} r^{2}\right) m r \,q^{4}}{55}-\frac{9 m^{2} r^{2} \left(a^{2}-\frac{71}{45} r m +\frac{307}{108} r^{2}\right) q^{2}}{11}\nonumber \\
&+m^{3} r^{3} \left(a^{2}-r m +\frac{19}{10} r^{2}\right)\Biggr) \cos \! \left(\theta \right)^{8}-84 r^{2} \Biggl[-\frac{q^{8}}{36}+\left(-\frac{2}{33} a^{2}+\frac{7}{22} r m -\frac{41}{396} r^{2}\right) q^{6}\nonumber \\
&+\frac{6 m \left(a^{2}-\frac{475}{216} r m +\frac{11}{9} r^{2}\right) r \,q^{4}}{11}-\frac{15 m^{2} \left(a^{2}-\frac{61}{45} r m +\frac{89}{90} r^{2}\right) r^{2} q^{2}}{11}+m^{3} r^{3} \left(a^{2}-r m +\frac{6}{7} r^{2}\right)\Biggr] a^{6} \cos \! \left(\theta \right)^{6}\nonumber \\
&+60 r^{4} \Biggl(-\frac{7 q^{8}}{66}+\left(-\frac{7}{33} a^{2}+\frac{392}{495} r m -\frac{287}{1980} r^{2}\right) q^{6}+\frac{63 \left(a^{2}-\frac{601}{324} r m +\frac{121}{189} r^{2}\right) m r \,q^{4}}{55}\nonumber \\
&-\frac{21 m^{2} \left(a^{2}-\frac{401}{315} r m +\frac{307}{504} r^{2}\right) r^{2} q^{2}}{11}+m^{3} r^{3} \left(a^{2}-r m +\frac{7}{12} r^{2}\right)\Biggr) a^{4} \cos \! \left(\theta \right)^{4}-10 \Biggl[-\frac{91 q^{8}}{330}\nonumber \\
&+\left(-\frac{28}{55} a^{2}+\frac{263}{165} r m -\frac{17}{60} r^{2}\right) q^{6}+\frac{108 m \left(a^{2}-\frac{2213}{1296} r m +\frac{44}{81} r^{2}\right) r \,q^{4}}{55}-\frac{27 m^{2} \left(a^{2}-\frac{167}{135} r m +\frac{427}{810} r^{2}\right) r^{2} q^{2}}{11}\nonumber \\
&+m^{3} r^{3} \left(a^{2}-r m +\frac{28}{55} r^{2}\right)\Biggr] r^{6} a^{2} \cos \! \left(\theta \right)^{2}\nonumber \\
&+\frac{2 \left(r m -q^{2}\right)^{2} r^{8} \left(-\frac{7 q^{4}}{12}+\left(-a^{2}+\frac{5}{3} r m -\frac{7}{12} r^{2}\right) q^{2}+m r \left(a^{2}-r m +\frac{1}{2} r^{2}\right)\right)}{11}\Biggr).
\label{kerrnessKN}
\end{align}
\end{theorem}

In Fig.\ref{ContourKerrnessKN6curva} we display contour plots for the scalar invariant $\chi$, eqn.(\ref{kerrnessKN}), for the Kerr-Newman black hole.

\begin{figure}[ptbh]
\centering
  \begin{subfigure}[b]{.60\linewidth}
    \centering
    \includegraphics[width=.99\textwidth]{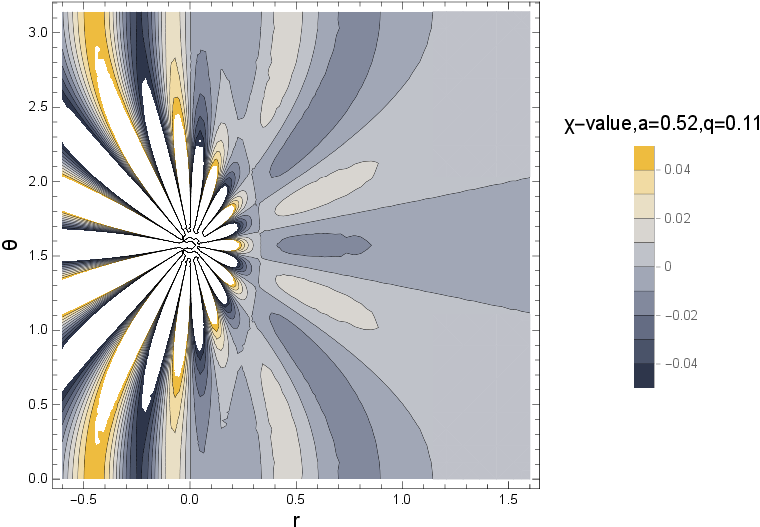}
    \caption{ Contour plot of curvature invariant $\chi$.}\label{Contourchisa052}
  \end{subfigure}%
  \begin{subfigure}[b]{.60\linewidth}
    \centering
    \includegraphics[width=.99\textwidth]{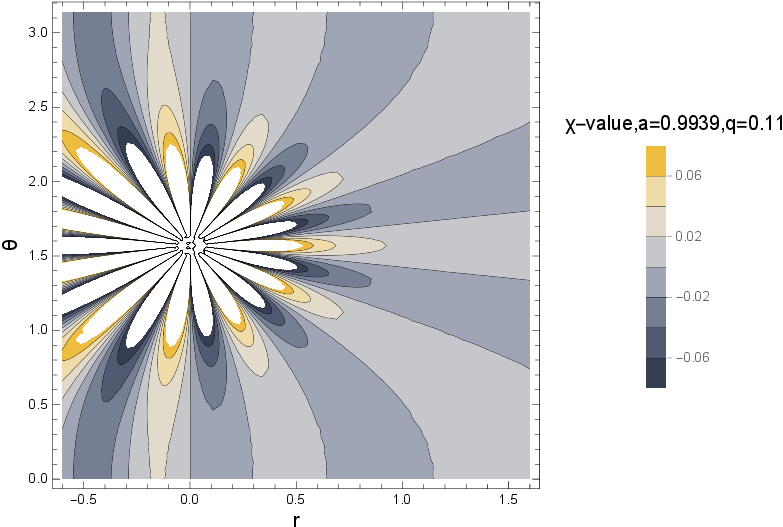}
    \caption{Contour plot of curvature invariant $\chi$.}\label{Contourchi09939}
  \end{subfigure}\\
  \begin{subfigure}[b]{.60\linewidth}
    \centering
    \includegraphics[width=.99\textwidth]{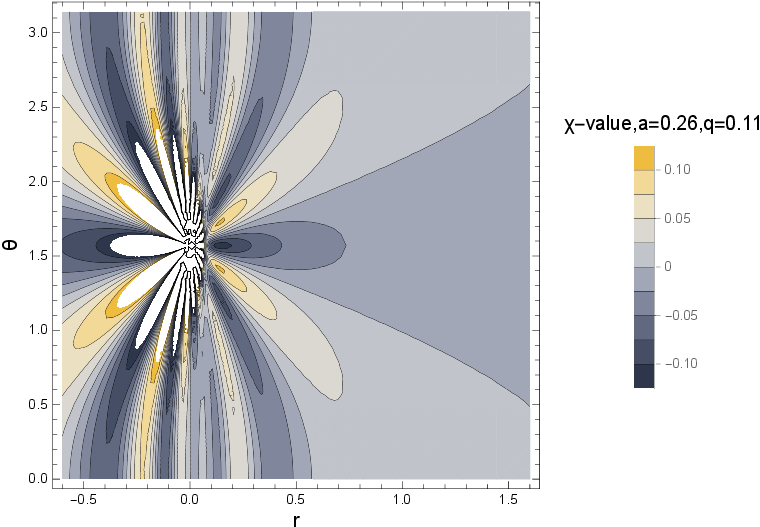}
    \caption{Contour Plot of curvature invariant $chi$.}\label{Contourchi026}
  \end{subfigure}%
  \caption{Contour plots of the curvature $\chi$ invariant, Eqn.(\ref{kerrnessKN}), for the Kerr-Newman black hole. (a) for spin parameter $a=0.52$, charge $q=0.11$,dimensionless cosmological parameter $\Lambda=0,m=1$. (b) For spin  $a=0.9939$, charge $q=0.11$,dimensionless cosmological parameter $\Lambda=0,m=1$. (c) For low spin  $a=0.26$, electric charge $q=0.11$,dimensionless cosmological parameter $\Lambda=0$ and mass $m=1$.}\label{ContourKerrnessKN6curva}
\end{figure}

\section{Differential invariants for accelerating Kerr-Newman black holes in (anti-)de Sitter spacetime}\label{epitaxydiaforikesanal}
\subsection{Accelerating and rotating charged black holes with non-zero cosmological constant }

 The Pleba\'{n}ski-Demia\'{n}ski metric covers a large family of solutions which include  the  physically most significant case: that of an accelerating, rotating and charged black hole with a non-zero cosmological constant \cite{PlebanskiDemianski}. We focus on the following metric that describes an accelerating Kerr-Newman black hole in (anti-)de Sitter spacetime \cite{GrifPod},\cite{PodolskyGrif}:
 \begin{align}
 \mathrm{d}s^{2}  & =\frac{1}{\Omega^2}\Biggl\{-\frac{Q}{\rho^2}\left[\mathrm{d}t-a\sin^2\theta\mathrm{d}\phi\right]^2+
 \frac{\rho^2}{Q}\mathrm{d}r^2+\frac{\rho^2}{P}\mathrm{d}\theta^2\nonumber \\
 &+\frac{P}{\rho^2}\sin^2\theta\left[a\mathrm{d}t-(r^2+a^2)\mathrm{d}\phi\right]^2\Biggr\},
 \label{epitaxmelaniopi}
 \end{align}
 where
 \begin{align}
 \Omega=&1-\alpha r \cos\theta,\\
 P=&1-2\alpha m \cos\theta+(\alpha^2(a^2+q^2)+\frac{1}{3}\Lambda a^2)\cos^2\theta,\\
 Q=&((a^2+q^2)-2mr+r^2)(1-\alpha^2r^2)\nonumber \\
 &-\frac{1}{3}\Lambda(a^2+r^2)r^2\label{radiieventpolyQ},
 \end{align}
 and $\alpha$ is the acceleration of the black hole.

 The metric (\ref{epitaxmelaniopi}) becomes singular at the roots of $\Omega, \rho^2,Q,P$. Some of them are pseudosingularities (mere coordinate singularities) while others are true (curvature) singularities detected by the curvature invariants.
 We shall discuss the influence of the acceleration parameter $\alpha$ on these singularities.

 $\Omega$ becomes zero if:
 \begin{equation}
 r=\frac{1}{\alpha\cos\theta}.
 \label{ConformalBound}
 \end{equation}
As the metric blows up if $\Omega\rightarrow 0$, Eq. (\ref{ConformalBound}) determines the boundary of the spacetime, thus we have to restrict to regions where $\Omega>0$.
For $\alpha=0$ there is no restriction because $\Omega=1$.

$\rho^2$ becomes zero at the ring singularity:
\begin{equation}
r=0\;\;\;\text{and}\;\;\;\cos\theta=0
\end{equation}
The ring singularity $r=0,\theta=\pi/2$ is a curvature singularity for ($m\not =0$) and is unaffected by $\alpha$.
The real roots of $Q$, yield coordinate singularities which correspond to the up to $4$ horizons of the spacetime. We investigated these pseudosingularities in \cite{KraniotisCurvature}.

In general, at the roots of $P$ would be coordinate singularities, too. These would indicate further horizons where the vector field $\partial_{\theta}$  would change its causal character, just as the vector field $\partial_r$ does at the roots of $Q$. Nevertheless, since these horizons would lie on cones $\theta={\rm constant}$ instead of on spheres $r={\rm constant}$ would be hardly of any physical relevance \cite{arneGrezen}.  The equation $P=0$ is a quadratic equation for $\cos(\theta)$:
\begin{align}
P=0\Leftrightarrow a_1\cos^2(\theta)+b_1 \cos(\theta)+c_1=0,
\label{polarexpress}
\end{align}
where $a_1\equiv(\alpha^2(a^2+q^2)+\frac{1}{3}\Lambda a^2),b_1\equiv-2\alpha m,c_1=1$.
The roots of the quadratic equation are given by the formula:
\begin{equation}
\cos\theta_{\pm}=\frac{-b_1\pm\sqrt{b_1^2-4 a_1}}{2 a_1}.
\label{polarhorizons}
\end{equation}
If the radicand (i.e. the discriminant) in (\ref{polarhorizons}) is negative $P\not =0$ is guaranteed $\forall \theta\in\mathbb{R}$. In fact in this case, for $\Lambda>0$ we have $P>0$ for $\theta\in[0,\pi]$, since the leading coefficient in (\ref{polarexpress}) is positive.
For positive radicand (discriminant) from the theory of quadratic algebraic equations we know that $P$ has the opposite sign of $a_1$ for values of $\cos\theta$ between the two real roots and has the same sign as the sign of the leading coefficient $a_1$ for values of $\cos\theta$ outside the two roots
\footnote{Nevertheless for the values of the physical black hole parameters we investigated the real roots of $P$ occur for $\theta\not\in\mathbb{R}$.}.

 For $\Lambda=0$, if $m^2\geq a^2+q^2$ the expression for $Q$ factorises as:
 \begin{equation}
 Q=(r_{-}-r)(r_{+}-r)(1-\alpha^2 r^2),
 \end{equation}
 where
 \begin{equation}
 r_{\pm}=m\pm\sqrt{m^2-a^2-q^2}.
 \end{equation}
 The expressions for the radii $r_{\pm}$ are identical to those for the location of the outer and inner horizons of the nonaccelerating Kerr-Newman black hole. However, in the present case there is another horizon at $r=\alpha^{-1}$ known in the context of the $C-$ metric as an acceleration horizon.
 When $\Lambda \not =0$, the location of all horizons is modified.

Subsequently, we compute explicit algebraic expressions for the Karlhede and   Abdelqade-Lake local curvature invariants, for the metric (\ref{epitaxmelaniopi}), with  Maple\textsuperscript{TM}2021 \footnote{The symbolic computation in this case is quite demanding. Despite the formidable complexity of the tensorial and differentiation operations involved, the symbolic computation of these differential curvature invariants runs smoothly in a modern 8GB RAM laptop.}.

 We start our computations with the calculation of the third-order Karlhede curvature invariant. In the following theorem we computed an explicit algebraic expression for the $\mathfrak{K}$ invariant for accelerating rotating black hole with $\Lambda\not=0$.
 \begin{theorem}\label{KarlEpitaKerrdS}
 The analytic computation of the  Karlhede curvature invariant for accelerating Kerr black holes in (anti)-de Sitter spacetime yields the following explicit algebraic expression:
 \begin{align}
 \mathfrak{K}=\frac{-720 m^2(\alpha r \cos(\theta)-1)^6}{\left(r^{2}+a^{2} \cos \! \left(\theta \right)^{2}\right)^{9}}\sum_{i=0}^{12}k_i\cos \! \left(\theta \right)^{i},
 \label{marvelKarlhedeaccelKerrdS}
 \end{align}
 \end{theorem}
where the coefficients $k_i$ are given below:
\begin{align}
k_0&=-\Biggl[r \alpha^{2} \left(\alpha^{2}+\frac{\Lambda}{3}\right) a^{4}+\left(\frac{r^{3} \alpha^{2} \Lambda}{3}-2 m \,r^{2} \alpha^{4}+\left(2 \alpha^{2}-\frac{\Lambda}{3}\right) r +2 \alpha^{2} m \right) a^{2}+2 \alpha^{2} m \,r^{2}\nonumber \\
&-\frac{\Lambda  r^{3}}{3}-2 m +r \Biggr] r^{9},\\
k_1&=+2 \alpha  \left(\left(\alpha^{2} m -\frac{8 r \Lambda}{3}\right) a^{4}+\left(-\alpha^{4} m \,r^{4}+20 \alpha^{2} m \,r^{2}-\frac{8}{3} \Lambda  r^{3}-17 m \right) a^{2}+\alpha^{2} m \,r^{4}\right) r^{8},\\
k_2&=+27 \Biggl(r \alpha^{2} \left(\alpha^{2}+\frac{\Lambda}{3}\right) a^{6}+\left(\frac{\alpha^{4} \left(\alpha^{2}+\frac{\Lambda}{3}\right) r^{5}}{27}+\frac{28 r^{3} \alpha^{2} \Lambda}{81}-\frac{88 m \,r^{2} \alpha^{4}}{27}+\left(2 \alpha^{2}-\frac{\Lambda}{3}\right) r +\frac{88 \alpha^{2} m}{27}\right) a^{4}\nonumber \\
&+\left(\frac{2 \alpha^{2} \left(\alpha^{2}+\frac{\Lambda}{2}\right) r^{5}}{27}-\frac{40 \alpha^{4} m \,r^{4}}{27}-\frac{28 \Lambda  r^{3}}{81}+\frac{32 \alpha^{2} m \,r^{2}}{9}+r -\frac{56 m}{27}\right) a^{2}+\frac{\alpha^{2} r^{5}}{27}\Biggr) r^{7},\\
k_3&=-56 \alpha  \Biggl(\left(\alpha^{2} m -\frac{4 r \Lambda}{7}\right) a^{6}+\left(-\frac{11}{7} \alpha^{4} m \,r^{4}+\frac{241}{28} \alpha^{2} m \,r^{2}-\frac{2}{3} \Lambda  r^{3}-5 m \right) a^{4}\nonumber \\
&-\frac{\left(\alpha^{4} m \,r^{4}-48 \alpha^{2} m \,r^{2}+\frac{8}{3} \Lambda  r^{3}+17 m \right) r^{2} a^{2}}{28}+\frac{\alpha^{2} m \,r^{6}}{28}\Biggr) r^{6},\\
k_4&=-42 \Biggl[r \alpha^{2} \left(\alpha^{2}+\frac{\Lambda}{3}\right) a^{6}+\left(\frac{9 \alpha^{4} \left(\alpha^{2}+\frac{\Lambda}{3}\right) r^{5}}{14}+\frac{23 r^{3} \alpha^{2} \Lambda}{42}-\frac{26 m \,r^{2} \alpha^{4}}{3}+\left(2 \alpha^{2}-\frac{\Lambda}{3}\right) r +\frac{26 \alpha^{2} m}{3}\right) a^{4}\nonumber \\
&+\left(\frac{m \,r^{6} \alpha^{6}}{21}+\frac{9 \alpha^{2} \left(\alpha^{2}+\frac{\Lambda}{2}\right) r^{5}}{7}-\frac{241 \alpha^{4} m \,r^{4}}{21}-\frac{23 \Lambda  r^{3}}{42}+\frac{310 \alpha^{2} m \,r^{2}}{21}+r -\frac{10 m}{3}\right) a^{2}\nonumber \\
&-\frac{17 \alpha^{2} r^{4} \left(\alpha^{2} m \,r^{2}-m -\frac{27}{34} r \right)}{21}\Biggr] a^{2} r^{5},\\
k_5&=+140 \Biggl(a^{6} \alpha^{2} m +\left(-\frac{13}{5} \alpha^{4} m \,r^{4}+\frac{38}{5} \alpha^{2} m \,r^{2}-\frac{8}{35} \Lambda  r^{3}-\frac{13}{5} m \right) a^{4}+\Biggl[-\frac{22}{35} \alpha^{4} m \,r^{6}+\frac{31}{7} \alpha^{2} m \,r^{4}-\frac{8}{35} \Lambda  r^{5}\nonumber \\
&-2 m \,r^{2}\Biggr] a^{2}+\frac{2 \alpha^{2} m \,r^{6}}{5}\Biggr) \alpha  a^{2} r^{4},\\
k_6&=-42 \Biggl(r \alpha^{2} \left(\alpha^{2}+\frac{\Lambda}{3}\right) a^{6}+\left(\left(-\alpha^{6}-\frac{1}{3} \Lambda  \,\alpha^{4}\right) r^{5}+\frac{20 m \,r^{2} \alpha^{4}}{3}+\left(2 \alpha^{2}-\frac{\Lambda}{3}\right) r -\frac{20 \alpha^{2} m}{3}\right) a^{4}\nonumber \\
&+\left(-\frac{4 m \,r^{6} \alpha^{6}}{3}+\left(-2 \alpha^{4}-\Lambda  \,\alpha^{2}\right) r^{5}+\frac{76 \alpha^{4} m \,r^{4}}{3}-\frac{76 \alpha^{2} m \,r^{2}}{3}+r +\frac{4 m}{3}\right) a^{2}\nonumber \\
&+\frac{20 \alpha^{2} \left(\alpha^{2} m \,r^{2}-m -\frac{3}{20} r \right) r^{4}}{3}\Biggr) a^{4} r^{3}
\end{align}

\begin{align}
k_7&=-56 \alpha  \Biggl(\left(\alpha^{2} m +\frac{4 r \Lambda}{7}\right) a^{6}+\left(-5 \alpha^{4} m \,r^{4}+\frac{155}{14} \alpha^{2} m \,r^{2}+\frac{4}{7} \Lambda  r^{3}-\frac{11}{7} m \right) a^{4}\nonumber \\
&-\frac{13 \left(\alpha^{4} r^{4}-\frac{38}{13} \alpha^{2} r^{2}+1\right) m \,r^{2} a^{2}}{2}+\frac{5 \alpha^{2} m \,r^{6}}{2}\Biggr) a^{4} r^{2},\\
k_8&=+27 \Biggl[r \alpha^{2} \left(\alpha^{2}+\frac{\Lambda}{3}\right) a^{6}+\left(\frac{14 \alpha^{4} \left(\alpha^{2}+\frac{\Lambda}{3}\right) r^{5}}{9}+\frac{23 r^{3} \alpha^{2} \Lambda}{27}+\frac{34 m \,r^{2} \alpha^{4}}{27}+\left(2 \alpha^{2}-\frac{\Lambda}{3}\right) r -\frac{34 \alpha^{2} m}{27}\right) a^{4}\nonumber \\
&+\left(-\frac{140 m \,r^{6} \alpha^{6}}{27}+\frac{28 \alpha^{2} \left(\alpha^{2}+\frac{\Lambda}{2}\right) r^{5}}{9}+\frac{620 \alpha^{4} m \,r^{4}}{27}-\frac{23 \Lambda  r^{3}}{27}-\frac{482 \alpha^{2} m \,r^{2}}{27}+r +\frac{2 m}{27}\right) a^{2}\nonumber \\
&+\frac{364 \alpha^{2} r^{4} \left(\alpha^{2} m \,r^{2}-m +\frac{3}{26} r \right)}{27}\Biggr] a^{6} r,\\
k_9&=+2 \alpha  a^{6} \Biggl(\left(\alpha^{2} m +\frac{8 r \Lambda}{3}\right) a^{6}+\left(-17 \alpha^{4} m \,r^{4}+48 \alpha^{2} m \,r^{2}+\frac{56}{3} \Lambda  r^{3}-m \right) a^{4}+\Biggl[-140\alpha^{4} m \,r^{6}+241 \alpha^{2} m \,r^{4}\nonumber \\
&+16 \Lambda  r^{5}-44 m \,r^{2}\Biggr] a^{2}+28 \alpha^{2} m \,r^{6}\Biggr),\\
k_{10}&=-\Biggl(\alpha^{2} \left(\alpha^{2}+\frac{\Lambda}{3}\right) a^{6}+\left(\left(27 \alpha^{6}+9 \Lambda  \,\alpha^{4}\right) r^{4}+\frac{28 \Lambda  \,\alpha^{2} r^{2}}{3}+2 \alpha^{2}-\frac{\Lambda}{3}\right) a^{4}+\Biggl[1-56 \alpha^{6} m \,r^{5}\nonumber \\
&+\left(54 \alpha^{4}+27 \Lambda  \,\alpha^{2}\right) r^{4}+96 \alpha^{4} m \,r^{3}-\frac{28 \Lambda  r^{2}}{3}-40 \alpha^{2} m r \Biggr] a^{2}+88 \alpha^{2} \left(\alpha^{2} m \,r^{2}-m +\frac{27}{88} r \right) r^{3}\Biggr) a^{8},\\
k_{11}&=-2 \alpha  a^{8} \Biggl(\left(\alpha^{2} m +\frac{8 r \Lambda}{3}\right) a^{4}+\left(-17 \alpha^{4} m \,r^{4}+20 \alpha^{2} m \,r^{2}+\frac{8}{3} \Lambda  r^{3}-m \right) a^{2}+\alpha^{2} m \,r^{4}\Biggr),\\
k_{12}&=\Biggl[\alpha^{2} \left\{\alpha^{2} \left(\alpha^{2}+\frac{\Lambda}{3}\right) r^{2}+\frac{\Lambda}{3}\right\} a^{4}+\left(-2 \alpha^{6} m \,r^{3}+\left(2 \alpha^{4}+\Lambda  \,\alpha^{2}\right) r^{2}+2 \alpha^{4} m r -\frac{\Lambda}{3}\right) a^{2}\nonumber \\
&+2 \alpha^{2} \left(\alpha^{2} m \,r^{2}-m +\frac{1}{2} r \right) r \Biggr] a^{10}.
\end{align}

The invariants $I_1$ and $I_2$ for accelerating Kerr-Newman in (anti-)de Sitter spacetime have been computed in \cite{KraniotisCurvature},eqns.(95) and (94) correspondigly in \cite{KraniotisCurvature}.
In Theorems \ref{weylcovI3epi}-\ref{sebenepitaxkerr} we present our novel explicit analytic expressions for the Abdelqade-Lake  local invariants $I_3-I_7$ for the case of accelerating Kerr black holes \footnote{We have computed  these  Abdelqade-Lake invariants for the accelerating KN(a-)dS black hole. The resulting expressions are very lengthy to reproduce them here. Nevertheless, we have use them in order to obtain a very compact expression for the invariant $Q_2$ in latter sections of this paper. }.

  \begin{theorem}\label{weylcovI3epi}
  We computed in closed analytic form the invariant $I_3\equiv\nabla_{\mu}C_{\alpha\beta\gamma\delta}\nabla^{\mu}C^{\alpha\beta\gamma\delta}$ for an accelerating Kerr black hole. Our result is:
  \begin{align}
  &\nabla_{\mu}C_{\alpha\beta\gamma\delta}\nabla^{\mu}C^{\alpha\beta\gamma\delta}\nonumber \\
  &=-\frac{720 \left(\alpha  r \cos \! \left(\theta \right)-1\right)^{7} m^{2}}
  {\left(r^{2}+a^{2} \cos \! \left(\theta \right)^{2}\right)^{9}} \Biggl(\alpha  a^{10} \left(a^{4} \alpha^{4} r -2 \alpha^{2} \left(\alpha^{2} m \,r^{2}-m -r \right) a^{2}+2 \alpha^{2} m \,r^{2}-2 m +r \right) \cos \! \left(\theta \right)^{11}\nonumber \\
  &+\Biggl(a^{6} \alpha^{4}+\left(-2 \alpha^{4} m r +2 \alpha^{2}\right) a^{4}+\left(34 \alpha^{4} m \,r^{3}-38 \alpha^{2} m r +1\right) a^{2}\nonumber \\
  &-2 m \,r^{3} \alpha^{2}\Biggr) a^{8} \cos \! \left(\theta \right)^{10}
  -2 \Biggl[\alpha^{2} \left(\frac{27 r^{3} \alpha^{2}}{2}+m \right) a^{4}+\Biggl(-28 m \,r^{4} \alpha^{4}+31 \alpha^{2} m \,r^{2}\nonumber \\
  &+27 r^{3} \alpha^{2}-m \Biggr) a^{2}+44 \alpha^{2} m \,r^{4}-43 m \,r^{2}+\frac{27 r^{3}}{2}\Biggr] \alpha  a^{8} \cos \! \left(\theta \right)^{9}-27 \Biggl(a^{6} r \alpha^{4}-\frac{22 \left(\alpha^{2} m \,r^{2}+\frac{17}{11} m -\frac{27}{11} r \right) \alpha^{2} a^{4}}{27}\nonumber \\
  &+\left(\frac{280}{27} m \,r^{4} \alpha^{4}-\frac{394}{27} \alpha^{2} m \,r^{2}+\frac{2}{27} m +r \right) a^{2}-\frac{56 \alpha^{2} m \,r^{4}}{27}\Biggr) a^{6} r \cos \! \left(\theta \right)^{8}+56 \Biggl(\alpha^{2} \left(\frac{3 r^{3} \alpha^{2}}{4}+m \right) a^{4}\nonumber \\
  &+\left(-\frac{5}{2} m \,r^{4} \alpha^{4}+\frac{85}{14} \alpha^{2} m \,r^{2}+\frac{3}{2} r^{3} \alpha^{2}-\frac{11}{7} m \right) a^{2}+\frac{13 \left(\alpha^{2} m \,r^{2}-\frac{11}{13} m +\frac{3}{26} r \right) r^{2}}{2}\Biggr) \alpha  a^{6} r^{2} \cos \! \left(\theta \right)^{7}+42 a^{4} r^{3} \Biggl(\nonumber \\
  & a^{6} r \alpha^{4}+\frac{10 \left(\alpha^{2} m \,r^{2}-2 m +\frac{3}{5} r \right) \alpha^{2} a^{4}}{3}+\left(\frac{26}{3} m \,r^{4} \alpha^{4}-\frac{50}{3} \alpha^{2} m \,r^{2}+\frac{4}{3} m +r \right) a^{2}-\frac{10 \alpha^{2} m \,r^{4}}{3}\Biggr) \cos \! \left(\theta \right)^{6}\nonumber \\
  &-140 \Biggl[\alpha^{2} \left(-\frac{3 r^{3} \alpha^{2}}{10}+m \right) a^{4}+\left(-\frac{2}{5} m \,r^{4} \alpha^{4}+5 \alpha^{2} m \,r^{2}-\frac{3}{5} r^{3} \alpha^{2}-\frac{13}{5} m \right) a^{2}+2 r^{2} \left(\alpha^{2} m \,r^{2}-\frac{1}{2} m -\frac{3}{20} r \right)\nonumber \\
  &\Biggr] \alpha  a^{4} r^{4} \cos \! \left(\theta \right)^{5}+42 \Biggl\{a^{6} r \alpha^{4}-\frac{22 \left(\alpha^{2} m \,r^{2}-\frac{13}{11} m -\frac{3}{11} r \right) \alpha^{2} a^{4}}{3}+\left(-\frac{44}{21} m \,r^{4} \alpha^{4}+\frac{170}{21} \alpha^{2} m \,r^{2}-\frac{10}{3} m +r \right) a^{2}\nonumber \\
  &+\frac{4 \alpha^{2} m \,r^{4}}{3}\Biggr\} a^{2} r^{5} \cos \! \left(\theta \right)^{4}+56 \alpha  a^{2} r^{6} \Biggl(\alpha^{2} \left(-\frac{27 r^{3} \alpha^{2}}{56}+m \right) a^{4}+\Biggl(-\frac{1}{28} m \,r^{4} \alpha^{4}+\frac{197}{28} \alpha^{2} m \,r^{2}-\frac{27}{28} r^{3} \alpha^{2}\nonumber \\
  &-5 m \Biggr) a^{2}+\frac{17 \left(\alpha^{2} m \,r^{2}+\frac{11}{17} m -\frac{27}{34} r \right) r^{2}}{28}\Biggr) \cos \! \left(\theta \right)^{3}-27 \Biggl(a^{6} r \alpha^{4}-\frac{86 \left(\alpha^{2} m \,r^{2}-\frac{44}{43} m -\frac{27}{43} r \right) \alpha^{2} a^{4}}{27}\nonumber \\
  &+\left(-\frac{2}{27} m \,r^{4} \alpha^{4}+\frac{62}{27} \alpha^{2} m \,r^{2}-\frac{56}{27} m +r \right) a^{2}+\frac{2 \alpha^{2} m \,r^{4}}{27}\Biggr) r^{7} \cos \! \left(\theta \right)^{2}-2 \alpha  \Biggl(\alpha^{2} \left(-\frac{r^{3} \alpha^{2}}{2}+m \right) a^{4}\nonumber \\
  &+\left(19 \alpha^{2} m \,r^{2}-r^{3} \alpha^{2}-17 m \right) a^{2}+r^{2} \left(m -\frac{r}{2}\right)\Biggr) r^{8} \cos \! \left(\theta \right)+r^{9} \Biggl(a^{4} \alpha^{4} r -2 \alpha^{2} \left(\alpha^{2} m \,r^{2}-m -r \right) a^{2}\nonumber \\
  &+2 \alpha^{2} m \,r^{2}-2 m +r \Biggr)\Biggr).
  \label{idreiepitkerr}
  \end{align}
  \end{theorem}

  \begin{remark}
  The curvature invariant $I_3$ is equal with the Karlhede invariant for accelerating Kerr black holes. We have also checked that our expression Eqn.(\ref{marvelKarlhedeaccelKerrdS}) in Theorem \ref{KarlEpitaKerrdS}, for vanishing cosmological constant ($\Lambda=0$), reduces to expression Eqn.(\ref{idreiepitkerr}) of Theorem \ref{weylcovI3epi}.
  \end{remark}

\begin{theorem}
We have computed the following analytic expression for the invariant $I_4$ for the accelerating Kerr black hole:
\begin{align}
&I_4=-\frac{5760 m^{2} \left(\alpha  r \cos \! \left(\theta \right)-1\right)^{7} a}{\left(r^{2}+a^{2} \cos \! \left(\theta \right)^{2}\right)^{9}} \Biggl(\left(-\frac{1}{2} \alpha^{4} m \,r^{2}+\frac{1}{2} \alpha^{2} m \right) a^{10} \cos \! \left(\theta \right)^{11}\nonumber \\
&+\left(a^{4} r \alpha^{4}-2 \alpha^{2} \left(\alpha^{2} m \,r^{2}-m -r \right) a^{2}+\frac{5 \alpha^{2} m \,r^{2}}{2}-\frac{5 m}{2}+r \right) \alpha  a^{8} r \cos \! \left(\theta \right)^{10}\nonumber \\
&+\left(a^{6} r \alpha^{4}-\frac{3 \alpha^{2} \left(\alpha^{2} m \,r^{2}+\frac{1}{3} m -\frac{4}{3} r \right) a^{4}}{2}+\left(16 \alpha^{4} m \,r^{4}-20 \alpha^{2} m \,r^{2}+r \right) a^{2}-2 \alpha^{2} m \,r^{4}\right) a^{6} \cos \! \left(\theta \right)^{9}\nonumber \\
&-2 \Biggl[\alpha^{2} \left(3 \alpha^{2} r^{3}+m \right) a^{4}+\left(-7 \alpha^{4} m \,r^{4}+\left(\frac{41}{4} m \,r^{2}+6 r^{3}\right) \alpha^{2}-\frac{5 m}{4}\right) a^{2}+14 \alpha^{2} m \,r^{4}-13 m \,r^{2}\nonumber \\
&+3 r^{3}\Biggr] \alpha  a^{6} r \cos \! \left(\theta \right)^{8}-6 \Biggl(a^{6} r \alpha^{4}+\frac{\alpha^{2} \left(\alpha^{2} m \,r^{2}-8 m +6 r \right) a^{4}}{3}+\left(\frac{49}{6} \alpha^{4} m \,r^{4}-\frac{79}{6} \alpha^{2} m \,r^{2}+\frac{1}{3} m +r \right) a^{2}\nonumber \\
&-\frac{7 \alpha^{2} m \,r^{4}}{3}\Biggr) a^{4} r^{2} \cos \! \left(\theta \right)^{7}+14 \left(\alpha^{2} a^{4}+\left(-\alpha^{4} r^{4}+5 \alpha^{2} r^{2}-2\right) a^{2}+\frac{7 r^{4} \alpha^{2}}{2}-\frac{5 r^{2}}{2}\right) \alpha  m \,a^{4} r^{3} \cos \! \left(\theta \right)^{6}\nonumber \\
&+35 m \,a^{2} r^{4} \left(\left(\alpha^{4} r^{2}-\frac{7}{5} \alpha^{2}\right) a^{4}+\left(\frac{4}{5} \alpha^{4} r^{4}-2 \alpha^{2} r^{2}+\frac{2}{5}\right) a^{2}-\frac{2 r^{4} \alpha^{2}}{5}\right) \cos \! \left(\theta \right)^{5}-14 \Biggl[\alpha^{2} \left(-\frac{3 \alpha^{2} r^{3}}{7}+m \right) a^{4}\nonumber \\
&+\left(-\frac{\alpha^{4} m \,r^{4}}{7}+\frac{79 \left(m -\frac{12 r}{79}\right) r^{2} \alpha^{2}}{14}-\frac{7 m}{2}\right) a^{2}+\frac{8 \left(\alpha^{2} m \,r^{2}-\frac{1}{8} m -\frac{3}{8} r \right) r^{2}}{7}\Biggr] \alpha  a^{2} r^{5} \cos \! \left(\theta \right)^{4}+6 \Biggl(a^{6} r \alpha^{4}\nonumber \\
&-\frac{13 \alpha^{2} \left(\alpha^{2} m \,r^{2}-\frac{14}{13} m -\frac{6}{13} r \right) a^{4}}{3}+\left(-\frac{5}{12} \alpha^{4} m \,r^{4}+\frac{41}{12} \alpha^{2} m \,r^{2}-\frac{7}{3} m +r \right) a^{2}+\frac{\alpha^{2} m \,r^{4}}{3}\Biggr) r^{6} \cos \! \left(\theta \right)^{3}\nonumber \\
&+2 \left(\alpha^{2} \left(-\frac{\alpha^{2} r^{3}}{2}+m \right) a^{4}+\left(\left(10 m \,r^{2}-r^{3}\right) \alpha^{2}-8 m \right) a^{2}+\frac{r^{2} \left(\alpha^{2} m \,r^{2}+3 m -2 r \right)}{4}\right) \alpha  r^{7} \cos \! \left(\theta \right)^{2}\nonumber \\
&-\left(a^{4} r \alpha^{4}-\frac{5 \alpha^{2} \left(\alpha^{2} m \,r^{2}-m -\frac{4}{5} r \right) a^{2}}{2}+2 \alpha^{2} m \,r^{2}-2 m +r \right) r^{8} \cos \! \left(\theta \right)-\frac{\alpha^{3} m \,r^{11}}{2}+\frac{\alpha  m \,r^{9}}{2}\Biggr).
\end{align}
\end{theorem}
\begin{corollary}
For zero acceleration $\alpha=0$ the invariant $I_4$ takes the form:
\begin{align}
I_4&=\frac{5760 m^{2} a r \cos \! \left(\theta \right)}{\left(r^{2}+a^{2} \cos \! \left(\theta \right)^{2}\right)^{9}} \left(a^{2} \cos \! \left(\theta \right)^{2}-2 \cos \! \left(\theta \right) a r -r^{2}\right)\nonumber \\
&\times  \left(a^{2} \cos \! \left(\theta \right)^{2}+2 \cos \! \left(\theta \right) a r -r^{2}\right) \left(a^{2} \cos \! \left(\theta \right)^{2}-2 m r +r^{2}\right) \left(a^{2} \cos \! \left(\theta \right)^{2}-r^{2}\right)
\end{align}
\end{corollary}

\begin{theorem}
The analytic computation of the invariant $I_5$ for the accelerating Kerr black hole yields the result:
\begin{align}
I_5=-\frac{82944 m^{4} \left(\alpha  r \cos \! \left(\theta \right)-1\right)^{13}}{\left(r^{2}+a^{2} \cos \! \left(\theta \right)^{2}\right)^{15}}\sum_{i=0}^{17}T_i\;(\cos(\theta))^i,
\end{align}
\end{theorem}
where the coefficients $T_i$ are given below:
\begin{align}
&T_0=-\Biggl(\left(\alpha^{6} r^{2}-\alpha^{4}\right) a^{6}+\left(-2 m \,r^{3} \alpha^{6}+2 m \alpha^{4} r +\alpha^{4} r^{2}-2 \alpha^{2}\right) a^{4}+\left(4 m \,r^{3} \alpha^{4}-4 m r \alpha^{2}-\alpha^{2} r^{2}-1\right) a^{2}\nonumber \\
&-2 \left(\alpha^{2} r^{2} m -m +\frac{1}{2} r \right) r \Biggr) r^{14},\\
&T_1=-\Biggl(\left(\alpha^{6} r^{3}+13 \alpha^{4} r \right) a^{6}+\left(-62 m \,r^{2} \alpha^{4}+r^{3} \alpha^{4}+64 m \alpha^{2}+26 r \alpha^{2}\right) a^{4}+\Biggl[60 \alpha^{2} r^{2} m -\alpha^{2} r^{3}-56 m \nonumber \\
&+13 r \Biggr] a^{2}+2 m \,r^{2}-r^{3}\Biggr) \alpha  r^{14},
\end{align}
\begin{align}
&T_2=+42 \Biggl(\left(\alpha^{6} r^{2}+\frac{1}{6} \alpha^{4}\right) a^{8}+\left(-\frac{8}{3} m \,r^{3} \alpha^{6}-\frac{13}{42} r^{4} \alpha^{6}+\frac{10}{3} m \alpha^{4} r +\alpha^{4} r^{2}+\frac{1}{3} \alpha^{2}\right) a^{6}\nonumber \\
&+\left(-\frac{1}{21} m \,r^{5} \alpha^{6}+16 m \,r^{3} \alpha^{4}-\frac{13}{21} \alpha^{4} r^{4}-\frac{44}{3} m r \alpha^{2}-\alpha^{2} r^{2}+\frac{1}{6}\right) a^{4}+\frac{2}{21} \Biggl[m \,r^{4} \alpha^{4}-16 \alpha^{2} r^{2} m -\frac{13}{4} \alpha^{2} r^{3}\nonumber \\
&+21 m -\frac{21}{2} r \Biggr] r \,a^{2}
-\frac{\alpha^{2} m \,r^{5}}{21}\Biggr) r^{12},\\
&T_3=-98 a^{2} \alpha  \Biggl(\alpha^{4} \left(-\frac{3}{7} \alpha^{2} r^{3}+m -\frac{51}{14} r \right) a^{6}+\frac{148 \left(-\frac{1}{2072} r^{5} \alpha^{4}+\alpha^{2} r^{2} m -\frac{3}{148} \alpha^{2} r^{3}-\frac{71}{74} m -\frac{51}{148} r \right) \alpha^{2} a^{4}}{7}\nonumber \\
&+\left(\frac{32}{49} m \,r^{4} \alpha^{4}-\frac{1}{49} r^{5} \alpha^{4}-\frac{872}{49} \alpha^{2} r^{2} m +\frac{3}{7} \alpha^{2} r^{3}+\frac{111}{7} m -\frac{51}{14} r \right) a^{2}-\frac{4 \left(\alpha^{2} r^{2} m +\frac{1}{56} \alpha^{2} r^{3}+m -\frac{3}{4} r \right) r^{2}}{7}\Biggr) r^{12},\\
&T_4=-462 a^{2} \Biggl(\left(\alpha^{6} r^{2}-\frac{1}{22} \alpha^{4}\right) a^{8}+\left(-\frac{118}{33} m \,r^{3} \alpha^{6}-\frac{17}{22} r^{4} \alpha^{6}+\frac{59}{11} m \alpha^{4} r +\alpha^{4} r^{2}-\frac{1}{11} \alpha^{2}\right) a^{6}\nonumber \\
&+\left(-\frac{10}{33} m \,r^{5} \alpha^{6}+\frac{652}{33} m \,r^{3} \alpha^{4}-\frac{17}{11} \alpha^{4} r^{4}-\frac{530}{33} m r \alpha^{2}-\alpha^{2} r^{2}-\frac{1}{22}\right) a^{4}\nonumber \\
&+\frac{4 \left(m \,r^{4} \alpha^{4}-\frac{35}{22} \alpha^{2} r^{2} m -\frac{51}{88} \alpha^{2} r^{3}+\frac{73}{44} m -\frac{3}{4} r \right) r \,a^{2}}{3}-\frac{2 \alpha^{2} m \,r^{5}}{11}\Biggr) r^{10},\\
&T_5=+980 a^{4} \alpha  \Biggl(\alpha^{4} \left(-\frac{33}{70} \alpha^{2} r^{3}+m -\frac{33}{20} r \right) a^{6}-\frac{1}{10}\Biggl[m \,r^{4} \alpha^{4}-\frac{1}{14} r^{5} \alpha^{4}-\frac{1026}{7} \alpha^{2} r^{2} m +\frac{33}{7} \alpha^{2} r^{3}+\frac{1004}{7} m \nonumber \\
&+33 r \Biggr] \alpha^{2} a^{4}+\left(\frac{71}{35} m \,r^{4} \alpha^{4}+\frac{1}{70} r^{5} \alpha^{4}-\frac{482}{35} \alpha^{2} r^{2} m +\frac{33}{70} \alpha^{2} r^{3}+\frac{321}{35} m -\frac{33}{20} r \right) a^{2}\nonumber \\
&-\frac{111 \left(\alpha^{2} r^{2} m -\frac{1}{222} \alpha^{2} r^{3}+\frac{14}{111} m -\frac{11}{37} r \right) r^{2}}{70}\Biggr) r^{10},\\
&T_6=+994 \Biggl(\left(\alpha^{6} r^{2}+\frac{5}{142} \alpha^{4}\right) a^{8}+\left(-\frac{428}{71} m \,r^{3} \alpha^{6}-\frac{231}{142} r^{4} \alpha^{6}+\frac{784}{71} m \alpha^{4} r +\alpha^{4} r^{2}+\frac{5}{71} \alpha^{2}\right) a^{6}\nonumber \\
&+\left(-\frac{177}{71} m \,r^{5} \alpha^{6}+\frac{2680}{71} m \,r^{3} \alpha^{4}-\frac{231}{71} \alpha^{4} r^{4}-\frac{1864}{71} m r \alpha^{2}-\alpha^{2} r^{2}+\frac{5}{142}\right) a^{4}+\frac{530}{71} \Biggl[m \,r^{4} \alpha^{4}-\frac{266}{265} \alpha^{2} r^{2} m \nonumber \\
&-\frac{231}{1060} \alpha^{2} r^{3}+\frac{2}{5} m -\frac{71}{530} r \Biggr] r \,a^{2}-\frac{73 \alpha^{2} m \,r^{5}}{71}\Biggr) a^{4} r^{8},
\end{align}
\begin{align}
&T_7=-3038 \Biggl(\alpha^{4} \left(-\frac{71}{217} \alpha^{2} r^{3}+m -\frac{1465}{3038} r \right) a^{6}-\frac{10}{31} \Biggl[m \,r^{4} \alpha^{4}+\frac{3}{140} r^{5} \alpha^{4}-\frac{1588}{49} \alpha^{2} r^{2} m +\frac{71}{70} \alpha^{2} r^{3}+\frac{7919}{245} m \nonumber \\
&+\frac{293}{98} r \Biggr] \alpha^{2} a^{4}+\left(\frac{1004}{217} m \,r^{4} \alpha^{4}-\frac{3}{217} r^{5} \alpha^{4}-\frac{20864}{1519} \alpha^{2} r^{2} m +\frac{71}{217} \alpha^{2} r^{3}+\frac{8383}{1519} m -\frac{1465}{3038} r \right) a^{2}\nonumber \\
&-\frac{642 \left(\alpha^{2} r^{2} m +\frac{1}{428} \alpha^{2} r^{3}-\frac{24}{107} m -\frac{71}{642} r \right) r^{2}}{217}\Biggr) a^{6} \alpha  r^{8},\\
&T_8=+35 \Biggl(a^{8} \alpha^{4}+\left(\frac{6864}{35} m \,r^{3} \alpha^{6}+\frac{293}{7} r^{4} \alpha^{6}-\frac{16766}{35} m \alpha^{4} r +2 \alpha^{2}\right) a^{6}+\Biggl[\frac{1568}{5} m \,r^{5} \alpha^{6}-\frac{59488}{35} m \,r^{3} \alpha^{4}\nonumber \\
&+\frac{586}{7} \alpha^{4} r^{4}+\frac{31676}{35} m r \alpha^{2}+1\Biggr] a^{4}+\left(-\frac{3728}{5} m \,r^{5} \alpha^{4}+\frac{25168}{35} \alpha^{2} r^{3} m +\frac{293}{7} \alpha^{2} r^{4}-\frac{434}{5} m r \right) a^{2}\nonumber \\
&+\frac{424 \alpha^{2} m \,r^{5}}{5}\Biggr) a^{6} r^{6},\\
&T_9=+2968 a^{8} \alpha  \Bigg(\alpha^{4} \left(\frac{1465 r}{2968}+m \right) a^{6}-\frac{217 \left(m \,r^{4} \alpha^{4}-\frac{5}{434} r^{5} \alpha^{4}-\frac{12584}{1519} \alpha^{2} r^{2} m +\frac{1864}{217} m -\frac{1465}{1519} r \right) \alpha^{2} a^{4}}{212}\nonumber \\
&+\left(\frac{7919}{742} m \,r^{4} \alpha^{4}+\frac{5}{212} r^{5} \alpha^{4}-\frac{7436}{371} \alpha^{2} r^{2} m +\frac{196}{53} m +\frac{1465}{2968} r \right) a^{2}-\frac{8383 \alpha^{2} m \,r^{4}}{1484}+\frac{5 \alpha^{2} r^{5}}{424}+\frac{858 m \,r^{2}}{371}\Biggr) r^{6},\\
&T_{10}=-994 a^{8} \Biggl(\left(\alpha^{6} r^{2}-\frac{3}{142} \alpha^{4}\right) a^{8}+\left(\frac{144}{71} m \,r^{3} \alpha^{6}-\frac{1465}{994} r^{4} \alpha^{6}-\frac{642}{71} m \alpha^{4} r +\alpha^{4} r^{2}-\frac{3}{71} \alpha^{2}\right) a^{6}\nonumber \\ &+\left(\frac{8383}{497} m \,r^{5} \alpha^{6}-\frac{20864}{497} m \,r^{3} \alpha^{4}-\frac{1465}{497} \alpha^{4} r^{4}+\frac{1004}{71} m r \alpha^{2}-\alpha^{2} r^{2}-\frac{3}{142}\right) a^{4}+\Biggl[-\frac{15838}{497} m \,r^{5} \alpha^{4}\nonumber \\
&+\frac{15880}{497} \alpha^{2} r^{3} m -\frac{1465}{994} \alpha^{2} r^{4}-\frac{70}{71} m r -r^{2}\Biggr] a^{2}+\frac{217 \alpha^{2} m \,r^{5}}{71}\Biggr) r^{4},\\
&T_{11}=-1022 a^{10} \alpha  \Biggl(\alpha^{4} \left(\frac{71}{73} \alpha^{2} r^{3}+m +\frac{231}{146} r \right) a^{6}-\frac{212 \alpha^{2}}{73} \Biggl[m \,r^{4} \alpha^{4}+\frac{5}{424} r^{5} \alpha^{4}-\frac{133}{53} \alpha^{2} r^{2} m -\frac{71}{212} \alpha^{2} r^{3}\nonumber \\
&+\frac{5}{2} m -\frac{231}{212} r \Biggr] a^{4}+\left(\frac{1864}{73} m \,r^{4} \alpha^{4}-\frac{5}{73} r^{5} \alpha^{4}-\frac{2680}{73} \alpha^{2} r^{2} m -\frac{71}{73} \alpha^{2} r^{3}+\frac{177}{73} m +\frac{231}{146} r \right) a^{2}\nonumber \\
&-\frac{784 \left(\alpha^{2} r^{2} m +\frac{5}{1568} \alpha^{2} r^{3}-\frac{107}{196} m +\frac{71}{784} r \right) r^{2}}{73}\Biggr) r^{4},
\end{align}

\begin{align}
&T_{12}=+462 a^{10} \Biggl(\left(\alpha^{6} r^{2}+\frac{1}{66} \alpha^{4}\right) a^{8}+\left(-\frac{14}{33} m \,r^{3} \alpha^{6}-\frac{7}{2} r^{4} \alpha^{6}-\frac{37}{11} m \alpha^{4} r +\alpha^{4} r^{2}+\frac{1}{33} \alpha^{2}\right) a^{6}\nonumber \\
&+\left(\frac{214}{11} m \,r^{5} \alpha^{6}-\frac{964}{33} m \,r^{3} \alpha^{4}-7 \alpha^{4} r^{4}+\frac{142}{33} m r \alpha^{2}-\alpha^{2} r^{2}+\frac{1}{66}\right) a^{4}+\Biggl[-\frac{1004}{33} m \,r^{5} \alpha^{4}+\frac{342}{11} \alpha^{2} r^{3} m -\frac{7}{2} \alpha^{2} r^{4}\nonumber \\
&-\frac{7}{33} m r -r^{2}\Biggr] a^{2}+\frac{70 \alpha^{2} m \,r^{5}}{33}\Biggr) r^{2},\\
&T_{13}=+84 \Biggl(\alpha^{4} \left(\frac{11}{2} \alpha^{2} r^{3}+m +\frac{17}{4} r \right) a^{6}-\frac{73 \left(m \,r^{4} \alpha^{4}-\frac{3}{146} r^{5} \alpha^{4}-\frac{70}{73} \alpha^{2} r^{2} m -\frac{33}{73} \alpha^{2} r^{3}+\frac{44}{73} m -\frac{51}{73} r \right) \alpha^{2} a^{4}}{6}\nonumber \\
&+\left(\frac{265}{3} m \,r^{4} \alpha^{4}+\frac{1}{2} r^{5} \alpha^{4}-\frac{326}{3} \alpha^{2} r^{2} m -\frac{11}{2} \alpha^{2} r^{3}+\frac{5}{3} m +\frac{17}{4} r \right) a^{2}\nonumber \\
&-\frac{59 \left(\alpha^{2} r^{2} m -\frac{1}{118} \alpha^{2} r^{3}-\frac{2}{3} m +\frac{11}{59} r \right) r^{2}}{2}\Biggr) a^{12} \alpha  r^{2},\\
&T_{14}=-42 \Biggl(\left(\alpha^{6} r^{2}-\frac{1}{42} \alpha^{4}\right) a^{8}-\frac{4 \alpha^{2} \left(m \,r^{3} \alpha^{4}+\frac{51}{8} \alpha^{4} r^{4}+m r \alpha^{2}-\frac{3}{4} \alpha^{2} r^{2}+\frac{1}{28}\right) a^{6}}{3}\nonumber \\
&+\left(37 m \,r^{5} \alpha^{6}-\frac{872}{21} m \,r^{3} \alpha^{4}-17 \alpha^{4} r^{4}+\frac{32}{21} m r \alpha^{2}-\alpha^{2} r^{2}-\frac{1}{42}\right) a^{4}+\Biggl[-\frac{142}{3} m \,r^{5} \alpha^{4}+\frac{148}{3} \alpha^{2} r^{3} m -\frac{17}{2} \alpha^{2} r^{4}\nonumber \\
&-r^{2}\Biggr] a^{2}+\frac{7 \alpha^{2} m \,r^{5}}{3}\Biggr) a^{12},\\
&T_{15}=-2 \Biggl(\alpha^{4} \left(21 \alpha^{2} r^{3}+m +\frac{13}{2} r \right) a^{6}+\left(-42 m \,r^{4} \alpha^{6}-\frac{7}{2} r^{5} \alpha^{6}+32 m \,r^{2} \alpha^{4}+21 r^{3} \alpha^{4}-2 m \alpha^{2}+13 r \alpha^{2}\right) a^{4}\nonumber \\
&+\left(308 m \,r^{4} \alpha^{4}-7 r^{5} \alpha^{4}-336 \alpha^{2} r^{2} m -21 \alpha^{2} r^{3}+m +\frac{13}{2} r \right) a^{2}-70 \alpha^{2} m \,r^{4}-\frac{7 \alpha^{2} r^{5}}{2}+56 m \,r^{2}-21 r^{3}\Biggr) a^{14}\alpha,\\
&T_{16}=+\Biggl(a^{6} \alpha^{6}+\left(-2 \alpha^{6} m r -13 \alpha^{6} r^{2}+\alpha^{4}\right) a^{4}+\left(56 m \,r^{3} \alpha^{6}-60 m \alpha^{4} r -26 \alpha^{4} r^{2}-\alpha^{2}\right) a^{2}\nonumber \\
&-64 m \,r^{3} \alpha^{4}+62 m r \alpha^{2}-13 \alpha^{2} r^{2}-1\Biggr) a^{16},\\
&T_{17}=a^{16} \alpha  \Biggl(a^{6} \alpha^{6} r -2 \left(\alpha^{2} r^{2} m -\frac{1}{2} \alpha^{2} r^{3}-m -\frac{1}{2} r \right) \alpha^{4} a^{4}+4 \left(\alpha^{2} r^{2} m +\frac{1}{2} \alpha^{2} r^{3}-m -\frac{1}{4} r \right) \alpha^{2} a^{2}\nonumber \\
&-2 \alpha^{2} r^{2} m +\alpha^{2} r^{3}+2 m -r \Biggr).
\label{gradflowaccelI5}
\end{align}
\begin{remark}
On the symmetry axis ($\theta=0$), the norm of the gradient vector $k_{\mu}$ for an accelerating Kerr black hole, acquires the form:
\begin{align}
I_5&=-\frac{82944 m^{4} \left(\alpha  r -1\right)^{13} \left(\alpha  r +1\right) \left(\alpha^{2} a^{2}+1\right)^{2} \left(a^{2}-2 m r +r^{2}\right)}{\left(a^{2}+r^{2}\right)^{15}} \Biggl(\alpha  a^{8}-21 a^{6} \alpha  r^{2}+35 a^{4} \alpha  r^{4}\nonumber \\
&-7 a^{2} \alpha  r^{6}-7 a^{6} r +35 a^{4} r^{3}-21 a^{2} r^{5}+r^{7}\Biggr)^{2}.\label{symaxisI5accelKerr}
\end{align}
The vanishing of $I_5$ on the axis singles out the horizons of the Kerr black hole as well as the acceleration horizon $r=1/\alpha$ (away from the discrete positive real roots of the septic radial polynomial).
\end{remark}
\begin{remark}
In the equatorial plane, $\theta=\frac{\pi}{2}$, eqn.(\ref{gradflowaccelI5}) reduces to the expression:
\begin{align}
I_5&=-\frac{82944 m^{4}}{r^{16}} \Biggl(-2 m \alpha^{2} \left(a \alpha -1\right)^{2} \left(a \alpha +1\right)^{2} r^{3}+\left(a^{6} \alpha^{6}+\alpha^{4} a^{4}-\alpha^{2} a^{2}-1\right) r^{2}\nonumber \\
&+\left(2 \alpha^{4} a^{4} m -4 \alpha^{2} a^{2} m +2 m \right) r -a^{2} \left(\alpha^{2} a^{2}+1\right)^{2}\Biggr).
\end{align}
\end{remark}

In Fig.\ref{RegionIfunfAccelKerr} we display region plots for the invariant $I_5$ for an accelerating Kerr black hole. Again we note the intriguing global behaviour of $I_5$ we observed for the non-accelerating rotating black holes. In general terms, the area of regions where $I_5<0$ on any $r-\theta$ hypersurface decreases as the Kerr parameter increases. On the other hand, for fixed value of the Kerr parameter $a$, increasing the acceleration $\alpha$ has a minimal effect and only distorts the symmetry of the plots with small or zero acceleration.

\begin{figure}[ptbh]
\centering
  \begin{subfigure}[b]{.60\linewidth}
    \centering
    \includegraphics[width=.99\textwidth]{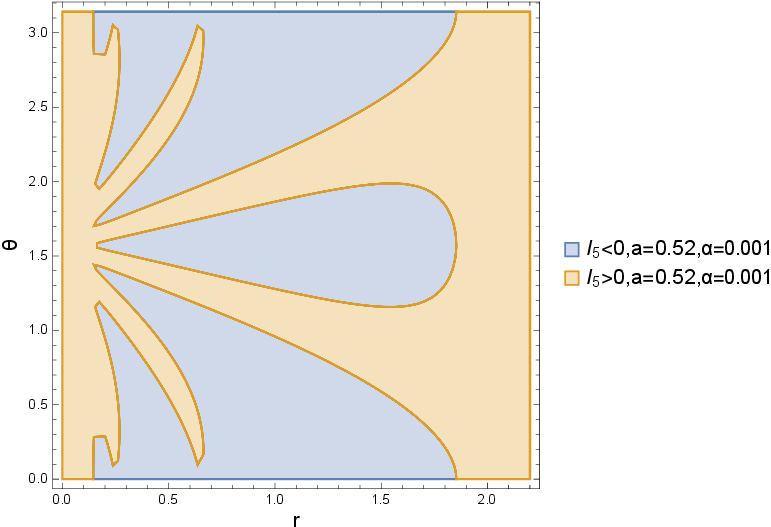}
    \caption{Region plot of $I_5$ for accelerating Kerr BH.}\label{RegionNormI5AccelKa052alpha0001}
  \end{subfigure}%
  \begin{subfigure}[b]{.60\linewidth}
    \centering
    \includegraphics[width=.99\textwidth]{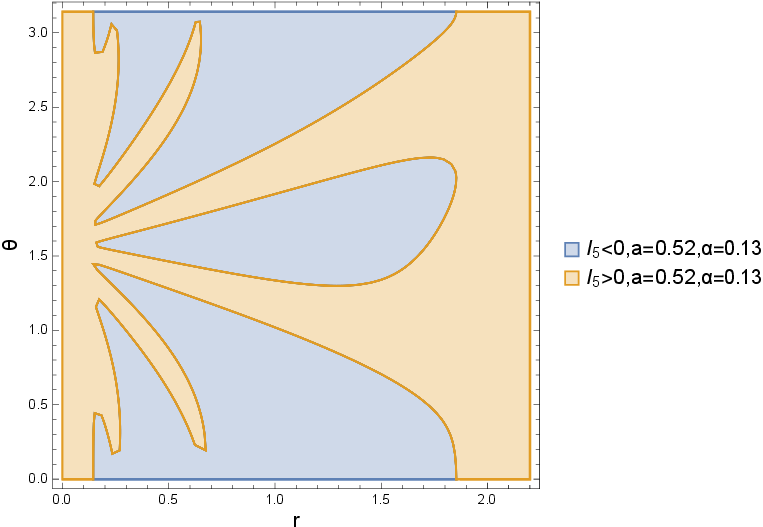}
    \caption{Region plot of  $I_5$ for accelerating  Kerr BH.}\label{RegionI5accelKa052alpha013}
  \end{subfigure}\\
  \begin{subfigure}[b]{.60\linewidth}
    \centering
    \includegraphics[width=.99\textwidth]{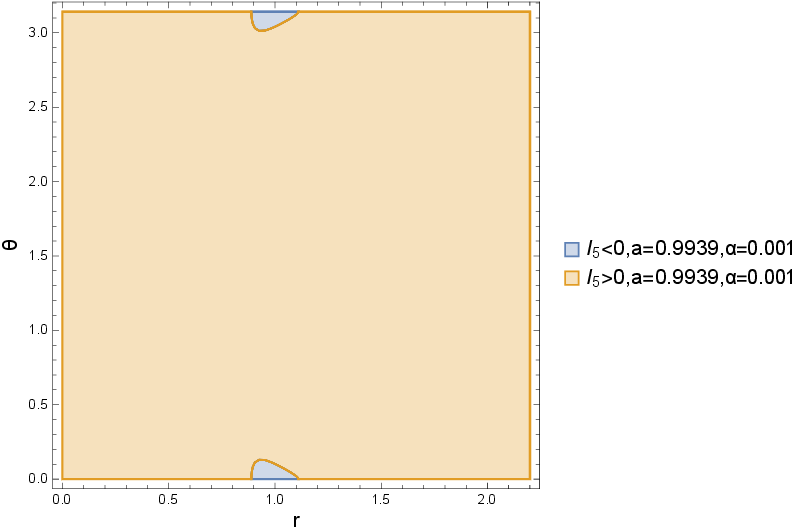}
    \caption{Region Plot of $I_5$ for accelerating Kerr BH.}\label{RegionI5kaccel09939alpha0001}
  \end{subfigure}%
  \begin{subfigure}[b]{.60\linewidth}
    \centering
    \includegraphics[width=.99\textwidth]{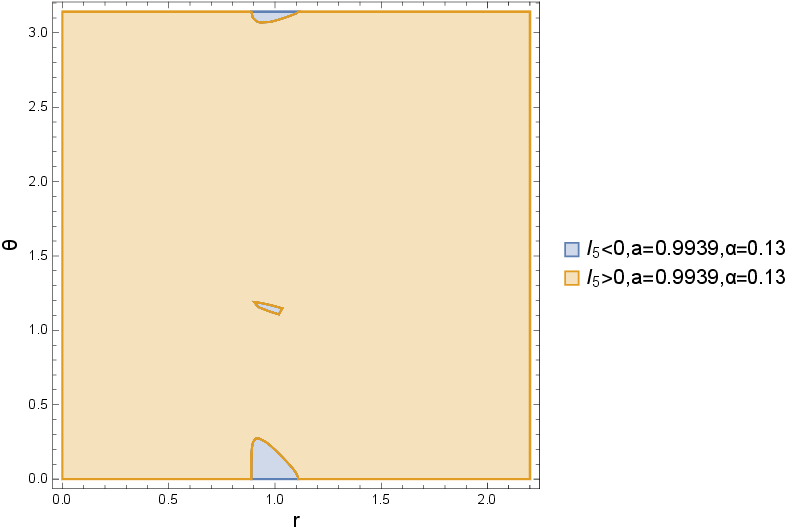}
    \caption{Region Plot of $I_5$ for accelerating Kerr BH.}\label{RegionI5kaccela09939alpha013}
  \end{subfigure}
  \caption{Region plots of the  curvature invariant $I_5$ , for the accelerating Kerr black hole. (a) for spin parameter $a=0.52$, acceleration $\alpha=0.001,m=1$. (b) For spin  $a=0.52$, acceleration  $\alpha=0.13,m=1$. (c) For high spin  $a=0.9939$, acceleration parameter $\alpha=0.001$, and mass $m=1$. (d) For high spin  $a=0.9939$, $\alpha=0.13$, and mass $m=1$.}\label{RegionIfunfAccelKerr}
\end{figure}

\begin{theorem}
We computed the invariant $I_6$ for the accelerating Kerr black hole with the result:
\begin{align}
I_6=-\frac{82944 a^{2} m^{4} \left(\alpha  r \cos \! \left(\theta \right)-1\right)^{13}}{\left(r^{2}+a^{2} \cos \! \left(\theta \right)^{2}\right)^{15}}\sum_{i=0}^{17}\mathcal{C}_i\;(\cos(\theta))^i,
\end{align}
\end{theorem}
where the coefficients $\mathcal{C}_i$ are:
\begin{align}
&\mathcal{C}_0=+r^{14} \left(8 \alpha^{4} m \,r^{3}-8 m r \alpha^{2}+\left(\alpha^{2} a^{2}+1\right)^{2}\right),\\
&\mathcal{C}_1=-2 r^{14} \alpha  \left(\left(32 a^{2} m \alpha^{4}-28 m \alpha^{2}\right) r^{2}-\frac{15 \left(\alpha^{2} a^{2}+1\right)^{2} r}{2}+a^{4} m \alpha^{4}-30 a^{2} m \alpha^{2}+29 m \right),\\
&\mathcal{C}_2=-48 r^{12} \Biggl(\frac{m \,r^{5} \alpha^{4}}{6}-\frac{5 \alpha^{2} \left(\alpha^{2} a^{2}+1\right)^{2} r^{4}}{16}-\frac{31 \alpha^{2} \left(a^{4} \alpha^{4}-\frac{162}{31} \alpha^{2} a^{2}+\frac{19}{31}\right) m \,r^{3}}{12}\nonumber \\
&+\left(a^{6} \alpha^{6}+a^{4} \alpha^{4}-\alpha^{2} a^{2}-1\right) r^{2}+\frac{77 \left(a^{4} \alpha^{4}-\frac{42}{11} \alpha^{2} a^{2}+\frac{7}{11}\right) m r}{24}-\frac{7 a^{2} \left(\alpha^{2} a^{2}+1\right)^{2}}{48}\Biggr),\\
&\mathcal{C}_3=+84 r^{12} \Biggl(\frac{\alpha^{2} \left(\alpha^{2} a^{2}+1\right)^{2} r^{5}}{84}+\left(-\frac{1}{42} a^{4} m \alpha^{6}+\frac{5}{7} a^{2} m \alpha^{4}-\frac{29}{42} m \alpha^{2}\right) r^{4}\nonumber \\
&+\left(-\frac{4}{7} a^{6} \alpha^{6}-\frac{4}{7} a^{4} \alpha^{4}+\frac{4}{7} \alpha^{2} a^{2}+\frac{4}{7}\right) r^{3}+\frac{515 \left(a^{4} \alpha^{4}-\frac{442}{515} \alpha^{2} a^{2}-\frac{17}{515}\right) m \,r^{2}}{21}-\frac{49 a^{2} \left(\alpha^{2} a^{2}+1\right)^{2} r}{12}\nonumber \\
&+a^{2} m \left(a^{4} \alpha^{4}-24 \alpha^{2} a^{2}+\frac{55}{3}\right)\Biggr) \alpha,\\
&\mathcal{C}_4=+448 r^{10} \Biggl(-\frac{11 \alpha^{2} \left(a^{4} \alpha^{4}-\frac{42}{11} \alpha^{2} a^{2}+\frac{7}{11}\right) m \,r^{5}}{32}-\frac{49 a^{2} \alpha^{2} \left(\alpha^{2} a^{2}+1\right)^{2} r^{4}}{64}\nonumber \\
&-\frac{29 \left(a^{4} \alpha^{4}-\frac{164}{29} \alpha^{2} a^{2}+\frac{17}{29}\right) a^{2} \alpha^{2} m \,r^{3}}{8}+\left(a^{8} \alpha^{6}+a^{6} \alpha^{4}-a^{4} \alpha^{2}-a^{2}\right) r^{2}\nonumber\\
&+\frac{87 a^{2} m \left(a^{4} \alpha^{4}-\frac{268}{87} \alpha^{2} a^{2}+\frac{35}{87}\right) r}{16}+\frac{3 a^{4} \left(\alpha^{2} a^{2}+1\right)^{2}}{64}\Biggr),\\
&\mathcal{C}_5=-1022 \Biggl(-\frac{\alpha^{2} \left(\alpha^{2} a^{2}+1\right)^{2} r^{5}}{146}-\frac{6 \left(a^{4} \alpha^{4}-24 \alpha^{2} a^{2}+\frac{55}{3}\right) \alpha^{2} m \,r^{4}}{73}\nonumber \\
&-\frac{32 \left(\alpha  a -1\right) \left(\alpha  a +1\right) \left(\alpha^{2} a^{2}+1\right)^{2} r^{3}}{73}+\frac{1028 \left(a^{4} \alpha^{4}-\frac{240}{257} \alpha^{2} a^{2}-\frac{3}{257}\right) m \,r^{2}}{73}\nonumber \\
&-\frac{237 a^{2} \left(\alpha^{2} a^{2}+1\right)^{2} r}{146}+a^{2} m \left(a^{4} \alpha^{4}-\frac{998}{73} \alpha^{2} a^{2}+\frac{645}{73}\right)\Biggr) r^{10} a^{2} \alpha,
\end{align}
\begin{align}
&\mathcal{C}_6=-1008 r^{8} a^{2} \Biggl(-\frac{29 \alpha^{2} m \left(a^{4} \alpha^{4}-\frac{268}{87} \alpha^{2} a^{2}+\frac{35}{87}\right) r^{5}}{12}-\frac{79 a^{2} \alpha^{2} \left(\alpha^{2} a^{2}+1\right)^{2} r^{4}}{48}\nonumber \\
&-\frac{215 a^{2} \alpha^{2} \left(a^{4} \alpha^{4}-\frac{1338}{215} \alpha^{2} a^{2}+\frac{267}{215}\right) m \,r^{3}}{36}+\left(a^{8} \alpha^{6}+a^{6} \alpha^{4}-a^{4} \alpha^{2}-a^{2}\right) r^{2}\nonumber \\
&+\frac{263 \left(a^{4} \alpha^{4}-\frac{618}{263} \alpha^{2} a^{2}+\frac{217}{789}\right) a^{2} m r}{24}-\frac{5 a^{4} \left(\alpha^{2} a^{2}+1\right)^{2}}{144}\Biggr),\\
&\mathcal{C}_7=+2968 r^{8} a^{4} \alpha  \Biggl(\frac{3 \alpha^{2} \left(\alpha^{2} a^{2}+1\right)^{2} r^{5}}{424}-\frac{73 \left(a^{4} \alpha^{4}-\frac{998}{73} \alpha^{2} a^{2}+\frac{645}{73}\right) \alpha^{2} m \,r^{4}}{212}\nonumber \\
&-\frac{18 \left(\alpha  a -1\right) \left(\alpha  a +1\right) \left(\alpha^{2} a^{2}+1\right)^{2} r^{3}}{53}+\frac{7933 \left(a^{4} \alpha^{4}-\frac{10446}{7933} \alpha^{2} a^{2}+\frac{497}{7933}\right) m \,r^{2}}{742}\nonumber \\
&-\frac{1395 a^{2} \left(\alpha^{2} a^{2}+1\right)^{2} r}{2968}+a^{2} m \left(a^{4} \alpha^{4}-\frac{3977}{371} \alpha^{2} a^{2}+\frac{2087}{371}\right)\Biggr),\\
&\mathcal{C}_8=+35 r^{6} a^{4} \Biggl(-\frac{1578 \left(a^{4} \alpha^{4}-\frac{618}{263} \alpha^{2} a^{2}+\frac{217}{789}\right) \alpha^{2} m \,r^{5}}{5}-\frac{279 a^{2} \alpha^{2} \left(\alpha^{2} a^{2}+1\right)^{2} r^{4}}{7}\nonumber \\
&-\frac{6864 a^{2} \alpha^{2} m \left(a^{4} \alpha^{4}-\frac{26}{3} \alpha^{2} a^{2}+\frac{11}{3}\right) r^{3}}{35}+\frac{16696 a^{2} m \left(a^{4} \alpha^{4}-\frac{3977}{2087} \alpha^{2} a^{2}+\frac{371}{2087}\right) r}{35}+a^{4} \left(\alpha^{2} a^{2}+1\right)^{2}\Biggr),\\
&\mathcal{C}_9=-3038 r^{6} a^{6} \alpha  \Biggl(-\frac{5 \alpha^{2} \left(\alpha^{2} a^{2}+1\right)^{2} r^{5}}{434}-\frac{212 \alpha^{2} \left(a^{4} \alpha^{4}-\frac{3977}{371} \alpha^{2} a^{2}+\frac{2087}{371}\right) m \,r^{4}}{217}\nonumber \\
&+\frac{12584 \left(a^{4} \alpha^{4}-\frac{26}{11} \alpha^{2} a^{2}+\frac{3}{11}\right) m \,r^{2}}{1519}+\frac{45 a^{2} \left(\alpha^{2} a^{2}+1\right)^{2} r}{98}+a^{2} m \left(a^{4} \alpha^{4}-\frac{1854}{217} \alpha^{2} a^{2}+\frac{789}{217}\right)\Biggr),\\
&\mathcal{C}_{10}=+1008 r^{4} a^{6} \Biggl(\frac{2087 \alpha^{2} m \left(a^{4} \alpha^{4}-\frac{3977}{2087} \alpha^{2} a^{2}+\frac{371}{2087}\right) r^{5}}{126}-\frac{155 a^{2} \alpha^{2} \left(\alpha^{2} a^{2}+1\right)^{2} r^{4}}{112}\nonumber \\
&+\frac{71 a^{2} \alpha^{2} \left(a^{4} \alpha^{4}-\frac{10446}{497} \alpha^{2} a^{2}+\frac{7933}{497}\right) m \,r^{3}}{36}+\left(a^{8} \alpha^{6}+a^{6} \alpha^{4}-a^{4} \alpha^{2}-a^{2}\right) r^{2}\nonumber \\
&-\frac{215 a^{2} \left(a^{4} \alpha^{4}-\frac{998}{645} \alpha^{2} a^{2}+\frac{73}{645}\right) m r}{24}+\frac{a^{4} \left(\alpha^{2} a^{2}+1\right)^{2}}{48}\Biggr),\\
&\mathcal{C}_{11}=+980 r^{4} a^{8} \alpha  \Biggl(\frac{\alpha^{2} \left(\alpha^{2} a^{2}+1\right)^{2} r^{5}}{28}-\frac{31 \left(a^{4} \alpha^{4}-\frac{1854}{217} \alpha^{2} a^{2}+\frac{789}{217}\right) \alpha^{2} m \,r^{4}}{10}\nonumber \\
&+\frac{36 \left(\alpha  a -1\right) \left(\alpha  a +1\right) \left(\alpha^{2} a^{2}+1\right)^{2} r^{3}}{35}+\frac{267 \left(a^{4} \alpha^{4}-\frac{446}{89} \alpha^{2} a^{2}+\frac{215}{267}\right) m \,r^{2}}{35}+\frac{237 a^{2} \left(\alpha^{2} a^{2}+1\right)^{2} r}{140}\nonumber \\
&+a^{2} m \left(a^{4} \alpha^{4}-\frac{268}{35} \alpha^{2} a^{2}+\frac{87}{35}\right)\Biggr),
\end{align}
\begin{align}
&\mathcal{C}_{12}=-448 r^{2} a^{8} \Biggl(\frac{645 \alpha^{2} \left(a^{4} \alpha^{4}-\frac{998}{645} \alpha^{2} a^{2}+\frac{73}{645}\right) m \,r^{5}}{32}-\frac{237 a^{2} \alpha^{2} \left(\alpha^{2} a^{2}+1\right)^{2} r^{4}}{64}\nonumber \\
&-\frac{3 a^{2} \left(a^{4} \alpha^{4}+80 \alpha^{2} a^{2}-\frac{257}{3}\right) \alpha^{2} m \,r^{3}}{8}+\left(a^{8} \alpha^{6}+a^{6} \alpha^{4}-a^{4} \alpha^{2}-a^{2}\right) r^{2}-\frac{55 \left(a^{4} \alpha^{4}-\frac{72}{55} \alpha^{2} a^{2}+\frac{3}{55}\right) a^{2} m r}{16}\nonumber \\
&-\frac{a^{4} \left(\alpha^{2} a^{2}+1\right)^{2}}{64}\Biggr),\\
&\mathcal{C}_{13}=-98 \Biggl(-\frac{3 \alpha^{2} \left(\alpha^{2} a^{2}+1\right)^{2} r^{5}}{14}-10 \alpha^{2} \left(a^{4} \alpha^{4}-\frac{268}{35} \alpha^{2} a^{2}+\frac{87}{35}\right) m \,r^{4}+\Biggl[\frac{32}{7} a^{6} \alpha^{6}+\frac{32}{7} a^{4} \alpha^{4}-\frac{32}{7} \alpha^{2} a^{2}\nonumber \\
&-\frac{32}{7}\Bigg] r^{3}+\frac{68 \left(a^{4} \alpha^{4}-\frac{164}{17} \alpha^{2} a^{2}+\frac{29}{17}\right) m \,r^{2}}{7}+\frac{7 a^{2} \left(\alpha^{2} a^{2}+1\right)^{2} r}{2}+a^{2} m \left(a^{4} \alpha^{4}-6 \alpha^{2} a^{2}+\frac{11}{7}\right)\Biggr) r^{2} a^{10} \alpha,\\
&\mathcal{C}_{14}=+48 a^{10} \Biggl(\frac{385 \left(a^{4} \alpha^{4}-\frac{72}{55} \alpha^{2} a^{2}+\frac{3}{55}\right) \alpha^{2} m \,r^{5}}{12}-\frac{343 a^{2} \alpha^{2} \left(\alpha^{2} a^{2}+1\right)^{2} r^{4}}{48}\nonumber \\
&-\frac{17 \left(a^{4} \alpha^{4}+26 \alpha^{2} a^{2}-\frac{515}{17}\right) a^{2} \alpha^{2} m \,r^{3}}{12}+\left(a^{8} \alpha^{6}+a^{6} \alpha^{4}-a^{4} \alpha^{2}-a^{2}\right) r^{2}\nonumber \\
&+\left(-\frac{29}{24} a^{6} m \alpha^{4}+\frac{5}{4} a^{4} m \alpha^{2}-\frac{1}{24} a^{2} m \right) r +\frac{a^{4} \left(\alpha^{2} a^{2}+1\right)^{2}}{48}\Biggr),\\
&\mathcal{C}_{15}=+48 a^{12} \alpha  \Biggl(\frac{7 \alpha^{2} \left(\alpha^{2} a^{2}+1\right)^{2} r^{5}}{48}-\frac{49 \left(a^{4} \alpha^{4}-6 \alpha^{2} a^{2}+\frac{11}{7}\right) \alpha^{2} m \,r^{4}}{24}+\left(a^{6} \alpha^{6}+a^{4} \alpha^{4}-\alpha^{2} a^{2}-1\right) r^{3}\nonumber \\
&+\frac{19 \left(a^{4} \alpha^{4}-\frac{162}{19} \alpha^{2} a^{2}+\frac{31}{19}\right) m \,r^{2}}{12}+\frac{5 a^{2} \left(\alpha^{2} a^{2}+1\right)^{2} r}{16}-\frac{a^{4} m \alpha^{2}}{6}\Biggr),\\
&\mathcal{C}_{16}=+15 r \left(\left(-\frac{58}{15} a^{4} m \alpha^{4}+4 a^{2} m \alpha^{2}-\frac{2}{15} m \right) r^{2}+a^{2} \left(\alpha^{2} a^{2}+1\right)^{2} r +\frac{56 a^{2} \left(\alpha^{2} a^{2}-\frac{8}{7}\right) m}{15}\right) a^{12} \alpha^{2},\\
&\mathcal{C}_{17}=\left(\left(\alpha^{2} a^{2}+1\right)^{2} r^{3}-8 a^{2} m \alpha^{2} r^{2}+8 a^{2} m \right) a^{14} \alpha^{3},
\end{align}

\begin{theorem}\label{sebenepitaxkerr}
We computed the invariant $I_7$ for an accelerating Kerr black hole:
\begin{align}
&I_7=-\frac{580608 \left(\alpha  r \cos \! \left(\theta \right)-1\right)^{13} a \,m^{4}}{\left(r^{2}+a^{2} \cos \! \left(\theta \right)^{2}\right)^{15}}\sum_{i=0}^{17}\mathcal{F}_i\;(\cos(\theta))^i,
\end{align}
\end{theorem}
where
\begin{align}
&\mathcal{F}_0=-\frac{r^{15} \alpha  \left(\left(-4 a^{2} m \alpha^{4}+4 m \alpha^{2}\right) r^{2}+\left(a^{2} \alpha^{2}+1\right)^{2} r +4 a^{2} \alpha^{2} m -4 m \right)}{7},\\
&\mathcal{F}_1=+r^{14} \Biggl(-\frac{\alpha^{2} \left(a^{2} \alpha^{2}+1\right)^{2} r^{3}}{7}-\frac{16 m \alpha^{2} \left(a^{4} \alpha^{4}-\frac{23}{4} a^{2} \alpha^{2}+\frac{3}{4}\right) r^{2}}{7}+\left(a^{6} \alpha^{6}+a^{4} \alpha^{4}-a^{2} \alpha^{2}-1\right) r \nonumber \\
&+\frac{18 m \left(a^{4} \alpha^{4}-\frac{44}{9} a^{2} \alpha^{2}+\frac{7}{9}\right)}{7}\Biggr),\\
&\mathcal{F}_2=-2 r^{13} \Biggl(\left(\frac{2}{7} a^{2} m \alpha^{4}-\frac{2}{7} m \alpha^{2}\right) r^{4}+\left(-\frac{1}{2} a^{6} \alpha^{6}-\frac{1}{2} a^{4} \alpha^{4}+\frac{1}{2} a^{2} \alpha^{2}+\frac{1}{2}\right) r^{3}\nonumber \\
&+\frac{230 m \left(a^{4} \alpha^{4}-\frac{20}{23} a^{2} \alpha^{2}-\frac{3}{115}\right) r^{2}}{7}-\frac{45 a^{2} \left(a^{2} \alpha^{2}+1\right)^{2} r}{7}+a^{2} m \left(a^{4} \alpha^{4}-32 a^{2} \alpha^{2}+27\right)\Biggr) \alpha,\\
&\mathcal{F}_3=-25 r^{12} \Biggl(-\frac{18 m \alpha^{2} \left(a^{4} \alpha^{4}-\frac{44}{9} a^{2} \alpha^{2}+\frac{7}{9}\right) r^{4}}{175}-\frac{18 a^{2} \alpha^{2} \left(a^{2} \alpha^{2}+1\right)^{2} r^{3}}{35},\nonumber \\
&-\frac{76 a^{2} \left(a^{4} \alpha^{4}-\frac{736}{133} a^{2} \alpha^{2}+\frac{71}{133}\right) m \alpha^{2} r^{2}}{25}+\left(a^{8} \alpha^{6}+a^{6} \alpha^{4}-a^{4} \alpha^{2}-a^{2}\right) r +\frac{104 \left(a^{4} \alpha^{4}-\frac{7}{2} a^{2} \alpha^{2}+\frac{1}{2}\right) a^{2} m}{25}\Biggr),\\
&\mathcal{F}_4=+52 r^{11} a^{2} \Biggl(-\frac{m \alpha^{2} \left(a^{4} \alpha^{4}-32 a^{2} \alpha^{2}+27\right) r^{4}}{26}-\frac{25 \left(a \alpha -1\right) \left(a \alpha +1\right) \left(a^{2} \alpha^{2}+1\right)^{2} r^{3}}{52}\nonumber \\
&+\frac{228 \left(a^{4} \alpha^{4}-\frac{33}{38} a^{2} \alpha^{2}-\frac{1}{38}\right) m \,r^{2}}{13}-\frac{5 a^{2} \left(a^{2} \alpha^{2}+1\right)^{2} r}{2}+a^{2} m \left(a^{4} \alpha^{4}-17 a^{2} \alpha^{2}+12\right)\Biggr) \alpha
\end{align}
\begin{align}
&\mathcal{F}_5=117 r^{10} \Biggl(-\frac{8 \left(a^{4} \alpha^{4}-\frac{7}{2} a^{2} \alpha^{2}+\frac{1}{2}\right) m \alpha^{2} r^{4}}{9}-\frac{10 a^{2} \alpha^{2} \left(a^{2} \alpha^{2}+1\right)^{2} r^{3}}{9}-\frac{40 a^{2} m \left(a^{4} \alpha^{4}-\frac{29}{5} a^{2} \alpha^{2}+\frac{4}{5}\right) \alpha^{2} r^{2}}{9}\nonumber \\
&+\left(a^{8} \alpha^{6}+a^{6} \alpha^{4}-a^{4} \alpha^{2}-a^{2}\right) r +\frac{22 a^{2} \left(a^{4} \alpha^{4}-\frac{8}{3} a^{2} \alpha^{2}+\frac{1}{3}\right) m}{3}\Biggr) a^{2},\\
&\mathcal{F}_6=-286 r^{9} a^{4} \Biggl(-\frac{2 m \alpha^{2} \left(a^{4} \alpha^{4}-17 a^{2} \alpha^{2}+12\right) r^{4}}{11}+\left(-\frac{9}{22} a^{6} \alpha^{6}-\frac{9}{22} a^{4} \alpha^{4}+\frac{9}{22} a^{2} \alpha^{2}+\frac{9}{22}\right) r^{3}\nonumber \\
&+\frac{134 m \left(a^{4} \alpha^{4}-\frac{72}{67} a^{2} \alpha^{2}+\frac{1}{67}\right) r^{2}}{11}-a^{2} \left(a^{2} \alpha^{2}+1\right)^{2} r +a^{2} m \left(a^{4} \alpha^{4}-12 a^{2} \alpha^{2}+7\right)\Biggr) \alpha,\\
&\mathcal{F}_7=-\frac{715 r^{8} a^{4}}{7} \Biggl(-\frac{42 \left(a^{4} \alpha^{4}-\frac{8}{3} a^{2} \alpha^{2}+\frac{1}{3}\right) m \alpha^{2} r^{4}}{5}-\frac{14 a^{2} \alpha^{2} \left(a^{2} \alpha^{2}+1\right)^{2} r^{3}}{5}\nonumber \\
&-\frac{52 a^{2} \left(a^{4} \alpha^{4}-\frac{92}{13} a^{2} \alpha^{2}+\frac{27}{13}\right) m \alpha^{2} r^{2}}{5}+\left(a^{8} \alpha^{6}+a^{6} \alpha^{4}-a^{4} \alpha^{2}-a^{2}\right) r +\frac{108 a^{6} \alpha^{4} m}{5}-\frac{228 a^{4} \alpha^{2} m}{5}+\frac{24 a^{2} m}{5}\Biggr),\\
&\mathcal{F}_8=+\frac{3432 r^{7} a^{6} \alpha}{7}  \Biggl(-\frac{7 m \alpha^{2} \left(a^{4} \alpha^{4}-12 a^{2} \alpha^{2}+7\right) r^{4}}{12}+\left(-\frac{5}{24} a^{6} \alpha^{6}-\frac{5}{24} a^{4} \alpha^{4}+\frac{5}{24} a^{2} \alpha^{2}+\frac{5}{24}\right) r^{3}\nonumber \\
&+\frac{28 m \left(a^{4} \alpha^{4}-\frac{12}{7} a^{2} \alpha^{2}+\frac{1}{7}\right) r^{2}}{3}+a^{6} \alpha^{4} m -\frac{19 a^{4} \alpha^{2} m}{2}+\frac{9 a^{2} m}{2}\Biggr), \\
&\mathcal{F}_9=-\frac{715 r^{6} a^{6}}{7} \Biggl(\left(\frac{108}{5} a^{4} \alpha^{6} m -\frac{228}{5} a^{2} m \alpha^{4}+\frac{24}{5} m \alpha^{2}\right) r^{4}+\frac{32 a^{2} \alpha^{2} m \left(a^{4} \alpha^{4}-12 a^{2} \alpha^{2}+7\right) r^{2}}{5}\nonumber \\
&+\left(a^{8} \alpha^{6}+a^{6} \alpha^{4}-a^{4} \alpha^{2}-a^{2}\right) r -\frac{98 a^{2} m \left(a^{4} \alpha^{4}-\frac{12}{7} a^{2} \alpha^{2}+\frac{1}{7}\right)}{5}\Biggr),\\
&\mathcal{F}_{10}=-286 r^{5} a^{8} \Biggl(\left(-\frac{12}{7} a^{4} \alpha^{6} m +\frac{114}{7} a^{2} m \alpha^{4}-\frac{54}{7} m \alpha^{2}\right) r^{4}+\left(\frac{5}{14} a^{6} \alpha^{6}+\frac{5}{14} a^{4} \alpha^{4}-\frac{5}{14} a^{2} \alpha^{2}-\frac{5}{14}\right) r^{3}\nonumber \\
&+\frac{54 m \left(a^{4} \alpha^{4}-\frac{92}{27} a^{2} \alpha^{2}+\frac{13}{27}\right) r^{2}}{7}+a^{2} \left(a^{2} \alpha^{2}+1\right)^{2} r +a^{2} m \left(a^{4} \alpha^{4}-8 a^{2} \alpha^{2}+3\right)\Biggr) \alpha,\\
&\mathcal{F}_{11}=117 r^{4} a^{8} \Biggl(\frac{154 m \alpha^{2} \left(a^{4} \alpha^{4}-\frac{12}{7} a^{2} \alpha^{2}+\frac{1}{7}\right) r^{4}}{9}-\frac{22 a^{2} \alpha^{2} \left(a^{2} \alpha^{2}+1\right)^{2} r^{3}}{9}+\frac{4 a^{2} \alpha^{2} m \left(a^{4} \alpha^{4}-72 a^{2} \alpha^{2}+67\right) r^{2}}{9}\nonumber \\
&+\left(a^{8} \alpha^{6}+a^{6} \alpha^{4}-a^{4} \alpha^{2}-a^{2}\right) r -\frac{16 a^{2} \left(a^{4} \alpha^{4}-\frac{17}{12} a^{2} \alpha^{2}+\frac{1}{12}\right) m}{3}\Biggr),
\end{align}
\begin{align}
&\mathcal{F}_{12}=+52 r^{3} \Biggl(-\frac{11 m \alpha^{2} \left(a^{4} \alpha^{4}-8 a^{2} \alpha^{2}+3\right) r^{4}}{2}+\left(\frac{9}{4} a^{6} \alpha^{6}+\frac{9}{4} a^{4} \alpha^{4}-\frac{9}{4} a^{2} \alpha^{2}-\frac{9}{4}\right) r^{3}+\Biggl[8 a^{4} \alpha^{4} m \nonumber \\
&-58 a^{2} \alpha^{2} m +10 m \Biggr] r^{2}+\frac{5 a^{2} \left(a^{2} \alpha^{2}+1\right)^{2} r}{2}+a^{2} m \left(a^{4} \alpha^{4}-7 a^{2} \alpha^{2}+2\right)\Biggr) a^{10} \alpha,\\
&\mathcal{F}_{13}=-25 \Biggl(\frac{624 \left(a^{4} \alpha^{4}-\frac{17}{12} a^{2} \alpha^{2}+\frac{1}{12}\right) m \alpha^{2} r^{4}}{25}-\frac{26 a^{2} \alpha^{2} \left(a^{2} \alpha^{2}+1\right)^{2} r^{3}}{5}-\frac{24 a^{2} \alpha^{2} m \left(a^{4} \alpha^{4}+33 a^{2} \alpha^{2}-38\right) r^{2}}{25}\nonumber \\
&+\left(a^{8} \alpha^{6}+a^{6} \alpha^{4}-a^{4} \alpha^{2}-a^{2}\right) r -\frac{54 a^{2} m \left(a^{4} \alpha^{4}-\frac{32}{27} a^{2} \alpha^{2}+\frac{1}{27}\right)}{25}\Biggr) r^{2} a^{10},\\
&\mathcal{F}_{14}=-2 r \,a^{12} \Biggl(-26 m \alpha^{2} \left(a^{4} \alpha^{4}-7 a^{2} \alpha^{2}+2\right) r^{4}+\left(\frac{25}{2} a^{6} \alpha^{6}+\frac{25}{2} a^{4} \alpha^{4}-\frac{25}{2} a^{2} \alpha^{2}-\frac{25}{2}\right) r^{3}\nonumber \\
&+\frac{142 m \left(a^{4} \alpha^{4}-\frac{736}{71} a^{2} \alpha^{2}+\frac{133}{71}\right) r^{2}}{7}+\frac{45 a^{2} \left(a^{2} \alpha^{2}+1\right)^{2} r}{7}+a^{2} m \left(a^{4} \alpha^{4}-\frac{44}{7} a^{2} \alpha^{2}+\frac{9}{7}\right)\Biggr) \alpha, \\
&\mathcal{F}_{15}=\Biggl(54 m \alpha^{2} \left(a^{4} \alpha^{4}-\frac{32}{27} a^{2} \alpha^{2}+\frac{1}{27}\right) r^{4}-\frac{90 a^{2} \alpha^{2} \left(a^{2} \alpha^{2}+1\right)^{2} r^{3}}{7}-\frac{12 a^{2} m \alpha^{2} \left(a^{4} \alpha^{4}+\frac{100}{3} a^{2} \alpha^{2}-\frac{115}{3}\right) r^{2}}{7}\nonumber \\
&+\left(a^{8} \alpha^{6}+a^{6} \alpha^{4}-a^{4} \alpha^{2}-a^{2}\right) r +\frac{4 a^{4} \alpha^{2} m}{7}-\frac{4 a^{6} \alpha^{4} m}{7}\Biggr) a^{12},\\
&\mathcal{F}_{16}=\Biggl(-2 \left(a^{4} \alpha^{4}-\frac{44}{7} a^{2} \alpha^{2}+\frac{9}{7}\right) m \alpha^{2} r^{3}+\left(a^{6} \alpha^{6}+a^{4} \alpha^{4}-a^{2} \alpha^{2}-1\right) r^{2}+\frac{12 m \left(a^{4} \alpha^{4}-\frac{23}{3} a^{2} \alpha^{2}+\frac{4}{3}\right) r}{7}\nonumber \\
&+\frac{a^{2} \left(a^{2} \alpha^{2}+1\right)^{2}}{7}\Biggr) a^{14} \alpha,\\
&\mathcal{F}_{17}=\frac{a^{16} \alpha^{2} \left(\left(-4 a^{2} m \alpha^{4}+4 m \alpha^{2}\right) r^{2}+\left(a^{2} \alpha^{2}+1\right)^{2} r +4 a^{2} \alpha^{2} m -4 m \right)}{7}.
\end{align}

We now compute for the first time the invariant $Q_2\equiv\frac{1}{27}\frac{I_5 I_6-I_7^2}{(I_1^2+I_2^2)^{5/2}}$ for the case of an accelerating Kerr black hole.
\begin{theorem}\label{Q2invAccelKerr}
The exact analytic expression for the invariant $Q_2$ for the accelerating Kerr black hole is the following:
\begin{align}
Q_2=-\frac{\left(a^{2}-2 m r +r^{2}\right) \left(\cos \! \left(\theta \right)^{2} a^{2} \alpha^{2}-2 \cos \! \left(\theta \right) \alpha  m +1\right) \sin \! \left(\theta \right)^{2} a^{2} \left(\alpha  r \cos \! \left(\theta \right)+1\right)^{2} \left(\alpha^{2} r^{2}-1\right)}{\left(\alpha  r \cos \! \left(\theta \right)-1\right)^{4} \left(r^{2}+\cos \! \left(\theta \right)^{2} a^{2}\right) m^{2} \left(a^{2} \alpha^{2}+1\right)}.
\end{align}
\end{theorem}
The generalisation of Theorem (\ref{Q2invAccelKerr}) in the presence of the cosmological constant is:
\begin{theorem}\label{q2detecKerrdSaccel}
The exact analytic expression for the invariant $Q_2$ for the accelerating Kerr black hole in (anti)de-Sitter spacetime is the following:
\begin{align}
Q_2&=-\frac{1}{\left(\alpha  r \cos \! \left(\theta \right)-1\right)^{4} \left(r^{2}+a^{2} \cos \! \left(\theta \right)^{2}\right) m^{2} \left(a^{2} \alpha^{2}+1\right)}\Biggl[a^{2} \left(\alpha  r \cos \! \left(\theta \right)+1\right)^{2} \sin \! \left(\theta \right)^{2} \Biggl(\left(\frac{\Lambda}{3}+\alpha^{2}\right) r^{4}\nonumber \\
&-2 m \alpha^{2} r^{3}+\left(a^{2} \alpha^{2}+\frac{1}{3} a^{2} \Lambda -1\right) r^{2}+2 m r -a^{2}\Biggr) \left(1+a^{2} \left(\frac{\Lambda}{3}+\alpha^{2}\right) \cos \! \left(\theta \right)^{2}-2 \cos \! \left(\theta \right) \alpha  m \right)\Biggr].
\label{Q2accelKdS}
\end{align}
\end{theorem}
\begin{remark}
We note that the differential invariant $Q_2$, as given in Eqn. (\ref{Q2accelKdS}), vanishes at the horizons of the accelerating rotating black hole with $\Lambda\not =0$.
\end{remark}
For zero acceleration $\alpha=0$ we derive the following corollary for the invariant $Q_2$ in Kerr-(anti)de Sitter spacetime:
\begin{corollary}\label{porismaQ2kerrdeSitter}
\begin{align}
Q_2=-\frac{a^{2} \sin \! \left(\theta \right)^{2} \left(\frac{r^{4} \Lambda}{3}+\left(-1+\frac{a^{2} \Lambda}{3}\right) r^{2}+2 m r -a^{2}\right) \left(1+\frac{a^{2} \cos \left(\theta \right)^{2} \Lambda}{3}\right)}{\left(r^{2}+a^{2} \cos \! \left(\theta \right)^{2}\right) m^{2}}.
\label{Q2KerrdeSitter1}
\end{align}
\end{corollary}

In the regional plots of Fig. \ref{DifInvaI4region} and Fig.\ref{DifInvaI4region09939} we determine the sign of the curvature invariant $I_4$ for the case of Kerr-(anti-)de Sitter black hole for two sets of values of the physical black hole parameters $a,\Lambda,m$ .
Also in Fig.\ref{DifInvaI7region09939} we determine the sign of the local invariant $I_7$ for the parameters: $a=0.9939,\Lambda=3.6\times 10^{-33},m=1,q=0,\alpha=0$.
In Fig.\ref{DifInvaQ2} we exhibit a three-dimensional plot for the scalar polynomial curvature invariant $Q_2$  as a function of the Boyer-Lindquist coordinates $r$ and $\theta$, for the case of a Kerr black hole with a cosmological constant for the set of values: $a=0.9939,\Lambda=3.6\times 10^{-33},m=1$.

\begin{figure}[ptbh]
\centering
  \begin{subfigure}[b]{.60\linewidth}
    \centering
    \includegraphics[width=.99\textwidth]{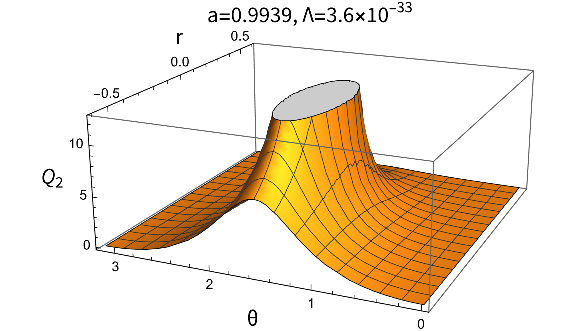}
    \caption{3D plot of $Q_2$ for  Kerr-de Sitter BH.}\label{DifInvaQ2}
  \end{subfigure}%
  \begin{subfigure}[b]{.60\linewidth}
    \centering
    \includegraphics[width=.99\textwidth]{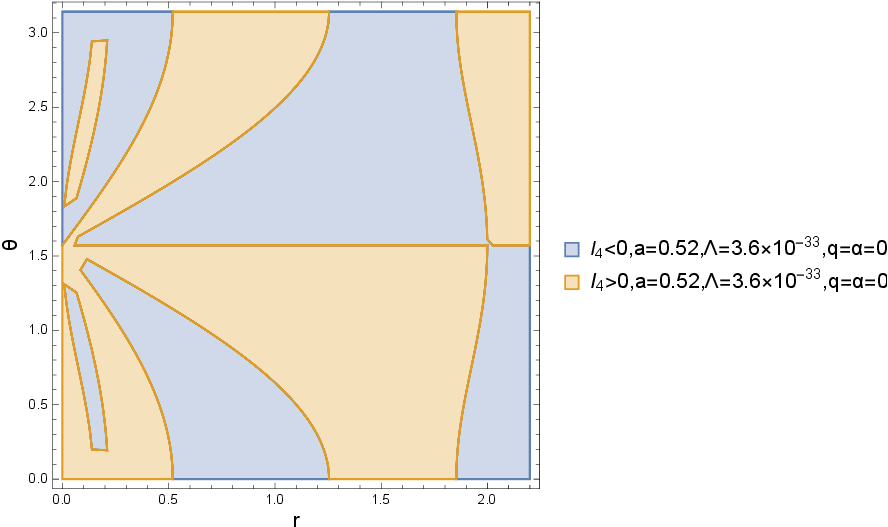}
    \caption{Region plot of  $I_4$ for Kerr-de Sitter BH.}\label{DifInvaI4region}
  \end{subfigure}\\
  \begin{subfigure}[b]{.60\linewidth}
    \centering
    \includegraphics[width=.99\textwidth]{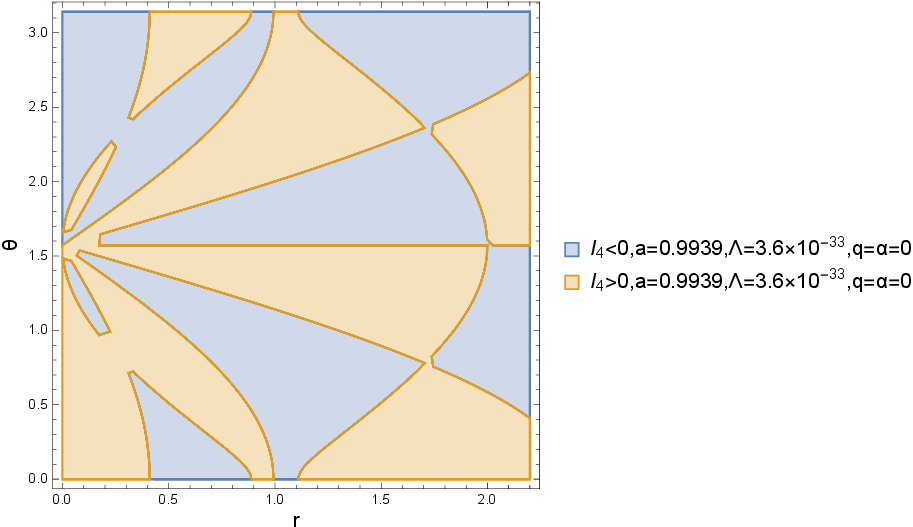}
    \caption{Region Plot of $I_4$ for Kerr-de Sitter BH.}\label{DifInvaI4region09939}
  \end{subfigure}%
  \begin{subfigure}[b]{.60\linewidth}
    \centering
    \includegraphics[width=.99\textwidth]{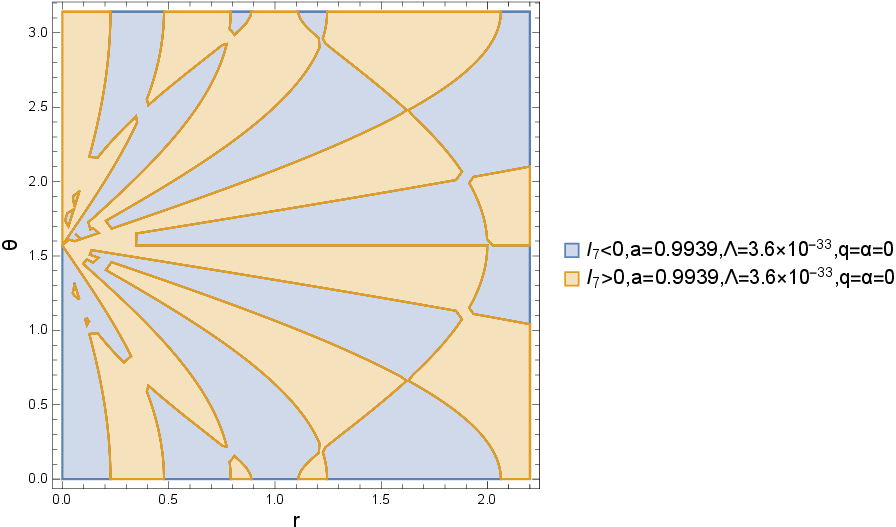}
    \caption{Region Plot of $I_7$ for  Kerr-de Sitter BH.}\label{DifInvaI7region09939}
  \end{subfigure}
  \caption{ (a) The invariant $Q_2$  plotted as a function of the Boyer-Lindquist coordinates $r$ and $\theta$ for a rotating black hole in the presence of the cosmological constant for the values $a=0.9939,\Lambda=3.6\times 10^{-33},m=1$  (b) Regions of negative and positive sign of $I_4$  for the values $a=0.52,\Lambda=3.6\times 10^{-33},m=1,q=0,\alpha=0$ . (c) Regions of negative and positive sign of $I_4$  for the values $a=0.9939,\Lambda=3.6\times 10^{-33},m=1,q=0,\alpha=0$ . (d) Regions of negative and positive differential invariant $I_7$  for the values $a=0.9939,\Lambda=3.6\times 10^{-33},m=1,q=0,\alpha=0$ .}\label{Total3dQ2signI4I7}
\end{figure}

In Fig.\ref{DifInvaQ2Region} we display the region plot for the invariant $Q_2$, eqn.(\ref{Q2KerrdeSitter1}) of Corollary \ref{porismaQ2kerrdeSitter}. We observe that $Q_2$ changes sign on the event and Cauchy horizons of the Kerr-de Sitter black hole. It vanishes at the stationary horizons and is non-zero everywhere else, a fact that we shall prove later in the main text.
In Fig.\ref{DifInvaQ2RegionAccel001} we exhibit the regions of negative and positive sign for the curvature invariant $Q_2$,eqn.(\ref{Q2accelKdS}) of Theorem \ref{q2detecKerrdSaccel} for the choice of values for the parameters: $a=0.52,\Lambda=3.6\times 10^{-33},\alpha=0.01,m=1$. We observe that $Q_2$ changes sign on the event and Cauchy horizons of the accelerating Kerr-de Sitter black hole.
In the region plot, Fig.\ref{DifInvaQ2RegionAccel01}, we determine the sign of the curvature invariant $Q_2$, for the choice of values for the parameters: $a=0.52,\Lambda=3.6\times 10^{-33},\alpha=0.1,m=1$. It is evident from the plot that $Q_2$ vanishes and changes sign at the event, Cauchy and acceleration  horizons of the accelerating, rotating black hole in (anti-)de Sitter spacetime. We repeated the analysis for a higher spin value of the black hole, in Fig.\ref{DifInvaQ2Regiona09939Accel01} and Fig.\ref{DifInvaQ2Regiona0939Accel01bH}. For  the choice of values for the parameters: $a=0.9939,\Lambda=3.6\times 10^{-33},\alpha=0.1,m=1$, Fig.\ref{DifInvaQ2Regiona09939Accel01} and Fig.\ref{DifInvaQ2Regiona0939Accel01bH}, exhibit the sign change of the invariant $Q_2$ at the horizons radii.

\begin{figure}[ptbh]
\centering
  \begin{subfigure}[b]{.60\linewidth}
    \centering
    \includegraphics[width=.99\textwidth]{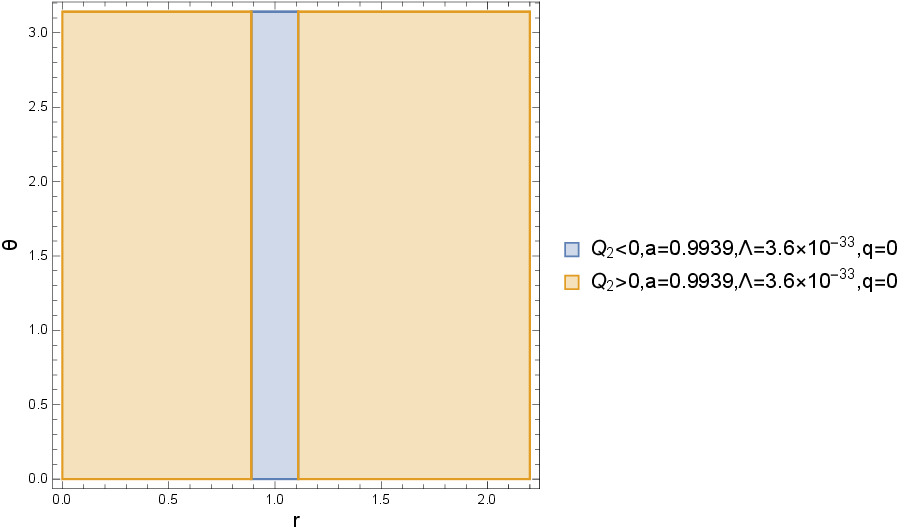}
    \caption{Region plot of $Q_2$ for non-accelerating Kerr-de Sitter BH.}\label{DifInvaQ2Region}
  \end{subfigure}%
  \begin{subfigure}[b]{.60\linewidth}
    \centering
    \includegraphics[width=.99\textwidth]{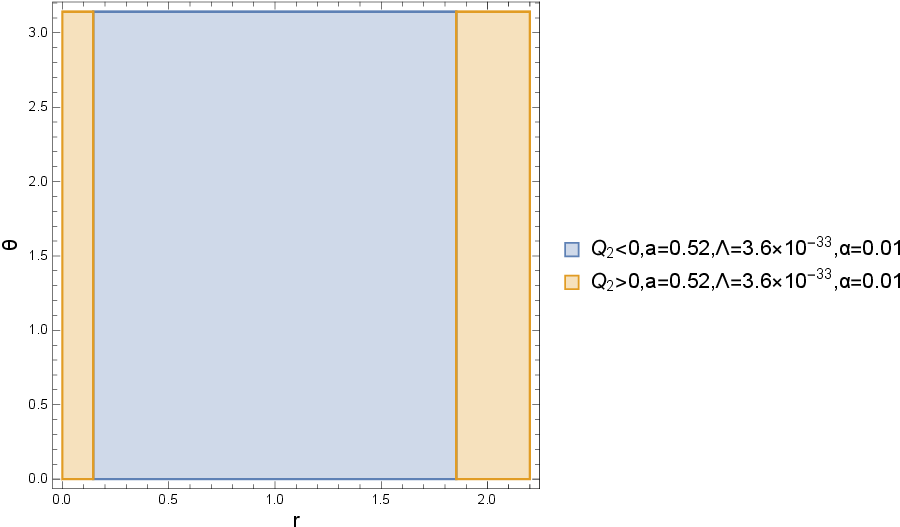}
    \caption{Region plot of  $Q_2$ for accelerating  Kerr-de Sitter BH.}\label{DifInvaQ2RegionAccel001}
  \end{subfigure}\\
  \begin{subfigure}[b]{.60\linewidth}
    \centering
    \includegraphics[width=.99\textwidth]{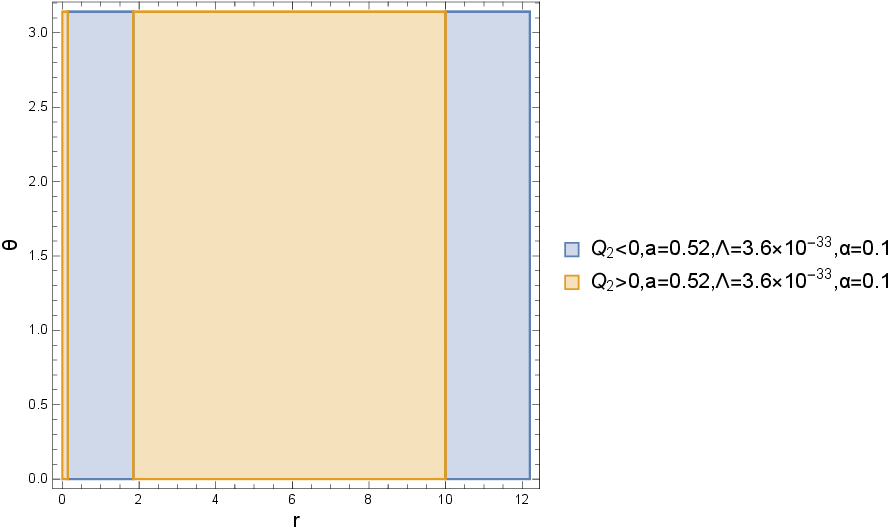}
    \caption{Region Plot of $Q_2$ for accelerating Kerr-de Sitter BH.}\label{DifInvaQ2RegionAccel01}
  \end{subfigure}%
  \begin{subfigure}[b]{.60\linewidth}
    \centering
    \includegraphics[width=.99\textwidth]{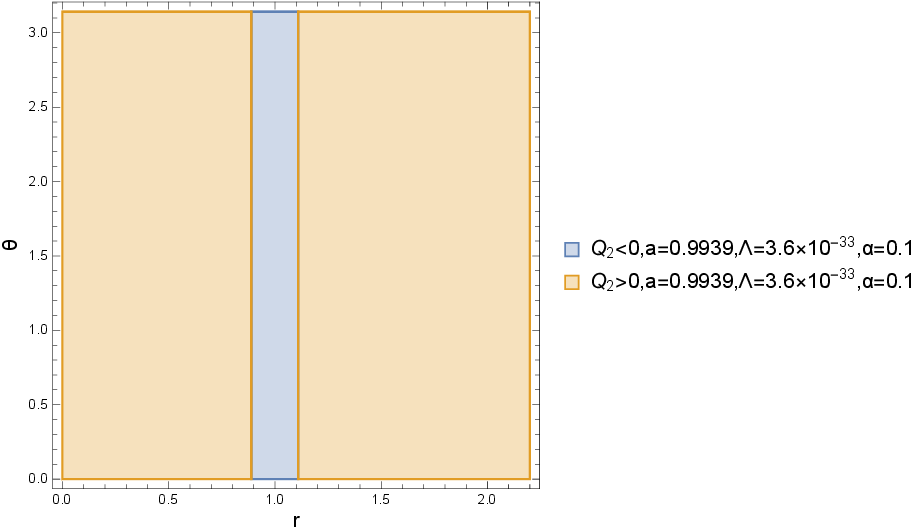}
    \caption{Region Plot of $Q_2$ for accelerating Kerr-de Sitter BH.}\label{DifInvaQ2Regiona09939Accel01}
  \end{subfigure}\\
  \begin{subfigure}[b]{.60\linewidth}
    \centering
    \includegraphics[width=.99\textwidth]{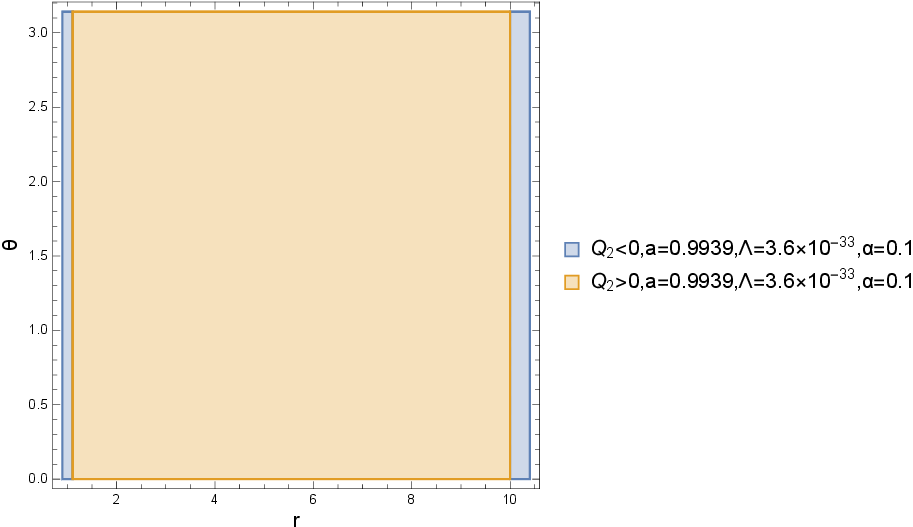}
    \caption{Region Plot of $Q_2$ for accelerating Kerr-de Sitter BH.}\label{DifInvaQ2Regiona0939Accel01bH}
  \end{subfigure}
  \caption{Region plots of the  curvature invariant $Q_2$ , for the accelerating Kerr black hole with $\Lambda\not=0$. (a) for the values $a=0.9939,\Lambda=3.6\times 10^{-33},m=1$ . (b) Region plot for the invariant $Q_2$ in the presence of the cosmological constant, Eqn. (\ref{Q2accelKdS}) , for the values $a=0.52,\Lambda=3.6\times 10^{-33},\alpha=0.01,m=1$. (c) Region plot for the curvature invariant $Q_2$  in the presence of the cosmological constant,Eqn. (\ref{Q2accelKdS}), for the values $a=0.52,\Lambda=3.6\times 10^{-33},\alpha=0.1,m=1$. We observe that the curvature invariant $Q_2$ vanishes and changes sign at the event, Cauchy and acceleration (modified due to $\Lambda$) horizons . (d) Region plot for the curvature invariant $Q_2$  in the presence of  $\Lambda$ for the values $a=0.9939,\Lambda=3.6\times 10^{-33},\alpha=0.1,m=1$. (e) Region plot for the curvature invariant $Q_2$  in the presence of the cosmological constant for the values $a=0.9939,\Lambda=3.6\times 10^{-33},\alpha=0.1,m=1$.}\label{RegionQ2KerrLambda}
\end{figure}

The most general explicit expression for the invariant $Q_2$ for the accelerating Kerr-Newman black hole in (anti-)de Sitter spacetime is given in Theorem \ref{exactQ2accelKNdS} and eqn.(\ref{totalQ2accelKN(a)dS}).

Now we will prove that the invariant $Q_2$ can be used to detect physically relevant surfaces such as horizons. We focus for simplicity on the invariant $Q_2$ for the case of the Kerr-de Sitter black hole eqn.(\ref{Q2KerrdeSitter1}), corollary \ref{porismaQ2kerrdeSitter}. For the proof we will need the following theorems:
\begin{theorem}\label{rene}
Descartes's Rule of Signs \cite{descartes}.
Let $f(x)=a_0x^{b_0}+a_1x^{b_1}+\cdots+a_nx^{b_n}$ a polynomial with nonzero real coefficients $a_i$, where the $b_i$ are integers satisfying $0\leq b_0<b_1<b_2<\cdots b_n$. Then the number of positive real zeros of $f(x)$ (counted with multiplicities) is either equal to the number of variations in sign in the sequence $a_0,\ldots,a_n$ of the  coefficients or less that by an even whole number.  The number of negative zeros of $f(x)$ (counted with multiplicities) is either equal to the number of variations in sign in the sequence of the coefficients of $f(-x)$ or less than that by an even whole number.
\end{theorem}
Denoting the number of variations in the signs of the sequence of the coefficients of $f$ by $V[(f(x)]$ and the number of positive zeros of $f$ counting multiplicities by $Pz[f(x)]$. Then it is valid the following:
\begin{lemma}
Let $f(x)=\sum_{i=0}^na_i x^{b_i}$ be a polynomial as in the theorem \ref{rene}. If $a_0a_n>0$, then $Pz[f(x)]$ is even: if $a_0a_n<0$, than $Pz[f(x)]$ is odd.
\end{lemma}
\begin{theorem}\label{bernardbolzano}
Bolzano. If $f:[\alpha,\beta]\rightarrow \mathbb{R}$ is a continuous function on the closed interval $[\alpha,\beta]$ and $f(\alpha)f(\beta)<0$ then there exists $\xi\in (\alpha,\beta)$, such that $f(\xi)=0$.
\end{theorem}
\begin{proof}
We focus on the radial polynomial $f(r)=\frac{r^{4} \Lambda}{3}+\left(-1+\frac{a^{2} \Lambda}{3}\right) r^{2}+2 m r -a^{2}=-Q(r)$. All the other terms in the expression of $Q_2$ are non-zero for $\theta\in(0,\pi)$. Now we use Descartes's Rule of Signs, Theorem \ref{rene}. For a value of $\Lambda$ consistent with observations, the sequence of signs of the coefficients of $f(r)$ is: $+-+-$ in skeletal notation.  We note that there are three sign changes in the radial polynomial $f(r)$. On the other hand, when $r\rightarrow-r$, there occurs one sign change in $f(-r)$. Thus, there are three positive real roots and one negative for the polynomial $f(r)$. We conclude that the invariant $Q_2$ vanishes at the stationary horizons since $f(r)$ vanishes there, and since every polynomial is a continuous function, $Q_2$ is non-zero everywhere else from Bolzano's theorem \ref{bernardbolzano}.
\end{proof}

\begin{theorem}
As regards the invariant $Q_3$ for an accelerating Kerr black hole in (anti-)de Sitter spacetime our analytic computation yields the following explicit algebraic expression:
\begin{align}
&Q_3=\frac{1}{6 \sqrt{\frac{\left(\alpha  r \cos \left(\theta \right)-1\right)^{6} m^{2} \left(a^{2} \alpha^{2}+1\right)}{\left(r^{2}+a^{2} \cos \left(\theta \right)^{2}\right)^{3}}} \left(r^{2}+a^{2} \cos \! \left(\theta \right)^{2}\right)^{2}}\Biggl[-3 a^{2} \Biggl(\left(\alpha^{2} \left(\frac{\Lambda}{3}+\alpha^{2}\right) r^{2}+\frac{\Lambda}{3}\right) a^{2}\nonumber \\
&-2 \alpha^{2} r \left(\left(-\alpha^{2}-\frac{\Lambda}{3}\right) r^{3}+\alpha^{2} m \,r^{2}+\frac{r}{2}-m \right)\Biggr) \cos \! \left(\theta \right)^{4}+6 m \alpha  \left(\alpha^{2} r^{4}+a^{2}\right) \cos \! \left(\theta \right)^{3}\nonumber \\
&+3 \left(-1+\left(\frac{\Lambda}{3}+\alpha^{2}\right) a^{2}\right) \left(\alpha^{2} r^{4}+a^{2}\right) \cos \! \left(\theta \right)^{2}-6 m \alpha  \left(\alpha^{2} r^{4}+a^{2}\right) \cos \! \left(\theta \right)\nonumber \\
&+\left(6+\left(-3 \alpha^{2}-\Lambda \right) r^{2}\right) a^{2}+6 \alpha^{2} m \,r^{3}-\Lambda  r^{4}-6 m r +3 r^{2}\Biggr].
\end{align}
\end{theorem}
\section{The Newman-Penrose formalism and differential curvature invariants}\label{RogerPenN}
In this section we shall apply the Newman-Penrose formalism \cite{NPformalismCI} and compute in an independent way the differential local invariants for the accelerating Kerr-Newman black hole in (anti-)de Sitter spacetime.
 It is gratifying that our results in NP formalism agree with those obtained for the corresponding curvature invariants in the previous sections.

From the relation $I_7=k_{\mu}l^{\mu}$, and the equation:
\begin{equation}
I_2=24 i(\Psi_2^2-\overline{\Psi}_2^2),
\end{equation}
we obtain:
\begin{align}
I_7&=48^2 i(\Psi_2^2\nabla_{\mu}\Psi_2\nabla^{\mu}\Psi_2-\overline{\Psi_2}^2\nabla_{\mu}\overline{\Psi_2}\nabla^{\mu}\overline{\Psi_2})\nonumber \\
&=-48^2\; 2\; \Im{(\Psi_2^2\nabla_{\mu}\Psi_2\nabla^{\mu}\Psi_2)},
\end{align}
where the Weyl scalar $\Psi_2$ is given by \cite{KraniotisCurvature}:
\begin{equation}
\Psi_2=-\frac{\left(1-\alpha  r \cos \! \left(\theta \right)\right)^{3} \left(m \left(i a \alpha +1\right) \left(r +i a \cos \! \left(\theta \right)\right)-q^{2} \left(1+\alpha  r \cos \! \left(\theta \right)\right)\right)}{\left(r -i a \cos \! \left(\theta \right)\right)^{3} \left(r +i a \cos \! \left(\theta \right)\right)}
\end{equation}

The functional relationship between invariants, the syzygy in Eqn.(\ref{erstesyzygie}), that we discovered  in this work also holds for the Kerr-de Sitter black hole spacetime may be expressed as the real part of the complex syzygy:
\begin{equation}
\nabla_{\mu}(I_1+iI_2)\nabla^{\mu}(I_1+iI_2)=\frac{12}{5}(I_1+iI_2)(I_3+iI_4),
\label{complexerste}
\end{equation}

On the other hand, in the NP formalism we have that:
\begin{align}
I_1+iI_2=48\Re{\Psi_2^2}-48i \Im{\Psi_2^2}=48\overline{\Psi_2}^2,
\end{align}
thus we derive the following covariant derivative:
\begin{equation}
\nabla_{\mu}(I_1+iI_2)=48\times 2\overline{\Psi}_2\nabla_{\mu}\overline{\Psi_2},
\end{equation}
and from (\ref{complexerste}) we derive:
\begin{align}
(96\overline{\Psi}_2)^2\nabla_{\mu}\overline{\Psi}_2\nabla^{\mu}\overline{\Psi}_2=
\frac{12\cdot48}{5}\overline{\Psi_2}^2(I_3+iI_4).
\end{align}
Thus we conclude that:
\begin{align}
I_3&=80\Re\nabla_{\mu}\overline{\Psi_2}\nabla^{\mu}\overline{\Psi_2},\\
I_4&=80\Im\nabla_{\mu}\overline{\Psi_2}\nabla^{\mu}\overline{\Psi_2}.
\end{align}
The computation of the differential invariants $I_5$ and $I_6$ in the NP formalism yields the result:
\begin{align}
I_5&=\nabla_{\mu}I_1\nabla^{\mu}I_1\nonumber \\
&=48^2 \Biggl[\Psi_2^2 (\nabla_{\mu}\Psi_2\nabla^{\mu}\Psi_2)+2\Psi_2\overline{\Psi_2}\nabla_{\mu}\Psi_2\nabla^{\mu}\overline{\Psi_2}+\overline{\Psi_2}^2\nabla_{\mu}\overline{\Psi_2}\nabla^{\mu}\overline{\Psi_2}\Biggr]\label{iotafunf},\\
I_6&=\nabla_{\mu}I_2\nabla^{\mu}I_2 ,\nonumber \\
&=-48^2 (\Psi_2^2 (\nabla_{\mu}\Psi_2\nabla^{\mu}\Psi_2)-2\Psi_2\overline{\Psi_2}\nabla_{\mu}\Psi_2\nabla^{\mu}\overline{\Psi_2}
+\overline{\Psi_2}^2\nabla_{\mu}\overline{\Psi_2}\nabla^{\mu}\overline{\Psi_2}) \label{sixNP}.
\end{align}

The invariant $Q_3$ as defined in Eqn.(\ref{Qdrei}), in the NP formalism and taking into account Eqns.(\ref{iotafunf})-(\ref{sixNP}), acquires the form:
\begin{align}
Q_3&=\frac{4\cdot 48^2\nabla_{\mu}\Psi_2\nabla^{\mu}\overline{\Psi_2}}{6\sqrt{3}\;48^{5/2}(\Psi_2\overline{\Psi_2})^{3/2}}\nonumber \\
&=\frac{1}{18}\frac{\nabla_{\mu}\Psi_2\nabla^{\mu}\overline{\Psi_2}}{(\Psi_2\overline{\Psi_2})^{3/2}},
\end{align}
whereas
\begin{align}
Q_1=\frac{2\Re(\overline{\Psi_2}^2 \nabla_{\mu}\Psi_2\nabla^{\mu}\Psi_2)}{9(\Psi_2\overline{\Psi_2})^{5/2}}.
\end{align}
Likewise, using the observation in Eqn.(\ref{donandrey}) and the theorem derived in \cite{PageD} we can derive an expression of the curvature invariant $Q_2$ in terms of the norm of the wedge product that involves exterior derivatives of the Weyl scalar $\Psi_2$ \cite{MAHMacCallum}:
\begin{equation}
Q_2=\frac{-2\left\lVert\nabla \overline {\Psi_2}\wedge \nabla\Psi_2\right\Vert^2}{18^2(\Psi_2\overline{\Psi_2})^3}.
\end{equation}

\begin{theorem}\label{exactQ2accelKNdS}
We calculated the exact algebraic expression for the invariant $Q_2$ for the accelerating Kerr-Newman black hole in (anti-)de Sitter spacetime.  Our result is:
\begin{align}
&Q_2=\frac{-1}{\mathcal{A}}\Biggl(\alpha  \left(\left(a^{4} m^{2} r +2 a^{2} m \,q^{2} r^{2}+\frac{8}{9} q^{4} r^{3}\right) \alpha^{2}+a^{2} m \left(m r -\frac{2 q^{2}}{3}\right)\right) \cos \! \left(\theta \right)^{3}\nonumber \\
&+\left(\left(a^{4} m^{2}+2 a^{2} m \,q^{2} r -2 m \,q^{2} r^{3}+\frac{16}{9} q^{4} r^{2}\right) \alpha^{2}+a^{2} m^{2}\right) \cos \! \left(\theta \right)^{2}+\alpha  \Biggl(a^{2} \alpha^{2} m^{2} r^{3}+2 a^{2} m \,q^{2}+r^{3} m^{2}\nonumber \\
&-2 m \,q^{2} r^{2}+\frac{16}{9} q^{4} r \Biggr) \cos \! \left(\theta \right)+m \,r^{2} \left(a^{2} m +\frac{2 q^{2} r}{3}\right) \alpha^{2}+r^{2} m^{2}-2 m \,q^{2} r +\frac{8 q^{4}}{9}\Biggr)^{2} \Biggl[r^{2} \left(a^{2}-2 m r +q^{2}+r^{2}\right) \alpha^{2}\nonumber \\
&+\frac{\Lambda  r^{4}}{3}+\left(\frac{\Lambda  a^{2}}{3}-1\right) r^{2}+2 m r -a^{2}-q^{2}\Biggr] a^{2} \sin \! \left(\theta \right)^{2} \left(1-2 \alpha  m \cos \! \left(\theta \right)+\left(\alpha^{2} \left(a^{2}+q^{2}\right)+\frac{\Lambda  a^{2}}{3}\right) \cos \! \left(\theta \right)^{2}\right).
\label{totalQ2accelKN(a)dS}
\end{align}
\end{theorem}
where
\begin{align}
\mathcal{A}&\equiv\left(\alpha  r \cos \! \left(\theta \right)-1\right)^{4} \Biggl(\left(\left(a^{2} m +q^{2} r \right)^{2} \alpha^{2}+a^{2} m^{2}\right) \cos \! \left(\theta \right)^{2}+2 q^{2} \alpha  \left(a^{2} m -m \,r^{2}+q^{2} r \right) \cos \! \left(\theta \right)\nonumber \\
&+r^{2} m^{2} \alpha^{2} a^{2}+\left(m r -q^{2}\right)^{2}\Biggr)^{3}.
\end{align}

\begin{corollary}
We observe that Eqn.(\ref{totalQ2accelKN(a)dS}) vanishes at the radii of the horizons of the accelerating Kerr-Newman black hole in (anti-)de Sitter spacetime. Thus, the differential curvature invariant $Q_2$ can serve as horizon detector for the most general class of accelerating rotating and charged black holes with non-zero cosmological constant.
\end{corollary}
In Fig.\ref{Q2AccelKNdS} we display $Q_2$ for various values of the physical parameters for an accelerating Kerr-Newman black hole in (anti-)de Sitter spacetime. The vanishing of $Q_2$ at the stationary horizons and the change of its sign as we cross them is illustrated in Figs.\ref{graphQ2a026q085Lobsallpha01},\ref{graphQ2a052q011alpha01Lobs},\ref{grafosq2a09939q011al01Lobs}.

\begin{corollary}
In the equatorial plane $\theta=\frac{\pi}{2}$, $Q_2$ reduces to:
\begin{align}
Q_2(\theta=\frac{\pi}{2})&=-\frac{\left(m \,r^{2} \left(a^{2} m +\frac{2 q^{2} r}{3}\right) \alpha^{2}+r^{2} m^{2}-2 m \,q^{2} r +\frac{8 q^{4}}{9}\right)^{2} }{\left(r^{2} m^{2} \alpha^{2} a^{2}+\left(m r -q^{2}\right)^{2}\right)^{3}}\nonumber \\
&\times \Biggl(r^{2} \left(a^{2}-2 m r +q^{2}+r^{2}\right) \alpha^{2}+\frac{\Lambda  r^{4}}{3}+\left(\frac{\Lambda  a^{2}}{3}-1\right) r^{2}+2 m r -a^{2}-q^{2}\Biggr) a^{2}.
\label{ishmerinhQ2}
\end{align}
The vanishing of $Q_2$ in the equatorial plane singles out the horizons (away from the discrete roots of the cubic radial polynomial in (\ref{ishmerinhQ2})).
\end{corollary}
On the axis $Q_2=0$.

\begin{figure}[ptbh]
\centering
  \begin{subfigure}[b]{.60\linewidth}
    \centering
    \includegraphics[width=.99\textwidth]{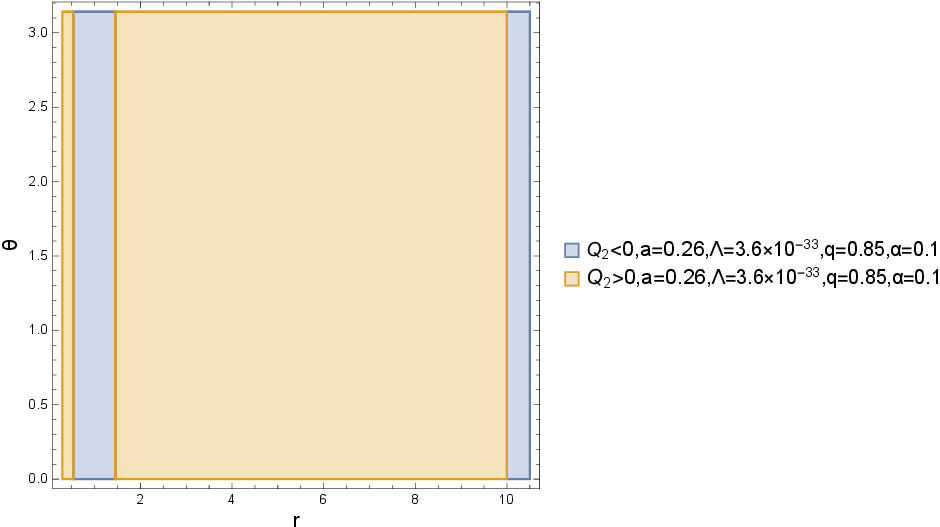}
    \caption{Region plot of  $Q_2$ for accelerating KNdS BH.}\label{graphQ2a026q085Lobsallpha01}
  \end{subfigure}%
  \begin{subfigure}[b]{.60\linewidth}
    \centering
    \includegraphics[width=.99\textwidth]{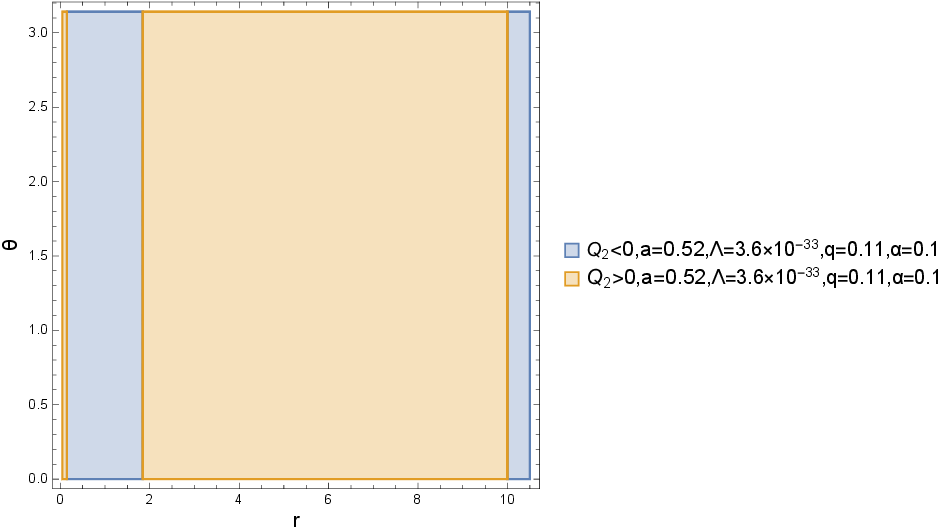}
    \caption{Region plot of $Q_2$ for accelerating KNdS BH}\label{graphQ2a052q011alpha01Lobs}
  \end{subfigure}\\
  \begin{subfigure}[b]{.60\linewidth}
    \centering
    \includegraphics[width=.99\textwidth]{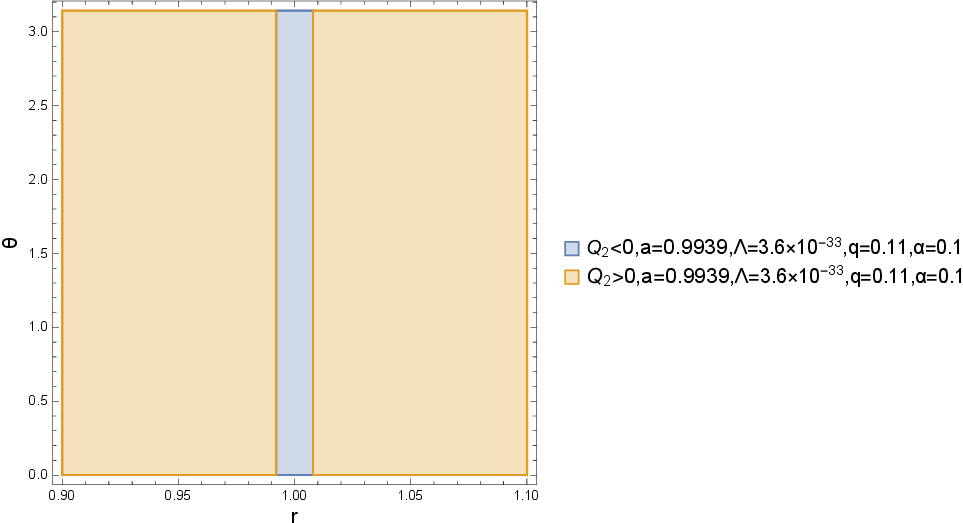}
    \caption{Region plot of $Q_2$ for accelerating KNdS BH.}\label{grafosq2a09939q011al01Lobs}
  \end{subfigure}%
  \caption{Regional plots for an accelerating Kerr-Newman black hole in (anti-)de Sitter spacetime, of the differential curvature invariant  $Q_2$  eqn.(\ref{totalQ2accelKN(a)dS}), that serves as a horizon detector, plotted as a function of the Boyer-Lindquist coordinates $r$ and $\theta$ . (a) for low spin Kerr parameter $a=0.26$, electric charge $q=0.85$, $\Lambda=3.6\times 10^{-33}$,$m=1$  acceleration $\alpha=0.1$. (b) For moderate spin  $a=0.52$, charge $q=0.11$,$\Lambda=3.6\times 10^{-33}$, $m=1$, $\alpha=0.1$ . (c) For high spin  $a=0.9939$, electric charge $q=0.11$,$\Lambda=3.6\times 10^{-33}$, mass $m=1$ and $\alpha=0.1$.}\label{Q2AccelKNdS}
\end{figure}

\subsection{Analytic computation of the Page-Shoom invariant $W$ in the NP formalism for accelerating Kerr-Newman black holes in (anti-)de Sitter spacetime}\label{winvps}

The covariant derivative operator $\nabla_{\mu}$ may be expressed in the form \cite{bicakpravda}:
\begin{equation}
\nabla_{\mu}=n_{\mu}D+l_{\mu}\Delta-\bar{m}_{\mu}\delta-m_{\mu}\delta^{*}.
\label{covderivation}
\end{equation}

Bianchi identities in terms of the spin coefficients and the Weyl and Ricci scalars read as follows \cite{Chandrasekhar}:
\begin{align}
-D\Psi_2+3\varrho\Psi_2+2\varrho\Phi_{11}-2D\Lambda&=0,
\label{BianchiEin}\\
-\Delta \Psi_2-3\mu \Psi_2-2\mu \Phi_{11}-2\Delta \Lambda&=0,\\
-\delta^{*}\Psi_2-3\pi \Psi_2+2\pi \Phi_{11}-2\delta^{*}\Lambda&=0,\\
-\delta\Psi_2+3\tau \Psi_2-2\tau \Phi_{11}-2\delta\Lambda&=0.
\label{BianchiVier}
\end{align}

In order to calculate the Ricci rotation coefficients that appear in Bianchi identities, Eqns.(\ref{BianchiEin})-(\ref{BianchiVier}), we will work with the following null-tetrad:
\begin{align}
l^{\mu}&=\left(-\frac{\sqrt{2}}{2}\frac{(r^2+a^2)\Omega}{\rho\sqrt{Q}},\frac{-\Omega\sqrt{2}\sqrt{Q}}{2\rho},0,-\frac{\Omega\sqrt{2}a}{2\sqrt{Q}\rho}\right),\label{symment1}\\
n^{\mu}&=\left(-\frac{\Omega \sqrt{2}(a^2+r^2)}{2\rho\sqrt{Q}},\frac{\Omega\sqrt{Q}\sqrt{2}}{2\rho},0,-\frac{\Omega\sqrt{2}a}{2\sqrt{Q}\rho}\right),\\
m^{\mu}&=\left(-\frac{\Omega a \sin(\theta)\sqrt{2}}{2\rho\sqrt{P}},0,-\frac{i\sqrt{2}\Omega\sqrt{P}}{2\rho},-\frac{\Omega\sqrt{2}}{2\sqrt{P}\rho\sin(\theta)}\right),\\
\bar{m}^{\mu}&=\left(-\frac{\Omega\sqrt{2}a \sin(\theta)}{2\rho\sqrt{P}},0,\frac{i\Omega \sqrt{2}\sqrt{P}}{2\rho},-\frac{\Omega \sqrt{2}}{\sqrt{P}2\rho\sin(\theta)}\right)\label{symmetnt2}.
\end{align}

 We computed the Ricci-rotation coefficients via  the formula for the $\lambda$-symbols given by Chandrasekhar \cite{Chandrasekhar}:
\begin{equation}
\lambda_{(a)(b)(c)}=e_{(b)i,j}[e_{(a)}^{\;i}e_{(c)}^{\;j}-e_{(a)}^{\;j}e_{(c)}^{\;i}].
\end{equation}

This formula has the advantage that one has to calculate ordinary derivatives of the dual co-tetrad.
Computation of the dual co-tetrad yields:
\begin{align}
\mathbf{l}&=\frac{\sqrt{Q}}{\sqrt{2}\Omega\rho}{\rm d}t-\frac{\rho}{\sqrt{Q}\Omega\sqrt{2}}{\rm d}r-\frac{\sqrt{Q}a \sin(\theta)^2}{\Omega\rho\sqrt{2}}{\rm d}\phi,\\
\mathbf{n}&=\frac{\sqrt{Q}}{\sqrt{2}\Omega\rho}{\rm d}t+\frac{\rho}{\sqrt{Q}\Omega\sqrt{2}}{\rm d}r-\frac{\sqrt{Q}a \sin(\theta)^2}{\Omega\rho\sqrt{2}}{\rm d}\phi,\\
\mathbf{m}&=\frac{\sin(\theta)\sqrt{P}a}{\sqrt{2}\rho\Omega}{\rm d}t+\frac{-i\rho}{\Omega\sqrt{2}\sqrt{P}}{\rm d}\theta-\frac{\sin(\theta)\sqrt{P}(r^2+a^2)}{\Omega \rho\sqrt{2}}{\rm d}\phi,\\
\mathbf{\bar{m}}&=\frac{\sin(\theta)\sqrt{P}a}{\sqrt{2}\rho\Omega}{\rm d}t+\frac{i\rho}{\Omega\sqrt{2}\sqrt{P}}{\rm d}\theta-\frac{\sin(\theta)\sqrt{P}(r^2+a^2)}{\Omega \rho\sqrt{2}}{\rm d}\phi.
\end{align}

The Ricci rotation coefficients $\gamma_{(a)(b)(c)}$ are expressed through the $\lambda$-coefficients as follows:
\begin{equation}
\gamma_{(a)(b)(c)}=\frac{1}{2}[\lambda_{(a)(b)(c)}+\lambda_{(c)(a)(b)}-\lambda_{(b)(c)(a)}]
\end{equation}

Using the null-tetrad in Eqns.(\ref{symment1})-(\ref{symmetnt2}), for the accelerating Kerr-Newman black hole in (anti-)de Sitter spacetime, we computed for the \textit{first time} the following Ricci spin coefficients that appear in the Bianchi identities:
\begin{align}
\mu&=\varrho=-\frac{\sqrt{Q}\alpha\cos(\theta)}{\sqrt{2}\rho}-\frac{-\Omega\sqrt{Q}(r-ia\cos(\theta)}{\sqrt{2}\rho^3}\label{symframespincoef1},\\
\tau&=\pi=\frac{i\sqrt{P}}{\sqrt{2}\rho}\alpha r\sin(\theta)-\frac{\sin(\theta)\sqrt{P}\Omega a}{\sqrt{2}\rho^3}(r-i a \cos(\theta).
\label{symframespin2}
\end{align}

The only non-zero curvature scalars in the NP-formalism for the metric (\ref{epitaxmelaniopi}) using the null-tetrad in Eqns.(\ref{symment1})-(\ref{symmetnt2}), are
the Weyl scalar $\Psi_2$ which is given by the following closed form expression:
\begin{align}
\Psi_2&\equiv C_{\mu\nu\lambda\sigma}\overline{m}^{\mu}n^{\nu}l^{\lambda}m^{\sigma}\nonumber \\
&=  -\frac{\left(\alpha  r \cos \! \left(\theta \right)-1\right)^{3} \left(\left(\alpha  r \,q^{2}+\left(a \alpha +i\right) m a \right) \cos \! \left(\theta \right)+m \left(i a \alpha -1\right) r +q^{2}\right)}{-\cos \! \left(\theta \right)^{4} a^{4}+2 \,i a^{3} r \cos \! \left(\theta \right)^{3}+2 \,i r^{3} a \cos \! \left(\theta \right)+r^{4}}\nonumber \\
&=\frac{\left(-m \left(1-i a \alpha \right)+\frac{q^{2} \left(1+\alpha  r \cos \left(\theta \right)\right)}{r -i a \cos \left(\theta \right)}\right) \left(1-\alpha  r \cos \! \left(\theta \right)\right)^{3}}{\left(r +i a \cos \! \left(\theta \right)\right)^{3}}\label{hweylpsi2},
\end{align}
and the Ricci-NP scalars:
\begin{equation}
\Phi_{11}\equiv\frac{1}{4}R_{\mu\nu}(l^{\mu}n^{\nu}+m^{\mu}\overline{m}^{\nu})=\frac{\left(\alpha  r \cos \! \left(\theta \right)-1\right)^{4} q^{2}}{2 \left(r^{2}+a^{2} \cos \! \left(\theta \right)^{2}\right)^{2}},\;\;\;\text{and}\;\;\;\Lambda\label{riccivathm11}.
\end{equation}

Now we have at our disposal all the arsenal necessary to prove the following theorem of the Page-Shoom invariant for the case of accelerating Kerr-Newman black holes in (anti-)de Sitter spacetime:
\begin{theorem}\label{cohompageshoom}
We computed in closed-form the Page-Shoom invariant for accelerating, rotating and charged black holes with non-zero cosmological constant ($\Lambda\not=0$) in the Newman-Penrose formalism  with the result:
\begin{align}
W&\equiv \left\lVert {\rm d}\Psi_2\wedge {\rm d}\overline{\Psi}_2 \right\rVert^2=
16\Re[(3\mu\Psi_2+2\mu\Phi_{11})^2(-3\bar{\pi}\overline{\Psi}_2+2\bar{\pi}\Phi_{11})^2]\nonumber \\
&-16\lVert3\mu\Psi_2+2\mu\Phi_{11}\rVert^2\lVert-3\pi \Psi_2+2\pi \Phi_{11}\rVert^2.
\label{pageshoominvt}
\end{align}
In (\ref{pageshoominvt}), the Weyl scalar $\Psi_2$ and the Ricci-NP scalar $\Phi_{11}$ are given by closed-form expressions,equations (\ref{hweylpsi2}) and (\ref{riccivathm11}), respectively. Whereas the spin coefficients $\mu,\pi$ are given by the explicit algebraic expressions in eqns.(\ref{symframespincoef1}),(\ref{symframespin2}).
\end{theorem}
\begin{proof}
Using the expression for the covariant derivative, eqn.(\ref{covderivation}), the Bianchi identities, eqns.(\ref{BianchiEin})-(\ref{BianchiVier}) and our computation for the spin coefficients, eqns.(\ref{symframespincoef1})-(\ref{symframespin2}) we obtain eqn.(\ref{pageshoominvt}).
\end{proof}
\begin{corollary}\label{anixneftisorizontagegonotwn}
From Eqn.(\ref{symframespincoef1}), it is evident that $\mu=\varrho$ vanishes on the stationary horizons. This follows from the fact that the real roots of the radial polynomial $Q$ (eqn. \ref{radiieventpolyQ}), yield coordinate singularities which correspond to the up to four horizons of the spacetime. As a result the invariant $W$ must vanish there as well.
\end{corollary}

\section{Discussion and conclusions}

In this work we have derived new explicit algebraic expressions  for the Karlhede and Abdelqader-Lake differential curvature invariants for two of the most general black hole solutions. Namely, i)  for the Kerr-Newman-(anti-)de Sitter black hole metric ii) for accelerating Kerr-Newman black hole in (anti-)de Sitter spacetime.
Despite the complexity of the computations involved using the tensorial method of calculation, our final expressions are reasonably compact and easy to use in applications.   We showed explicitly that some of the computed invariants vanish at the horizon and ergosurfaces radii of the type of the black holes we investigated.
In particular, the differential invariant $Q_2$ vanishes at the horizons radii of the accelerating rotating and charged black holes with non-zero cosmological constant.
This result adds further impetus on the program of using scalar curvature invariants for the identification and detection of black hole horizons.
Moreover, we  proved that $Q_2$ for the Kerr-de Sitter black hole vanishes at the stationary horizons and is non-zero everywhere else, using Descarte's rule of signs and Bolzano's theorem.
We have also confirmed our results obtained via the tensorial method, with the aid of the NP formalism in which the differential curvature invariants are expressed in terms of covariant derivatives of the Weyl scalar $\Psi_2$. In particular, using the Bianchi identities in NP formalism, eqns.(\ref{BianchiEin})-(\ref{BianchiVier}) and a specific null-tetrad  eqns.(\ref{symment1})-(\ref{symmetnt2}), we derived an explicit expression for the Page-Shoom invariant $W$, eqn.(\ref{pageshoominvt}), Theorem \ref{cohompageshoom}, for an accelerating Kerr-Newman black hole in (anti-)de Sitter spacetime. We then proved that $W$ vanishes at the stationary horizons, Corollary \ref{anixneftisorizontagegonotwn}. The reason is that the spin coefficient $\mu$ we computed in eqn.(\ref{symframespincoef1}) vanishes at the stationary horizons.

As we proved in Theorem \ref{difanalQ1} and Corollary \ref{outerergosurfaceQ1det},  $Q_1$ is a suitable invariant to use for detecting the outer ergosurface of the Kerr black hole in the presence of the cosmological constant $\Lambda$.

Armed with our exact explicit algebraic expressions we analysed in detail the norms $I_5$ and $I_6$ associated with the gradients of the two non-differential Weyl invariants (the first two Weyl invariants in the ZM scheme) of the accelerating and non-accelerating Kerr-Newman black holes in (anti-)de Sitter spacetime. We showed that whereas both locally single out the horizons, their global behaviour is even more interesting. Both reflect the background angular momentum and electric charge as the volume of space allowing a timelike gradient decreases with increasing angular momentum and charge \footnote{They reflect to a lesser extent the cosmological constant and black hole's acceleration.} becoming zero for highly spinning and highly charged black holes. In the latter case these black holes do not admit $\mathcal{T}$ regions.

There are many important ramifications within both methods worth of further exploration  for general Riemannian and pseudo-Riemannian metrics. A particularly interesting aspect is the relation of cohomogeneity of a Riemannian manifold with the regular level sets of scalar Weyl invariants.
Indeed in \cite{ConsoleOlmos} it was proven that: The cohomogeneity of a Riemannian manifold $M$ (with respect to the full isometry group) coincides with the codimension of the foliation by regular level sets of the scalar Weyl invariants \footnote{Recalling that homogeneity is the dimension of the regular orbit in $M$ of the full isometry group of the metric $g$, an equivalent wording of Theorem 1 in \cite{ConsoleOlmos} is: The homogeneity of a Riemannian manifold $(M,g)$ is equal to dimension of a generic level set of the Weyl invariants.}. We note that the work in \cite{ConsoleOlmos} generalised earlier work by Singer in which the author characterised homogeneous spaces locally via the Riemann tensor $R^{\kappa}_{\;\;\lambda\mu\nu}$ and its covariant derivatives \cite{singer}.
From the Weyl theory of invariants \cite{hweyl}, scalar Weyl (or polynomial curvature) invariants are obtained from the covariant derivatives of the Riemann tensor by tensor products and complete contractions.
Using these, a direct bundle-theoretic method for defining and extending local isometries out of curvature data was developed in \cite{sergiocarlos}.
Therefore it will be very interesting to apply such bundle-theoretic methods for general pseudo-Riemannian manifolds, and in particular for the important case of accelerating,  rotating and charged black holes with $\Lambda\not =0$  studied in this paper in order, among other issues, to obtain a deeper understanding of the vanishing  of the invariant $Q_2$ at the horizons of the accelerating Kerr-Newman black hole  in (anti-)de Sitter spacetime. Such investigations are beyond the scope of this paper and it will be a subject of a future publication.  Such studies of the Pleba\'{n}ski-Demia\'{n}ski class of solutions of the Einstein-Maxwell system of differential equations  will also include  possible NUT  and magnetic charges.

An interesting application of our results will be to investigate binary black hole mergers using scalar curvature invariants. A recent study investigated a quasi-circular orbit of two merging, equal mass and non-spinning BHs \cite{JeremyAlanErik}.

Another fundamental research avenue of our results, would be to investigate gravitational lensing, black hole shadow and superradiance  effects for accelerating, rotating and charged black holes with $\Lambda\not=0$ \cite{GVKraniotisAccelBH}.
\section*{Acknowledgements} I thank Dr. D. Kaltsas and Phil Valder for useful discussions. I also thank N. Tritaki for inspiring discussion on black holes and the arrow of time.

\appendix
\section{R and T regions}\label{callosRT}

In this appendix we shall define the notions of $\mathcal{R}$ and $\mathcal{T}$ regions that appear in the analysis of $I_5,I_6$.
In \cite{kayllL} and following  \cite{katok} gradient flows for a particular invariant $\mathcal{I}$ were investigated:
\begin{equation}
k^{\alpha}=-g^{\alpha\beta}\frac{\partial \mathcal{I}}{\partial x^{\beta}}.
\end{equation}
From the theory of Lie derivatives \cite{kentayro}, for any scalar $\mathcal{I}$, the Lie derivative associated with a vector field $\xi$ is given by
\begin{equation}
\pounds_{\xi}\mathcal{I}=\xi_{\alpha}\nabla^{\alpha}\mathcal{I}=\xi_{\alpha}k^{\alpha}
\end{equation}
In certain circumstances, the scalar $\xi_{\alpha}k^{\alpha}$ can be directly related to the invariant $\mathcal{I}$ itself. For instance, if the manifold admits a homothetic motion $\pounds_{\xi}g_{\alpha\beta}=2\phi g_{\alpha\beta}$, where $\phi$ is a constant, it is known that for polynomial invariants \cite{pelavaslake}
\begin{equation}
\pounds_{\xi}\mathcal{I}=\kappa\phi\mathcal{I},
\end{equation}
where $\kappa$ is an integer characteristic of $\mathcal{I}$. In the case $\xi$ is a Killing vector (or trivial homothetic vector field), $\phi=0$ \cite{Mcintosh}. We then have:
\begin{equation}
\xi_{\alpha}k^{\alpha}=0.
\label{orthogonaltokilling}
\end{equation}
This means, that polynomial gradient flows are orthogonal to Killing flows, should they exist \footnote{In the special case, that $k$ itself satisfies $\pounds_{k}g_{\alpha\beta}=2\phi g_{\alpha\beta}$, it follows that $\nabla_{\mu}k_{\lambda}$ is symmetric in $\mu$ and $\lambda$: $\nabla_{\mu}k_{\lambda}=\phi g_{\mu\lambda}$ that is, the $k^{\mu}$ is a concurrent vector field \cite{kentayro}. Then the associated streamlines are geodesics with $\nabla_{\mu}(\frac{1}{2}k_{\lambda}k^{\lambda})=\phi k_{\mu}$. }.
This property was used by Lake to define $\mathcal{R}$ and $\mathcal{T}$ regions for gradient flows.   Since a stationary spacetime admits a timelike Killing congruence every nonzero 4-vector orthogonal to a timelike 4-vector must be spacelike. It follows from eqn.(\ref{orthogonaltokilling}) that any gradient flow is necessarily spacelike in a stationary region. Consequently, in \cite{kayllL} an $\mathcal{R}$ region is defined as a region with a positive norm of the gradrient vector field $k_{\alpha}$. Thus, $k_{\alpha}k^{\alpha}>0$ throughout an $\mathcal{R}$ region. A region in which the gradient flow is timelike and the norm $k_{\alpha}k^{\alpha}<0$ defines a $\mathcal{T}$ region \cite{kayllL}. Boundary regions are then naturally defined by $k_{\alpha}k^{\alpha}=0$.

\section{The norm of the covariant derivative of the Ricci tensor for accelerating Kerr-Newman black holes in (anti-)de Sitter spacetime}\label{Curbastro}

In this Appendix we calculate \textit{for the first time} an explicit algebraic expression for the curvature invariant:
\begin{equation}
R_{\alpha\beta;\mu}R^{\alpha\beta;\mu},
\label{RiccicalculCovDer}
\end{equation}
for accelerating Kerr-Newman black holes in (anti-)de Sitter spacetime.

\begin{theorem}\label{RiccisynalparagogosN}
We computed in closed form the curvature invariant constructed from the covariant derivative of the Ricci tensor in Eqn.(\ref{RiccicalculCovDer}), for accelerating Kerr-Newman black hole in (anti-)de Sitter spacetime with the result:
\begin{align}
&\nabla_{\mu}R_{\alpha\beta}\nabla^{\mu}R^{\alpha\beta}=-\frac{80 \left(\alpha  r \cos \! \left(\theta \right)-1\right)^{8} q^{4} }{\left(r^{2}+a^{2} \cos \! \left(\theta \right)^{2}\right)^{7}}\Biggl(a^{2} \Biggl[\frac{4 \left(\left(\alpha^{2}+\frac{\Lambda}{3}\right) a^{2}-\frac{\alpha^{2} q^{2}}{4}\right) \alpha^{2} r^{4}}{5}-2 a^{2} \alpha^{4} m \,r^{3}\nonumber \\
&+a^{2} \alpha^{2} \left(\left(\alpha^{2}+\frac{\Lambda}{3}\right) a^{2}+\alpha^{2} q^{2}-1\right) r^{2}+2 a^{2} \alpha^{2} m r +\frac{\Lambda  a^{4}}{3}\Biggr] \cos \! \left(\theta \right)^{6}-2 a^{2} \alpha  \Biggl(-\frac{\alpha^{2} m \,r^{4}}{5}+\Biggl[\left(-\frac{6 \alpha^{2}}{5}-\frac{2 \Lambda}{5}\right) a^{2}\nonumber \\
&-\frac{6 \alpha^{2} q^{2}}{5}\Biggr] r^{3}+a^{2} m \Biggr) \cos \! \left(\theta \right)^{5}+\Biggl(\left(\frac{4 \alpha^{2} \left(\alpha^{2}+\frac{\Lambda}{3}\right) a^{2}}{5}+\alpha^{4} q^{2}\right) r^{6}+\frac{2 a^{2} \alpha^{4} m \,r^{5}}{5}-\frac{26 a^{2} \alpha^{2} m \,r^{3}}{5}-\frac{\Lambda  a^{4} r^{2}}{15}\nonumber \\
&-a^{4} \left(\left(\alpha^{2}+\frac{\Lambda}{3}\right) a^{2}+\alpha^{2} q^{2}-1\right)\Biggr) \cos \! \left(\theta \right)^{4}+2 \Biggl\{-\alpha^{2} m \,r^{6}+\frac{6 \left(\alpha^{2}+\frac{\Lambda}{3}\right) a^{2} r^{5}}{5}-\frac{13 a^{2} \alpha^{2} m \,r^{4}}{5}+\frac{13 a^{2} m \,r^{2}}{5}\nonumber \\
&+\left(-\frac{6}{5} a^{4}-\frac{6}{5} a^{2} q^{2}\right) r +a^{4} m \Biggr\} \alpha  \cos \! \left(\theta \right)^{3}+\Biggl(-\alpha^{2} \left(\left(\alpha^{2}+\frac{\Lambda}{3}\right) a^{2}+\alpha^{2} q^{2}-1\right) r^{6}-\frac{\Lambda  a^{2} r^{4}}{15}+\frac{26 a^{2} \alpha^{2} m \,r^{3}}{5}\nonumber \\
&-\frac{2 a^{2} m r}{5}+\frac{a^{2} q^{2}}{5}-\frac{4 a^{4}}{5}\Biggr) \cos \! \left(\theta \right)^{2}-\frac{2 \alpha  r^{2} \left(-5 \alpha^{2} m \,r^{4}+a^{2} m +6 a^{2} r \right) \cos \left(\theta \right)}{5}+\Biggl[\frac{\Lambda  r^{4}}{3}-2 \alpha^{2} m \,r^{3}\nonumber \\
&+\left(\left(\alpha^{2}+\frac{\Lambda}{3}\right) a^{2}+\alpha^{2} q^{2}-1\right) r^{2}+2 m r -\frac{4 a^{2}}{5}-q^{2}\Biggr] r^{2}\Biggr).
\label{CovRicciInvaccelKNdS}
\end{align}
\end{theorem}

\begin{corollary}
For vanishing cosmological constant, $\Lambda=0$, the norm of the covariant derivative of the Ricci tensor for an accelerating Kerr-Newman black hole is given below:
\begin{align}
&\nabla_{\mu}R_{\alpha\beta}\nabla^{\mu}R^{\alpha\beta}=-\frac{80 \left(\alpha  r \cos \! \left(\theta \right)-1\right)^{8} q^{4} }{\left(r^{2}+a^{2} \cos \! \left(\theta \right)^{2}\right)^{7}}\Biggl(a^{2} \Biggl(\frac{4 \left(a^{2} \alpha^{2}-\frac{1}{4} \alpha^{2} q^{2}\right) \alpha^{2} r^{4}}{5}-2 a^{2} \alpha^{4} m \,r^{3}+a^{2} \alpha^{2} \Biggl[a^{2} \alpha^{2}+\alpha^{2} q^{2}\nonumber \\
&-1\Biggr]r^{2}+2 a^{2} \alpha^{2} m r \Biggr) \cos \! \left(\theta \right)^{6}-2 a^{2} \alpha  \left(-\frac{\alpha^{2} m \,r^{4}}{5}+\left(-\frac{6}{5} a^{2} \alpha^{2}-\frac{6}{5} \alpha^{2} q^{2}\right) r^{3}+a^{2} m \right) \cos \! \left(\theta \right)^{5}\nonumber \\
&+\left(\left(\frac{4}{5} a^{2} \alpha^{4}+\alpha^{4} q^{2}\right) r^{6}+\frac{2 a^{2} \alpha^{4} m \,r^{5}}{5}-\frac{26 a^{2} \alpha^{2} m \,r^{3}}{5}-a^{4} \left(a^{2} \alpha^{2}+\alpha^{2} q^{2}-1\right)\right) \cos \! \left(\theta \right)^{4}+2 \Biggl[-\alpha^{2} m \,r^{6}\nonumber \\
&+\frac{6 a^{2} \alpha^{2} r^{5}}{5}-\frac{13 a^{2} \alpha^{2} m \,r^{4}}{5}+\frac{13 a^{2} m \,r^{2}}{5}+\left(-\frac{6}{5} a^{4}-\frac{6}{5} a^{2} q^{2}\right) r +a^{4} m \Biggr] \alpha  \cos \! \left(\theta \right)^{3}+\Biggl\{-\alpha^{2} \left(a^{2} \alpha^{2}+\alpha^{2} q^{2}-1\right) r^{6}\nonumber \\
&+\frac{26 a^{2} \alpha^{2} m \,r^{3}}{5}-\frac{2 a^{2} m r}{5}+\frac{a^{2} q^{2}}{5}-\frac{4 a^{4}}{5}\Biggr\} \cos \! \left(\theta \right)^{2}-\frac{2 \alpha  r^{2} \left(-5 \alpha^{2} m \,r^{4}+a^{2} m +6 a^{2} r \right) \cos \left(\theta \right)}{5}\nonumber \\
&+\left(-2 \alpha^{2} m \,r^{3}+\left(a^{2} \alpha^{2}+\alpha^{2} q^{2}-1\right) r^{2}+2 m r -\frac{4 a^{2}}{5}-q^{2}\right) r^{2}\Biggr).
\end{align}

\end{corollary}

\begin{corollary}
For zero acceleration, $\alpha=0$, Eqn.(\ref{CovRicciInvaccelKNdS}), reduces to the following expression for the  differential Ricci curvature invariant for the KN(a)dS black hole:
\begin{align}
&\nabla_{\mu}R_{\alpha\beta}\nabla^{\mu}R^{\alpha\beta}=-\frac{80 q^{4}}{\left(r^{2}+a^{2} \cos \! \left(\theta \right)^{2}\right)^{7}} \Biggl(\frac{a^{6} \Lambda  \cos \left(\theta \right)^{6}}{3}+\left(-\frac{\Lambda  a^{4} r^{2}}{15}-a^{4} \left(\frac{a^{2} \Lambda}{3}-1\right)\right) \cos \! \left(\theta \right)^{4}\nonumber \\
&+\left(-\frac{1}{15} \Lambda  a^{2} r^{4}-\frac{2}{5} a^{2} m r +\frac{1}{5} a^{2} q^{2}-\frac{4}{5} a^{4}\right) \cos \! \left(\theta \right)^{2}+\left(\frac{\Lambda  r^{4}}{3}+\left(\frac{a^{2} \Lambda}{3}-1\right) r^{2}+2 m r -\frac{4 a^{2}}{5}-q^{2}\right) r^{2}\Biggr).
\end{align}
\end{corollary}

\begin{corollary}
The curvature invariant $R_{\alpha\beta;\mu}R^{\alpha\beta;\mu}$  for the Reissner-Nordstr\"{o}m-(anti-)de Sitter black hole, is obtained by setting $a=\alpha=0$ in Eqn.(\ref{CovRicciInvaccelKNdS}) with the result:
\begin{align}
&\nabla_{\mu}R_{\alpha\beta}\nabla^{\mu}R^{\alpha\beta}=-\frac{80 q^{4} \left(\frac{1}{3} \Lambda  r^{4}-r^{2}+2 m r -q^{2}\right)}{r^{12}}.
\end{align}
\end{corollary}

\begin{corollary}\label{LouisW1}
Setting $\alpha=\Lambda=0$ in Eqn.(\ref{CovRicciInvaccelKNdS}), we obtain the norm of the covariant derivative of the Ricci tensor for a Kerr-Newman black hole:
\begin{align}
&\nabla_{\mu}R_{\alpha\beta}\nabla^{\mu}R^{\alpha\beta}=\frac{80 \left(-\cos \! \left(\theta \right)^{4} a^{4}+\left(\frac{2}{5} a^{2} m r +\frac{4}{5} a^{4}-\frac{1}{5} a^{2} q^{2}\right) \cos \! \left(\theta \right)^{2}-\left(-r^{2}+2 m r -\frac{4}{5} a^{2}-q^{2}\right) r^{2}\right) q^{4}}{\left(r^{2}+a^{2} \cos \! \left(\theta \right)^{2}\right)^{7}}.
\end{align}

\end{corollary}

\begin{corollary}\label{Lwitten2}
Setting $\alpha=a=\Lambda=0$  in Eqn.(\ref{CovRicciInvaccelKNdS}), we obtain the local curvature invariant built from the covariant derivative of the Ricci tensor for a Reissner-Nordstr\"{o}m black hole:
\begin{align}
&\nabla_{\mu}R_{\alpha\beta}\nabla^{\mu}R^{\alpha\beta}=-\frac{80 q^{4} \left(2 m r -q^{2}-r^{2}\right)}{r^{12}}.
\end{align}
\end{corollary}

\begin{remark}
Our results in Corollaries \ref{LouisW1},\ref{Lwitten2} agree with the corresponding expressions obtained in \cite{gass} for the Kerr-Newman and Reissner-Nordstr\"{o}m black holes.
\end{remark}

 \end{document}